\documentclass{acmsmall}

\usepackage[shortlabels]{enumitem}
 
\usepackage{latexsym,xspace,calc,amsthm}
\usepackage{amsmath,amssymb} %
\usepackage{centernot}

\usepackage{lmodern}
\usepackage{tgtermes}
\usepackage{fix-cm}

\usepackage{cancel}

\usepackage{savesym}
\savesymbol{subcaption}
\usepackage{subcaption}

\usepackage{graphbox} %

\makeatletter
\def\@biblabel#1{[#1]} %
\def\thebibliography#1{%
    \footnotesize
    \refsection*{{\refname}
        \@mkboth{\uppercase{\refname}}{\uppercase{\refname}}%
    }
    \list{\@biblabel{\@arabic\c@enumiv}}%
       {\settowidth\labelwidth{\@biblabel{#1}}%
        \leftmargin\labelwidth
        \advance\leftmargin\bibindent
        \itemindent-\bibindent
        \itemsep2pt
        \parsep \z@
        \usecounter{enumiv}%
        \let\p@enumiv\@empty
        \renewcommand\theenumiv{\@arabic\c@enumiv}%
    }%
    \let\newblock\@empty
    \sloppy
    \sfcode`\.=1000\relax
}
\makeatother

 \renewenvironment{thebibliography}[1]{%
   \begin{odlthebibliography}{#1}%
     \setlength{\parskip}{0ex}%
     \setlength{\itemsep}{3pt}%
     \fontsize{9.5}{9.5} %
     \selectfont
}%
 {%
   \end{odlthebibliography}%
 }

\usepackage{tikz}
\usepackage{tikz-cd}
\usetikzlibrary{cd,decorations.markings}
\tikzset{
    dharrow/.style={
        <->,
        postaction={decorate,-},
        }
}
\tikzset{
    dhdashedarrow/.style={
        <->,
        dashed,
        postaction={decorate,-},
        }
    }
\tikzset{
    lrharpoonarrow/.style={
        <[harpoon]->[harpoon],
        postaction={decorate,-},
        }
}
\tikzset{
    lrharpoondashedarrow/.style={
        <[harpoon]->[harpoon],
        dashed, %
        postaction={decorate,-},
        }
}
\usetikzlibrary{arrows}
\usetikzlibrary {arrows.meta} 
\usetikzlibrary{calc}
\usetikzlibrary{positioning}
\usetikzlibrary{snakes,automata,chains}
\usetikzlibrary{graphs}

\usepackage{binarytree}

\usepackage{amssymb,stmaryrd,amsmath}
\usepackage{mdwlist} %
\usepackage{float}   %
\usepackage{centernot} %

\usepackage[colorlinks]{hyperref}
\usepackage{breakurl}  %

\newtheoremstyle{jamiestyle}%
  {4pt}%
  {0pt}%
  {\it}%
  {0pt}%
  {\sc}%
  {.}%
  { }%
  {}%
\theoremstyle{jamiestyle}
\newtheorem{thrm}{Theorem}[subsection]
\newtheorem{prop}[thrm]{Proposition}
\newtheorem{lemm}[thrm]{Lemma}
\newtheorem{corr}[thrm]{Corollary}

\newtheoremstyle{jamienfstyle}%
  {4pt}%
  {0pt}%
  {\normalfont}%
  {0pt}%
  {\sc}%
  {.}%
  { }%
  {}%
\theoremstyle{jamienfstyle}
\newtheorem{nttn}[thrm]{Notation}
\newtheorem{defn}[thrm]{Definition}
\newtheorem{xmpl}[thrm]{Example}
\newtheorem{rmrk}[thrm]{Remark}

\usepackage{color}
\definecolor{mygreen}{rgb}{0,0.6,0}
\definecolor{mygray}{rgb}{0.5,0.5,0.5}
\definecolor{mymauve}{rgb}{0.58,0,0.82}

\usepackage{listings}
 
\definecolor{gray}{RGB}{128, 128, 128}
\definecolor{lightgray}{RGB}{200, 200, 200}
\definecolor{cyan}{RGB}{0, 255, 255}
\definecolor{blue}{RGB}{0, 0, 255}
\definecolor{red}{RGB}{255, 0, 0}
\definecolor{pink}{RGB}{255, 128, 128}
\definecolor{green}{RGB}{0, 128, 0}
\definecolor{lightyellow}{RGB}{255, 255, 200}
\definecolor{purple}{RGB}{128, 0, 128}

\lstdefinestyle{all}
    {basicstyle=\ttfamily\scriptsize,
     keywordstyle=\color{blue}\ttfamily\scriptsize,
     commentstyle=\color{green}\ttfamily\scriptsize,
     stringstyle=\color{red}\ttfamily\scriptsize}

\lstdefinelanguage{hask}{%
    frame=none,
    xleftmargin=2pt,
    belowcaptionskip=\bigskipamount,
    captionpos=b,
    tabsize=2,
    numbers=left,
    numberstyle=\tiny\color{gray},
    emphstyle={\bf},
	morecomment=[s][\color{green}]{\{-}{-\}},
    stringstyle=\mdseries\rmfamily,
    commentstyle=\color{green},
    keywords={},
    keywords=[1]{case, of, data, if, then, else, where, let, in, do},
    keywords=[2]{Chip, Config, CurrencySymbol, TokenName, PubKeyHash, Integer, Value, State, Action, Text, Maybe, Void, TxConstraints,  Contract},
    keywords=[3]{HasNative},
    keywords=[4]{=>},
    keywords=[5]{Just, Nothing, MkChip, MkConfig, SetPrice, Buy},
    keywordstyle=[1]\mdseries\sffamily\color{red},
    keywordstyle=[2]\mdseries\sffamily\color{blue},
    keywordstyle=[3]\mdseries\sffamily\color{green},
    keywordstyle=[4]\mdseries\sffamily,
    keywordstyle=[5]\mdseries\sffamily\color{purple},
    columns=flexible,
    basicstyle=\small\sffamily,
    showstringspaces=false,
    breaklines=false,
    showspaces=false,
    escapeinside={--}{\^^M},escapebegin={\color{green}--},escapeend={},
    literate= {+}{{$+$}}1 {/}{{$/$}}1 {*}{{$*$}}1 {=}{{$=$}}1
              {>}{{$>$}}1 {<}{{$<$}}1 {\\}{{$\lambda$}}1
              {\\\\}{{\char`\\\char`\\}}1
              {->}{{$\rightarrow$}}2 {>=}{{$\geq$}}2 {<-}{{$\leftarrow$}}2
              {<=}{{$\leq$}}2 {=>}{{$\Rightarrow$}}2
              {\ .}{{$\circ$}}2 {\ .\ }{{$\circ$}}2
              {>>}{{>>}}2 {>>=}{{>>=}}2
              {|}{{$\mid$}}1
              {\_}{{\underline{\hspace{2mm}}}}2
}

\lstdefinelanguage{solidity}{%
    frame=none,
    xleftmargin=2pt,
    belowcaptionskip=\bigskipamount,
    captionpos=b,
    tabsize=2,
    numbers=left,
    numberstyle=\tiny\color{gray},
    emphstyle={\bf},
	morecomment=[s][\color{green}]{\{-}{-\}},
    stringstyle=\mdseries\rmfamily,
    commentstyle=\color{green},
    keywords={},
    keywords=[1]{pragma, solidity, contract, event, constructor, require, function, return, emit},
    keywords=[2]{address, uint, mapping},
    keywords=[3]{public, payable, external, view, returns},
    keywords=[4]{=>, +=, -=, =, <=, ==},
    keywords=[5]{msg, sender, transfer, value},
    keywordstyle=[1]\mdseries\sffamily\color{red},
    keywordstyle=[2]\mdseries\sffamily\color{blue},
    keywordstyle=[3]\mdseries\sffamily\color{green},
    keywordstyle=[4]\mdseries\sffamily,
    keywordstyle=[5]\mdseries\sffamily\color{purple},
    columns=flexible,
    basicstyle=\small\sffamily,
    showstringspaces=false,
    breaklines=false,
    showspaces=false,
    escapeinside={--}{\^^M},escapebegin={\color{green}--},escapeend={},
    literate= {+}{{$+$}}1 {/}{{$/$}}1 {*}{{$*$}}1 {=}{{$=$}}1
              {>}{{$>$}}1 {<}{{$<$}}1 {\\}{{$\lambda$}}1
              {\\\\}{{\char`\\\char`\\}}1
              {->}{{$\rightarrow$}}2 {>=}{{$\geq$}}2 {<-}{{$\leftarrow$}}2
              {<=}{{$\leq$}}2 {=>}{{$\Rightarrow$}}2
              {\ .}{{$\circ$}}2 {\ .\ }{{$\circ$}}2
              {>>}{{>>}}2 {>>=}{{>>=}}2
              {|}{{$\mid$}}1
              {\_}{{\underline{\hspace{2mm}}}}2
}

\makeatletter
\newcommand\hpn[2][]{%
  \ext@arrow 9999{\hpnfill@}{#1}{#2}}
\newcommand\hpnfill@{%
  \arrowfill@\leftharpoonup\relbar\rightharpoondown}
\makeatother

\usepackage{tgadventor}
\DeclareMathAlphabet{\mathgeo}{T1}{qag}{m}{n}
\newcommand{\QQscaled}[2]{%
  \scalebox{0.65}{%
    \pdfliteral direct{2 Tr 0.25 w 0 G}%
    $#1\mathgeo{Q}$%
    \pdfliteral direct{0 Tr}%
  }%
}
\newcommand{\QQ}{\mathord{\mathpalette\QQscaled{}}}

\newcommand\Aring{\tall^{\hspace{-0.3pt}\QQ}} %
\newcommand\Ering{\texi^{\hspace{-0.6pt}\QQ}} %

\newcommand\everyone{\modality{\Box}}
\newcommand\somewhere{\modality{\Diamond}}
\newcommand\someOsomeP{\Ering_{\hspace{0pt}\somewhere}}
\newcommand\allOallP{\Aring_{\hspace{0pt}\everyone}}

\newcommand\Quorum{\modality{\Ering}}
\newcommand\someOallP{\modality{\Ering_{\Box}}} %
\newcommand\Contraquorum{\modality{\Aring}}
\newcommand\allOsomeP{\modality{\Aring_{\hspace{-1pt}\Diamond}}} %

\newcommand\mruarrow{\mathbin{\rlap{\scalebox{0.73}{$\recentsimple$}}\hspace{-2pt}\rightarrow}}
\newcommand{\mru}[2]{#1\mruarrow#2}
\newcommand{\mrup}[3]{#2[#1\mruarrow #3]}

\usepackage{tikz}
\newcommand*{\myhash}{%
  \begin{tikzpicture}
    \pgfmathsetlengthmacro\myWidth{.8*width("=")}%
    \pgfmathsetlengthmacro\myHeight{height("H")}%
    \pgfmathsetlengthmacro\mySepY{.3333*\myWidth}%
    \pgfmathsetlengthmacro\mySideBearing{.1*\myWidth}%
    \def\myAngle{70}%
    \pgfmathsetlengthmacro\mySepX{\mySepY/sin(\myAngle)}%
    \pgfmathsetlengthmacro\mySlantX{\myHeight/tan(\myAngle)}%
    \draw[line cap=round]
      (0, {(\myHeight - \mySepY)/2}) -- ++(\myWidth, 0)
      (0, {(\myHeight + \mySepY)/2}) -- ++(\myWidth, 0)
      ({(\myWidth - \mySepX - \mySlantX)/2}, 0)
      -- ({(\myWidth - \mySepX + \mySlantX)/2}, \myHeight)
      ({(\myWidth + \mySepX - \mySlantX)/2}, 0)
      -- ({(\myWidth + \mySepX + \mySlantX)/2}, \myHeight)
    ;%
    \useasboundingbox
      (-\mySideBearing, 0)
      (\myWidth + \mySideBearing, \myHeight)
    ;%
  \end{tikzpicture}%
}

\newcommand\THREE{{\mathbf 3}}
\newcommand\threeValid{\THREE_{\scalebox{0.6}{valid}}} %
\newcommand\threeCorrect{\THREE_{\scalebox{0.6}{correct}}}
\newcommand\threeTrue{\THREE_{\scalebox{0.6}{true}}}
\newcommand\threeFalse{\THREE_{\scalebox{0.6}{false}}}

\newcommand\Pnt{\ns{Point}}
\newcommand\Val{\ns{Value}}
\newcommand\Valuation{\ns{Valuation}}
\newcommand\PredSymb{\ns{PredSymb}}

\newcommand\tvT{{\mathbf t}}
\newcommand\tvF{{\mathbf f}}
\newcommand\tvB{{\mathbf b}}

\newcommand\recentsimple{{\reflectbox{$\mathsf R$}}} 
\newcommand\recent{\recentsimple} %
\newcommand\urecent{{\ubar\recentsimple}} %

\newcommand\figunderline[1]{\text{\rlap{\underline{\emph{#1}}}}}

\usepackage{stackengine}
\newcommand\ubar[1]{\stackunder[0.6pt]{$#1$}{\rule{1ex}{.18ex}\hspace{0.5pt}}}

\newcommand\figskip{\\[1ex]}

\newcommand\thea{a}
\newcommand\theb{b}
\newcommand\thec{c}

\newcommand\modality[1]{#1\hspace{0pt}}
\newcommand\modT{\modality{\tf T}}
\newcommand\modTB{\modality{\tf T\hspace{-2.5pt}\tf B}}
\newcommand\modFB{\modality{\tf F\hspace{-1.5pt}\tf B}}
\newcommand\modB{\modality{\tf B}}
\newcommand\modTF{\modality{\tf T\hspace{-2.5pt}\tf F}}
\newcommand\modF{\modality{\tf F}}

\newcommand\acontext{\f{ctx}}
\newcommand\avaluation{\varsigma} %

\newcommand\rulefont[1]{\ensuremath{{\mathsf{(#1)}}}}

\newcommand\opens{\ns{Open}}
\newcommand\neqempty{\pdfliteral direct{2 Tr 0.3 w 0 G}${\neq}\varnothing$\pdfliteral direct{0 Tr}}
\newcommand\opensne{\opens_{\hspace{-.5pt}\scalebox{0.6}{\neqempty}}} %

\newcommand{\dotarrow}{%
   \mathrel{\ooalign{\hss\raise.85ex\hbox{\scalebox{1.25}{\normalfont .}}%
   \kern0.35ex\hss\cr$\rightarrow$}}}

\newcommand\mycard{\myhash}
\newcommand\narrowmath[1]{\scalebox{0.8}[1.0]{$\mathsf{#1}$}}
\newcommand\sometime{\modality{\narrowmath{sometime}}}
\newcommand\forever{\modality{\narrowmath{forever}}}
\newcommand\tomorrow{\modality{\narrowmath{tomorrow}}}
\newcommand\yesterday{\modality{\narrowmath{yesterday}}}
\newcommand\infinitely{\modality{\narrowmath{infinitely}}}
\newcommand\final{\modality{\narrowmath{ultimately}}}
\newcommand\arity{\f{arity}}
\newcommand\correct[1]{\tf{correct}[#1]}

\newcommand\pointwise[1]{\tf{pointwise}[#1]}

\newcommand\udfn{\tf{undef}}
\newcommand\gslt{\f{GSLT}}
\newcommand\dense{\f{dense}}
\newcommand\noi{\f{noi}}

\newcommand\intersectswith{\between}

\newcommand\xxtwined{$2$-twined\xspace}
\newcommand\xxxtwined{$3$-twined\xspace}

\newcommand\deffont[1]{{\bfseries #1}}
\newcommand\powerset{\f{pow}}

\newcommand\f[1]{\mbox{$\mathit{#1}$}}
\newcommand\tf[1]{\mathsf{#1}}
\DeclareRobustCommand{\ns}[1]{%
  \ifmmode
    \mathord{\text{\usefont{OT1}{cmss}{sbc}{n}#1}}%
  \else
    {\usefont{OT1}{cmss}{sbc}{n}#1}%
  \fi
}
\usepackage{slantsc}
\newcommand\theory[1]{\ensuremath{\text{\itshape\scshape #1}}\xspace}
\newcommand\ThyPax{\theory{TyPax}}
\newcommand\ThySPax{\theory{TySPax}}

\newcommand\liff{\Longleftrightarrow}
\newcommand\limp{\Longrightarrow}

\newcommand\ssm{{{:}\text{=}}}
\DeclareMathSymbol{\shortminus}{\mathbin}{AMSa}{"39}
\newcommand\minus{{\shortminus}}
\newcommand\plus{{+}}
\newcommand\Forall[1]{\forall #1.}
\newcommand\Exists[1]{\exists #1.}

\newcommand\cent{\vdash}

\newcommand\ment{\vDash}
\newcommand\nment{\not\vDash}
\newcommand\mentval{\ment_{\avaluation}}
\newcommand\nmentval{\nment_{\avaluation}}

\newcommand\lmodel{[\hspace{-0.2em}[}
\newcommand\rmodel{]\hspace{-0.2em}]}
\newcommand\model[1]{{\lmodel #1 \rmodel}}
\newcommand\modellabel[2]{{\lmodel #1 \rmodel}_{#2}}

\newcommand\Ngeqz{\mathbb N_{\scriptscriptstyle\geq 0}}
\newcommand\Time{\Ngeqz}
\newcommand\weakmodusponens{weak modus ponens (Proposition~\ref{prop.mp.for.tnotor}(\ref{item.mp.for.tnotor}))\xspace}
\newcommand\strongmodusponens{strong modus ponens (Proposition~\ref{prop.mp.for.tnotor}(\ref{item.mp.for.timpc}))\xspace}

\makeatletter
\DeclareRobustCommand{\barcent}{\mathbin{\mathpalette\barcent@@\relax}}
\newcommand{\barcent@@}[2]{%
  \vbox{\offinterlineskip
    \sbox\z@{$\m@th#1\cent$}%
    \ialign{%
      \hfil##\hfil\cr
      $\m@th#1{}_{\minus}\kern-\scriptspace$\cr
      \noalign{\kern-.3\ht\z@}
      \box\z@\cr
    }%
  }%
}
\makeatother

\makeatletter
\def\pmb@#1#2{\setbox8\hbox{$\m@th#1{#2}$}%
  \setboxz@h{$\m@th#1\mkern-.1mu$}\pmbraise@\wdz@
  \binrel@{#2}%
  \dimen@-\wd8 %
  \binrel@@{%
    \mkern-.1mu\copy8 %
    \kern\dimen@\mkern-.2mu\copy8 %
    \kern\dimen@\mkern-.3mu\copy8 %
    \kern\dimen@\mkern-.4mu\copy8 %
    \kern\dimen@\mkern.1mu\copy8 %
    \kern\dimen@\mkern.2mu\copy8 %
    \kern\dimen@\mkern.3mu\copy8 %
    \kern\dimen@\mkern.0mu\raise\pmbraise@\copy8 %
    \kern\dimen@\mkern.4mu\box8 %
           }%
}
\makeatother

\newcommand\compressthislight[1]{{\hspace{.8pt}\raisebox{.5pt}{\scalebox{.85}{$#1$}}\hspace{.2pt}}}
\newcommand\compressthis[1]{\pmb{\compressthislight{#1}}}

\newcommand\tneg{{\pmb\neg}}

\newcommand\ttop{{\pmb\top}}
\newcommand\tbot{{\pmb\bot}}
\newcommand\teq{{\pmb{\text{=}}}}
\newcommand\tequiv{\mathbin{\compressthislight{\equiv}}}
\newcommand\tneq{{\pmb{\neq}}}
\newcommand\tand{\mathrel{\pmb\wedge}}
\newcommand\tor{\mathbin{\pmb\vee}}

\newcommand\tleq{\compressthis{\leq}}

\newcommand\timp{{\compressthis{\Rightarrow}}}

\newcommand\tnotor{\mathrel{\to}}

\newcommand\tlatticeiff{\mathbin{\rlap{$\leftarrow$}\hspace{-0.2ex}\rightarrow}}

\newcommand\timpc{\mathrel{\rightharpoonup}}

\newcommand\tiffcc{\mathbin{\rlap{$\leftharpoonup$}\hspace{-0.2ex}\rightharpoonup}}

\newcommand\tall{{\compressthis{\forall}}}
\newcommand\texi{{\compressthis{\exists}}}
\newcommand\texiunique{{\compressthis{\exists_{1}}}}
\newcommand\texiaffine{{\compressthis{\exists_{01}}}}

\newcommand\AllBut[2]{\tf{AllBut}_{#2}(#1)}
\newcommand\Nset[1]{\{0\dots #1\minus 1\}}

\newcommand\QLogic{Three-valued Fixpoint Quorum Logic\xspace}
\usepackage{upgreek}
\newcommand\QLogicShort{$\THREE\hspace{-.5pt}\upmu\mathsf{QL}$\xspace}

\makeatletter
\newcommand{\circlearrow}{}%
\DeclareRobustCommand{\circlearrow}{%
  \mathrel{\vphantom{\shortrightarrow}\mathpalette\circle@arrow\relax}%
}
\newcommand{\circle@arrow}[2]{%
  \m@th
  \ooalign{%
    \hidewidth$#1\circ\mkern1mu$\hidewidth\cr
    $#1\longrightarrow$\cr}%
}
\makeatother

\makeatletter
\newcommand*\bigcdot{\mathpalette\bigcdot@{.5}}
\newcommand*\bigcdot@[2]{\mathbin{\vcenter{\hbox{\scalebox{#2}{$\m@th#1\bullet$}}}}}
\makeatother

\usepackage{datetime}
\yyyymmdddate

\begin{document}
\title{A declarative approach to specifying Paxos using three-valued modal fixpoint logic}
\newcommand\titlerunning{\emph{Paxos via logic}} %
\newcommand\authorrunning{\emph{Murdoch J. Gabbay and Luca Zanolini}}
\author{Murdoch J. Gabbay \affil{Heriot-Watt University, Edinburgh, UK}
Luca Zanolini \affil{Ethereum Foundation, London, UK}}

\begin{abstract}
We present \emph{\QLogic}, or \QLogicShort for short; a three-valued modal fixed-point logic designed for declaratively specifying and reasoning about distributed algorithms, such as the Paxos consensus algorithm.

Our methodology represents a distributed algorithm as a logical theory, enabling correctness properties to be derived directly within the framework --- or revealing logical errors in the algorithm's design when they exist.
  
\QLogicShort adopts a declarative approach, specifying the overall logic of computation without prescribing control flow. Notably, message-passing and state-transitions are not explicitly modelled, and notions of `quorum' are abstracted from concrete thresholds to abstract modal operators.
This distinguishes our framework from approaches like TLA+. 
The logical abstraction emphasises the essence of distributed algorithms, offering a novel perspective on their specification and reasoning.
  
We define the syntax and semantics of \QLogicShort, explore its theoretical properties, and demonstrate its applicability through a detailed treatment of the Paxos consensus algorithm. By presenting Paxos as a logical theory and deriving its standard correctness properties, we showcase the framework’s capacity to handle non-trivial distributed systems.
  
We envision \QLogicShort as a versatile tool for specifying and reasoning about distributed algorithms. The Paxos example highlights the framework’s ability to capture intricate details, offering a new lens through which distributed algorithms can be specified, studied, and checked.

The ideas we present are easily implemented, and to demonstrate this we have implemented a freely-available proof-of-concept \QLogicShort model-checker, as an effective prototyping tool.

\end{abstract}

\maketitle
\thispagestyle{empty}

\tableofcontents

\section{Introduction}
\label{sect.intro}

\subsection{Initial overview}

\subsubsection*{Introducing \QLogicShort}

\emph{\QLogic}, or \emph{\QLogicShort} for short, is a logical framework for studying distributed algorithms in a \emph{declarative} manner, abstracting away in particular from notions of message-passing, concrete state-transitions, and explicit threshold-based notions of quorum.

To demonstrate the scope and utility of this approach, we focus on Paxos~\cite{DBLP:journals/tocs/Lamport98} --- one of the most widely deployed and influential consensus protocols in distributed systems~\cite{Chandra:2007:PaxosMadeLive,bolosky2011paxos,rao2011spinnaker,Baker:2011:Megastore}. 
By developing an axiomatisation of Paxos in our framework, we capture the essence of its design in a precise logical form and establish a principled foundation for reasoning about its correctness. 
Furthermore, this axiomatic perspective paves the way for extending the methodology to more complex algorithms, offering a path toward analyses that are both conceptually cleaner, and easier to manage.

Logic is widely applied in distributed algorithms: for example the logic TLA+ (based on Linear Temporal Logic) is an industry standard to model and verify distributed protocols~\cite{DBLP:journals/toplas/Lamport94,DBLP:books/aw/Lamport2002,DBLP:conf/charme/YuML99}. 
TLA+ encourages us to think in terms of a (sequential) temporal structure of actions (as per the name: Temporal Logic of Actions). 
In contrast, a \QLogicShort axiomatisation adopts a more abstract perspective: it distils the protocol into a set of logical principles that capture its essential structure. 
This abstraction highlights the core invariants of the algorithm and also provides a clean alternative foundation for reasoning and generalisation.

This idea of abstracting an algorithm to a logical theory is itself not new: \emph{declarative programming} is based on it.
What is new, however, is applying an axiomatic, message-free, time-free, three-valued modal logic perspective to distributed algorithms at the level of Paxos.

This paper demonstrates that such an axiomatisation is not only possible but also illuminating: it exposes the essential logical structure of the protocol in a way that is both rigorous and conceptually simple.\footnote{Note that just because a framework is conceptually simple does not mean that everything we say in that framework must also be simple.  \QLogicShort is a core language (like first-order logic, higher-order logic, the $\lambda$-calculus, LISP or indeed C); the programs and proofs that we express within it can be arbitrarily rich.} 
Concrete examples appear below and throughout the paper.\footnote{Readers eager to see a toy illustration of the method may skip directly to Subsection~\ref{subsect.simple}.}

For reasons of length, we concentrate on Paxos in this paper. 
This offers a rich and nontrivial case study --- complex and non-trivial, yet sufficiently well-understood to serve as a benchmark.
In a sister paper to this one with a more explicit expository emphasis~\cite{gabbay:decdaa}, the first author applies a simpler logic to study Bracha Broadcast and Crusader Agreement (which are simpler than Paxos).

To demonstrate the applicability of \QLogicShort, we created a proof-of-concept \QLogicShort model-checker~\cite{Gabbay2026-3muQL-explorer}, vibe-coded by the first author using AI from prompts starting with the draft of this paper, and with~\cite{gabbay:decdaa}.
This is a theory paper but the existence of a parallel prototype implementation indicates that this theory\footnote{\dots which is no more complex than that justifying any other declarative environment, like Haskell or Lean.} has potential for the delivery of compact, snappy, and effective developer tools, with succinct syntax and surprisingly compact search-spaces.
The reader is welcome to play with the tool and watch the logic from this paper come alive in an implementation.\footnote{Note that the implementation is \emph{not this paper}!  For example, by design the implementation omits the exposition and proofs that give meaning to and justify correctness of its computations --- that is the task of a maths paper like this one --- and as a prototype it should not be relied on in production.  Yet it illustrates and indicates the potential for what could become of the ideas in this paper in the hands of a competent software engineer.  So this paper has abundant theory, but it is not \emph{just} theory for its own sake.  Quite the opposite.}

The \QLogicShort model-checker is a prototype but it is not a toy. 
We know this, because the AI coding agent spontaneously applied the model-checker it had just written to test the papers from its prompt, and in doing so it falsified a conjecture from~\cite{gabbay:decdaa} (as documented in the package's README file).
So: not a toy.

In the long run, we envision our logical approach fulfilling the broader promise of declarative techniques: uncovering subtle errors, reducing development effort, accelerating productivity, and even suggesting new mathematical insights.
 
\subsubsection*{What is declarative programming}

By `declarative' we mean that we specify the logic of a computation without specifying its control flow.

Thus our axiomatisation of Paxos (in Figure~\ref{fig.logical.paxos}) relates to a concrete Paxos implementation and/or its TLA+ specification, much as this axiomatic, algebraic specification of the factorial function\footnote{By `axiomatic' we mean that the function is specified just by asserting properties that are true of it.  By `algebraic' we mean that these assertions are purely equational, i.e. we just assert some equalities between terms.}
$$
0!=1
\qquad 
(n+1)!=(n+1)*n!
$$
relates to this concrete implementation of factorial in BASH
\begin{quote}
{\tt\small\hspace{-2em} value=1; while [ \$number -gt 1 ]; do value=\$((value * number)); number=\$((number - 1)); done}
\end{quote}
The reader may be familiar with programming languages based on declarative principles, such as Haskell, Erlang, or Prolog.
Their design philosophy is that we should assert things that we want to be true, and leave it to the compiler to figure out how best to make this happen.
Contrast with C, which is a canonical imperative language.\footnote{We are not interested here in taking sides.  With a variety of tools to choose from, we pick the one that is best for the job.  What matters is that the widest possible \emph{variety} of levels of abstraction is made available.}

This paper develops a methodology for specifying distributed algorithms which lies firmly on the Haskell/Erlang/Prolog side of the spectrum --- concise, declarative, and abstraction-oriented --- in contrast to TLA+, which (while still high-level) explicitly represents state transitions, quorum sizes, and other operational detail.

As a warm-up, we will consider a minimal example: a simple voting system. 
We present this in \QLogicShort to illustrate the style of axiomatisation and to give the reader a concrete sense of how the framework operates in practice.

Consider a voting system such that 
\begin{quote}
`if a quorum of voters vote for candidate $c$, then $c$ wins the election'.
\end{quote}
We write this axiomatically simply as
$$
(\someOallP\tf{vote}(c)) \timpc \tf{wins}(c) .
$$ 
The $\someOallP$ is a logical modality that is true when there exists a quorum of participants at which $\tf{vote}(c)$ holds.\footnote{Details of this modality in our framework are in Proposition~\ref{prop.pointwise.dense.char}(\ref{item.dense.char.pointwise.QuorumBox}) and Lemma~\ref{lemm.allbut.concretely}.}

Even without going into the details of what these symbols mean, it is clear that details of message-passing, of how votes are communicated, and of how they are counted up, are absent.
In fact in this case the predicate is just a direct rendering of the English statement above it, without specifying how it is concretely implemented.
Thus, in keeping with a declarative philosophy, we specify what we want the algorithm to do, but not how it should do it: it is the point of the axiom to express formally and precisely our idea --- and to add as little as possible to that, beyond what is necessary to make it logically precise.

\subsubsection*{Who should read this document}

This work is intended for two main audiences:
\begin{enumerate*}
\item
\emph{logicians}, who may be intrigued by the study of a three-valued modal fixed-point logic that combines unusual theoretical features with concrete applications to distributed computation, and
\item
\emph{distributed-systems researchers}, who may find \QLogicShort a basis for useful tools for specifying, analysing, and communicating the design of complex protocols.
\end{enumerate*}

Our primary case studies concern consensus algorithms, with a particular emphasis on blockchain protocols.\footnote{This is a particular professional interest of the second author, who was part of the core team developing consensus algorithms for Ethereum.}
These represent one of the most demanding and safety-critical application domains.

\subsubsection*{References to key points in the text}

The reader who might want to jump to where the action is and read backwards and forwards from there, is welcome to proceed as follows:
\begin{itemize*}
\item
The axioms of Declarative Paxos are in Figure~\ref{fig.logical.paxos}.
\item
The syntax and semantics of \QLogicShort are in Figures~\ref{fig.predicate.syntax}, \ref{fig.3.phi.f}, and~\ref{fig.3.derived}.
\item
Truth-tables for three-valued logic and its connectives are in Figure~\ref{fig.3}.
\item
An outline of the Paxos algorithm (for the reader who might like to be reminded what it is) is in Remark~\ref{rmrk.high-level.paxos}.
\item
See also toy axiomatisations for an imaginary `simple' protocol in Subsection~\ref{subsect.simple}, designed to be vaguely reminiscent of Paxos, but very greatly simplified for the sake of exposition.
\end{itemize*}

\subsection{Motivation and background (further details)}

Let us unpack what consensus in distributed systems (like blockchains) is, why it matters, and why it is hard:

\subsubsection*{Distributed systems, and the fundamental challenge thereof}
A distributed computing system is one that distributes its computation over several machines.
Distributed systems are essential to computing: they enable a network of parties to collaborate to achieve common goals~\cite{cachinbook}. 

\emph{The fundamental challenge of distribution} is that we cannot assume that all participants in the algorithm are correct: we need to be resilient against failure, and against hostile (also called \emph{Byzantine}) behaviour by some participants (a typical assumption is that up to a third of participants may be faulty in some way).
This has a profound effect on the flavour of distributed systems; and it is reflected in this paper by the use of a third truth-value which we use to represent faulty/hostile behaviour, and by the use of sophisticated modalities to describe quorum-based behaviour.\footnote{This gives `distributed systems' a flavour somewhat different from `concurrency', `parallelism', or `multi-agent systems': the fundamental character of the theory of distributed systems is particularly marked by the fact that algorithms must survive in an unreliable and hostile space.  (They solve this problem using notions of quorums, as we will discuss.)} 

Being distributed is attractive to software engineers because it can give us desirable properties, including some or all of: scalability, redundancy, lower latency, reliability, and resilience.
For this reason, almost any nontrivial commercial software system that the reader might care to name is implemented in a distributed manner.
Furthermore, blockchain systems (like Bitcoin and Ethereum\footnote{\url{https://bitcoin.org/en/} and \url{https://ethereum.org/en/}}) and peer-to-peer systems are inherently distributed by design; if they were centralised, then they would not even make sense.

\subsubsection*{The challenge of consensus}
Among the many challenges in distributed systems, \emph{consensus} represents a fundamental abstraction.
Consensus captures the problem of 
\begin{quote}
reaching agreement among multiple participants on a common value, in the face of possibly unreliable communication, and in the presence of possibly faulty participants.
\end{quote}
Consensus algorithms are an essential prerequisite for coordination in any distributed system, and they play a particularly critical role in blockchain networks.
In such networks, consensus is the mechanism that ensures all participants agree on the state of the ledger --- the blockchain's source of truth --- thereby solving problems like \emph{double spending} (manipulating the blockchain ledger to spend the same token twice).
Consensus protocols also enable the secure ordering of transactions, and they ensure the integrity and consistency of the blockchain, even in the face of adversarial actors.

Consensus protocols in blockchain systems are particularly challenging to design and implement correctly. 
Indeed, there are numerous examples in which incomplete or insufficiently specified protocols have led to serious issues, such as network halts or vulnerabilities that adversaries could exploit.
For example:
\begin{enumerate*}
\item
Solana (\url{https://solana.com}) has experienced instances where its consensus mechanism~\cite{yakovenko2018solana,shoup2022poh} contributed to temporary network halts (\url{https://solana.com/news/9-14-network-outage-initial-overview}).
\item 
Ripple's (\url{https://ripple.com}) consensus protocol has been analysed for conditions in which safety and liveness may not be guaranteed~\cite{DBLP:conf/opodis/Amores-SesarCM20}. 
\item
Avalanche's (\url{https://www.avax.network}) consensus protocol has faced scrutiny regarding its assumptions and performance in adversarial scenarios~\cite{DBLP:conf/opodis/Amores-SesarCT22}.
\end{enumerate*}
Even mature blockchains like Ethereum are not immune to challenges.
Attacks such as \emph{balancing attacks}~\cite{DBLP:conf/sp/NeuTT21} and \emph{ex-ante reorganisation} (\emph{reorg}) attacks~\cite{DBLP:conf/fc/Schwarz-Schilling22} have highlighted vulnerabilities in Ethereum's protocol.

These examples underscore how subtle flaws or ambiguities in protocol design can have significant consequences, emphasising the need for rigorous analysis and well-defined specifications in this complex and safety-critical field.

Compounding the fact that consensus algorithms are foundational to the security and functionality of blockchain systems, is the fact that designing and implementing these algorithms is notoriously challenging.
Distributed algorithms --- including consensus protocols --- are difficult (at least for human brains) because of the inherent difficulty of coordinating multiple participants, which may be faulty, hostile, drifting in and out of network connectivity~\cite{sleepy}, and/or have different local clocks, states, agendas, and so forth.

In short, this is a notoriously difficult field which is full of `gotchas', and within this difficult field, consensus protocols are amongst the most intricate, demanding, and practically important objects of study. 

\subsubsection*{Introducing a logic-based declarative approach (applied to Paxos)}
In this paper, we develop a logic-based declarative approach to modelling distributed algorithms.
It stands in relation to existing techniques in roughly the same way as declarative programming languages stand to imperative ones. %

Technically speaking, \QLogicShort is a \emph{three-valued modal fixpoint logic}. 
This means:
\begin{itemize*}
\item
There are three truth-values: $\tvT$ (true), $\tvB$ (no value returned, i.e. `crashed'), and $\tvF$ (false)~\cite{sep-logic-manyvalued}.
\item
Modalities let us express things like `some participant has property $\phi$' or `a quorum of participants has property $\psi$' or `$\phi$ has been true in the past' or `$\psi$ will be true in the future'.\footnote{In order, these are as follows in our logic: $\somewhere\phi$, $\someOallP\phi$, $\recent\phi$, and $\sometime\phi$.  Precise semantics are in Figures~\ref{fig.3.phi.f} and~\ref{fig.3.derived}.} 
\item
The operator $\mu$ imports the expressive power of iteration~\cite{wiki:Fixedpoint_logic}.
Consensus algorithms are often iterative, thus it is natural to include this in our logic.
\end{itemize*} 
We will apply these constructs to give axiomatic models of the Paxos consensus protocol~\cite{DBLP:journals/tocs/Lamport98,lamport2001paxos,cachin2011paxos} and then we will state and derive standard correctness properties from the axioms.
Our example is elementary but nontrivial and strikes (we think) a good balance between being compact and still exhibiting enough richness of structure to illustrate the essential features of how our logic works.

The first half of this paper will focus on the logic, and the second half of this paper will apply the logic to the examples.
The reader from a logic background might like to start with the first half and refer to examples for motivation; the reader from a distributed systems background might like to do the opposite and start with the examples in the second half, referring back to the logic as required to understand what is being done.
 
\subsubsection*{Related work using logic}
Logic-based approaches have long been used to model and verify distributed protocols, with TLA+~\cite{DBLP:journals/toplas/Lamport94,DBLP:books/aw/Lamport2002,DBLP:conf/charme/YuML99} being one of the most prominent frameworks in the field. 
Introduced by Leslie Lamport in the 1990s, TLA+ builds on Linear Temporal Logic (LTL)~\cite{DBLP:conf/focs/Pnueli77} to describe systems in terms of state variables and their evolution through explicit next-state relations.
It has since become an industry standard for specifying and verifying distributed algorithms.  

TLA+ excels at specifying and verifying behaviours through operational, state-based descriptions.
In this respect, it is well-suited for detailing how algorithms proceed step by step.
The reader can find a modern example of a TLA+ Paxos specification (part of a formalisation of a \emph{multi-Paxos} elaboration of the protocol) described in~\cite{ChandLiuStoller2016}; see Figures~2, 3, 4, and~5 of that paper and accompanying text.  
A direct link to the relevant code file is in~\cite{LamportMerzDoligez:PaxosTLA}. 

This operational perspective can create a certain distance between the specification and the invariants it seeks to uphold.
Although TLA+ provides powerful mechanisms for expressing and checking invariants, the burden lies in verifying that these invariants are preserved across all possible state transitions --- an exercise that can become complex and unintuitive, especially in the setting of distributed algorithms where subtle corner cases abound.

\subsubsection*{Our approach}
Our novel framework leverages the mathematics of semitopology and three-valued logic. 

Semitopology, introduced in previous work~\cite{gabbay:semdca,gabbay:semtad}, provides a natural way to describe quorums --- sets of participants that collaborate to progress the system state --- by representing them as open sets.  
This perspective is particularly well-suited to distributed systems, where system behaviour emerges from the interactions of coalitions of components rather than from a sequence of steps orchestrated by a central controller.  

Three-valued logic (for an overview, see~\cite[Chapter~5]{bergmann:intmvf}) enables reasoning about the uncertainty and potential faults inherent in distributed systems.  
Unlike classical logic, which admits only $\tvT$ and $\tvF$, three-valued logic introduces a third truth value $\tvB$ to capture indeterminate states, as arise when the behaviour of some parties cannot be fully trusted.  
This yields compact and natural descriptions of distributed protocols, distinguishing clearly between normal operation and faulty behaviour.  
It also provides a formal foundation that could, in future work, be connected to automated reasoning tools such as SAT solvers or theorem provers, thereby enhancing both the expressiveness and verifiability of protocol specifications.  

Our approach is especially relevant in the context of blockchain consensus protocols: distributed and decentralised systems that critically depend on precise and robust specifications.  
It is not an exaggeration to say that blockchains have accelerated the development of new consensus protocols at an unprecedented pace~\cite{DBLP:conf/wdag/CachinV17}.  
Yet many of these proposals remain only informally specified, often described in blog posts or non-academic sources, leading to ambiguity about their guarantees and behaviour.  
This lack of formalisation contributes to a fragmented landscape in which the properties of protocols are unclear, hindering adoption and undermining confidence.  

By providing a logic-based declarative language and a systematic way to axiomatise consensus protocols, our framework aims to support formal specification, rigorous reasoning, and eventual verification of the properties of distributed systems.  
We believe this can lower costs, increase reliability, and reduce development time.  
More broadly, we argue that effective formalisation is essential for eliminating ambiguity in the design and evaluation of emerging consensus mechanisms, particularly in the blockchain space.  

Our approach thus bridges the gap between the declarative specification of invariants and the practical verification of protocol correctness.  

\subsubsection*{We will explain everything, and that's a feature}

Different readers will find some parts of this paper straightforward and other parts more challenging.
Which parts these are will largely depend on the reader's background.%
\footnote{This is not a hypothetical: 
the first author, coming from a logic background, needed time to fully understand how Paxos works, while the second author, experienced with Paxos and familiar with logic in broad terms, was less at ease with modal fixed-point logic and models of axiomatic logical theories. The overlap of readers who are already well-versed in both fields may not be large. We hope this paper contributes to narrowing this gap.}

For example:
\begin{itemize*}
\item
three-valued modal logic and its models, and technical results like the Knaster-Tarski fixpoint theorem used in Proposition~\ref{prop.mu.fixpoint} may be known to the reader with a background in logic (like the first author), but perhaps not to a distributed systems theorist (like the second author); conversely, 
\item
the Paxos algorithm, together with standard technical devices such as the discussion motivating $n$-twined semitopologies in Remark~\ref{rmrk.motivate.n-twined}, will be familiar territory to a distributed-systems expert. 
However, their application in this logical setting may be new to readers coming from a logic background.
\end{itemize*}
We play safe and explain everything. 
If the reader finds we are going a bit slow in places, then better that than going too fast and skipping details that might help some other reader follow the argument.

\subsection{Structure of the paper}

\begin{enumerate*}
\item
Section~\ref{sect.intro} is the Introduction; You are Here.
The rest of the paper splits into two halves: Sections~\ref{sect.three.truth.values} to~\ref{sect.n-twined} have to do with logic and semitopologies; and Sections~\ref{sect.declarative.paxos} to~\ref{sect.simpler.paxos} use \QLogicShort to declaratively specify and study Paxos.
\ \\
\item
Section~\ref{sect.three.truth.values} introduces the basics of three-valued logic, including notable connectives ($\tnotor$ and $\timpc$) and modalities ($\modT$ and $\modTB$).
We also introduce a notion of validity for truth-values.
\item
Section~\ref{sect.quorum.logic} defines semitopologies (our model of the quorums that are ubiquitous in consensus algorithms), then it builds the syntax and semantics of \QLogicShort.
\item
Section~\ref{sect.predicates} studies \QLogicShort predicates; notably developing some logical equivalences, and well-behavedness criteria like being \emph{pointwise}.
\item 
Section~\ref{sect.n-twined} returns to semitopologies, and shows how antiseparation properties can be represented in logical form in \QLogicShort.
One particular antiseparation property corresponds to the \emph{quorum intersection property} which is used in the correctness proofs for Paxos (see Remark~\ref{rmrk.ntwined.and.the.literature}), but we will put this fact in a more general logical/semitopological context.
\ \\
\item
Section~\ref{sect.declarative.paxos} introduces Declarative Paxos; this is representative of the kinds of algorithms that we built \QLogicShort to express.
We take our time to explain how Paxos works, state the axiomatisation, and discuss it in detail.
\item
Section~\ref{sect.paxos.correctness.properties} states the standard correctness properties for Paxos, phrases them in declarative form, and then proves them from the axioms of Declarative Paxos.
What all our work has accomplished is to convert proving correctness properties into (relatively straightforward) proofs of facts about an abstract model, reasoning from axioms.
\item
Section~\ref{sect.simpler.paxos} revisits the axioms of Declarative Paxos from Section~\ref{sect.declarative.paxos}, in the light of the proofs in Section~\ref{sect.paxos.correctness.properties} --- and then prunes them down to a minimal theory required to prove correctness.

We call this minimal theory \emph{Simpler Declarative Paxos}: it is simpler, and because further details have been elided, it is more abstract.
The reader may prefer Declarative Paxos to Simpler Declarative Paxos because it is more identifiably related to the Paxos algorithm, but if we want to prove correctness properties, then Simpler Declarative Paxos distils the logical essence of what makes those proofs work. 
\ \\
\item
Section~\ref{sect.conclusions} concludes with conclusions and reflections on future work.
\end{enumerate*}

\section{Three truth-values}
\label{sect.three.truth.values}

\subsection{Basic definition of $\THREE$, and some truth-tables}

\begin{defn}[Three-valued truth-values]
\label{defn.THREE}
Let $\THREE$ be the poset of \deffont{truth-values} with elements $\tvF$ (false), $\tvB$ (broken), and $\tvT$ (true) ordered as 
$$
\tvF<\tvB<\tvT.
$$
\end{defn}

\begin{figure} %
$$
\begin{array}{c}
\begin{array}{c@{\qquad}c@{\qquad}c@{\qquad}c@{\qquad}c@{\qquad}c}
\begin{array}{c c c c}
p\tand q & \tvT & \tvB & \tvF
\\
\tvT &  \tvT & \tvB & \tvF
\\
\tvB &  \tvB & \tvB & \tvF 
\\
\tvF &  \tvF & \tvF & \tvF 
\end{array}
&
\begin{array}{c c c c}
p\tor q & \tvT & \tvB & \tvF
\\
\tvT &  \tvT & \tvT & \tvT
\\
\tvB &  \tvT & \tvB & \tvB
\\
\tvF &  \tvT & \tvB & \tvF
\end{array}
&
\begin{array}{c c c c}
p\tnotor q & \tvT & \tvB & \tvF
\\
\tvT &  \tvT & \tvB & \tvF
\\
\tvB &  \tvT & \tvB & \tvB
\\
\tvF &  \tvT & \tvT & \tvT
\end{array}
&
\begin{array}{c c c c}
p \timpc q & \tvT & \tvB & \tvF
\\
\tvT       & \tvT & \tvF & \tvF
\\
\tvB       & \tvT & \tvB & \tvB
\\
\tvF       & \tvT & \tvT & \tvT
\end{array}
\end{array}
\\[7ex]
\begin{array}{c@{\qquad}c@{\qquad}c@{\qquad}c@{\qquad}c@{\qquad}c@{\qquad}c}
\begin{array}{cc}
\tneg p
\\
\tvT & \tvF
\\
\tvB & \tvB
\\
\tvF & \tvT
\end{array}
&
\begin{array}{cc}
\modT p 
\\
\tvT & \tvT
\\
\tvB & \tvF
\\
\tvF & \tvF
\end{array}
&
\begin{array}{cc}
\modF p 
\\
\tvT & \tvF
\\
\tvB & \tvF
\\
\tvF & \tvT
\end{array}
&
\begin{array}{cc}
\modB p
\\
\tvT & \tvF
\\
\tvB & \tvT
\\
\tvF & \tvF
\end{array}
&
\begin{array}{cc}
\modTB p  
\\
\tvT & \tvT
\\
\tvB & \tvT
\\
\tvF & \tvF
\end{array}
&
\begin{array}{cc}
\modTF p  
\\
\tvT & \tvT
\\
\tvB & \tvF
\\
\tvF & \tvT
\end{array}
&
\begin{array}{cc}
\modFB p  
\\
\tvT & \tvF
\\
\tvB & \tvT
\\
\tvF & \tvT
\end{array}
\end{array}
\\[7ex]
\begin{array}{c@{\qquad}c@{\qquad}c@{\qquad}c@{\qquad}c@{\qquad}c}
\begin{array}{c c c c}
p\timp q & \tvT & \tvB & \tvF
\\
\tvT &  \tvT & \tvB & \tvF
\\
\tvB &  \tvT & \tvB & \tvF 
\\
\tvF &  \tvT & \tvT & \tvT 
\end{array}
&
\begin{array}{c c c c}
p\tleq q & \tvT & \tvB & \tvF
\\
\tvT &  \tvT & \tvF & \tvF
\\
\tvB &  \tvT & \tvB & \tvF 
\\
\tvF &  \tvT & \tvT & \tvT 
\end{array}
\end{array}
\end{array}
$$
\emph{Above, the vertical axis of a table indicates values for $p$; the horizontal axis (if nontrivial) denotes values for $q$.}
\caption{Truth-tables for some operations on $\THREE$ (Remark~\ref{rmrk.three.elementary})}
\label{fig.3}
\end{figure}

\begin{figure} %
$$
\begin{array}{c@{\qquad}c@{\qquad}c@{\qquad}c@{\qquad}c@{\qquad}c}
\begin{array}{c c c c}
p\tlatticeiff q & \tvT & \tvB & \tvF
\\
\tvT &  \tvT & \tvB & \tvF
\\
\tvB &  \tvB & \tvB & \tvB
\\
\tvF &  \tvF & \tvB & \tvT
\end{array}
&
\begin{array}{c c c c}
p \tiffcc q & \tvT & \tvB & \tvF
\\
\tvT       & \tvT & \tvF & \tvF
\\
\tvB       & \tvF & \tvB & \tvB
\\
\tvF       & \tvF & \tvB & \tvT
\end{array}
&
\begin{array}{c c c c}
p \tequiv q & \tvT & \tvB & \tvF
\\
\tvT       & \tvT & \tvF & \tvF
\\
\tvB       & \tvF & \tvT & \tvF
\\
\tvF       & \tvF & \tvF & \tvT
\end{array}
\end{array}
$$
\caption{Truth-tables for notions of equivalence (Definition~\ref{defn.equivalent})}
\label{fig.3.iff}
\end{figure}

\begin{rmrk}
\label{rmrk.three.elementary}
We start with some elementary observations about $\THREE$.
\begin{enumerate*}
\item\label{item.three.elementary.bounded.lattice}
$\THREE$ %
has bottom element $\tvF$, top element $\tvT$, and meets $\tand$ and joins $\tor$.
This makes it a \emph{finite lattice}.
\item
$\THREE$ has a negation operation $\tneg$.
\item
$\THREE$ has several natural modalities that identify truth-values, including $\modT$, $\modF$, $\modB$, $\modTB$, and $\modTF$.

For example, the reader can examine the truth-table for $\modTF$ in Figure~\ref{fig.3} and see that it identifies elements of $\{\tvT,\tvF\}$.\footnote{One could also call these modalities \emph{characteristic functions} for various subsets of $\THREE$, by which we mean functions which (literally) return $\tvT$ on elements in the set and $\tvF$ on elements not in the set.}
\item
$\THREE$ also has multiple notions of implication and equivalence, some of which are illustrated in Figures~\ref{fig.3} and~\ref{fig.3.iff}.
We will discuss these implications in Subsection~\ref{subsect.3.imp}.
\end{enumerate*}
\end{rmrk}

\begin{rmrk}
We will use the third truth-value $\tvB$ to represent participants in a distributed algorithm \emph{crashing} (or $\tvB$reaking).
The method is simple: the truth-value of a predicate describing the outcome of a crashed computation, is deemed to be $\tvB$.

This is similar in spirit to how in domain theory an additional value --- usually written $\bot$ --- is added to a denotation to represent nontermination; the function-value of a nonterminating computation is deemed to be $\bot$~\cite{abramsky:domt}.
The nonterminating computation does not return a value, but we represent this absence of a value in the denotation, using $\bot$. 

In our declarative axiomatisation of Paxos, we use $\tvB$ in much the same way, as a truth-value to describe the \emph{absence} of a return value (either $\tvT$ or $\tvF$) from participants that have crashed.
\end{rmrk}

\subsection{Valid and correct truth-values}

\begin{defn}
\label{defn.tv.ment}
\leavevmode
\begin{enumerate*}
\item\label{item.tv.ment.correct}
Define subsets $\threeValid,\threeCorrect,\threeTrue,\threeFalse\subseteq\THREE$ by
$$
\threeValid = \{\tvT,\tvB\}
\quad\text{and}\quad
\threeCorrect = \{\tvT,\tvF\} 
\quad\text{and}\quad
\threeTrue = \{\tvT\} 
\quad\text{and}\quad
\threeFalse = \{\tvF\} .
$$ 
\item
We refer to $\f{tv}\in\THREE$ as being \deffont{valid} / \deffont{invalid},\ \deffont{correct} / \deffont{incorrect},\ and \deffont{true} / \deffont{untrue},\ and \deffont{false} / \deffont{unfalse}\footnote{\dots `unfalse' is `$\tvT$-or-$\tvB$', not `$\tvT$'.  Similarly, `untrue' is $\tvB$-or-$\tvF$.  We cannot resist stating the following gem of three-valued clarity: not-$\tvT$, is not $\tneg\tvT$, and not-$\tvF$ is not $\tneg\tvF$, and not-$\tvB$ \emph{is} not-$\tneg\tvB$, because $\tvB$ is $\tneg\tvB$.} depending on whether it is / is not in $\threeValid$, $\threeCorrect$, $\threeTrue$, and $\threeFalse$ respectively.
\item\label{item.tv.ment}
Define $\ment\f{tv}$ \ (read `\deffont{$\f{tv}$ is valid}') by
$$
\ment\f{tv} 
\quad\text{when}\quad
\f{tv}\in\threeValid ,
\ \ \text{which just means}\ \ 
\f{tv}\in\{\tvT,\tvB\} .
$$
Thus, `$\f{tv}$ is valid' holds precisely when $\f{tv}$ is a valid truth-value.
\end{enumerate*}
\end{defn}

\begin{rmrk}[Why $\tvB$ is valid]
\label{rmrk.motivate.validity}
We take a moment to discuss Definition~\ref{defn.tv.ment}:
\begin{enumerate}
\item
We can think of $\tvT$ and $\tvF$ as specific `yes/no' results from the successful completion of some computation.
Therefore, we will call $\tvT$ and $\tvF$ \emph{correct}.
\item
We can think of $\tvB$ as representing failure of a computation to complete correctly --- e.g. because a participant has crashed and so can return neither $\tvT$ nor $\tvF$.
\item\label{item.tb.valid}
In Definition~\ref{defn.tv.ment} we set $\threeValid=\{\tvT,\tvB\}$ and define $\ment\f{tv}$ when $\f{tv}\in\threeValid$, and we call $\tvT$ and $\tvB$ \emph{valid}.

This identifies our treatment of validity as being that in Priest's \emph{logic of paradox} LP~\cite{Priest1979LogicOfParadox}, which sets $\tvT$ and $\tvB$ to be valid; the jargon is that $\tvT$ and $\tvB$ are \emph{designated}.
This is as opposed to Kleene's $K_3$, which makes the other choice of designating only $\tvT$~\cite{Kleene1952IntroductionToMetamathematics}.
\end{enumerate}
It might seem strange at first to let $\tvB$ be valid; if $\tvB$ represents the (non)behaviour of a crashed participant, then surely crashing should \emph{not} be valid?

Not necessarily.
To see why, consider that many of the assertions we consider will have the form 
\begin{quote}
`if $p$ is not crashed, then $\ment\phi$ at $p$', 
\end{quote}
or equivalently 
\begin{quote}
`either $p$ is crashed, or $\ment\phi$ at $p$'.
\end{quote}
If we let $\tvB$ be valid (i.e. if we define $\ment\tvB$ to hold), then we can simplify the above to the conveniently succinct assertion 
\begin{quote}
`$\ment\phi$ at $p$'.
\end{quote}
The `if $p$ is not crashed' precondition gets swallowed up into the fact that if $p$ is crashed then $\phi$ returns $\tvB$, and $\ment\tvB$ so there is nothing to prove. 
\end{rmrk}

\begin{rmrk}
Later on in Definition~\ref{defn.validity.judgement}(\ref{item.ctx.ment.phi}) we will extend validity to predicates-in-context, just by evaluating a predicate-in-context to a truth-value using a denotation function, and checking that truth-value is valid in the sense we define here. 
\end{rmrk}

We make some elementary observations about $\THREE$ and Definition~\ref{defn.tv.ment}:
\begin{lemm}
\label{lemm.tand.tor}
Suppose $\f{tv},\f{tv}'\in\THREE$ are truth-values.
Then:
\begin{enumerate*}
\item
$\ment\tvT$ and $\ment\tvB$.
\item
$\nment\tvF$.
\item\label{item.tand.tor.tor}
$\ment\f{tv}\tor\f{tv}'$ if and only if ${\ment\f{tv}}\lor{\ment\f{tv}'}$.
\item
$\ment\f{tv}\tand\f{tv}'$ if and only if ${\ment\f{tv}}\land{\ment\f{tv}'}$.
\end{enumerate*}
\end{lemm}
\begin{proof}
Fact of Definition~\ref{defn.tv.ment} and the truth-tables in Figure~\ref{fig.3}.
\end{proof} 

Negation is conspicuous by its absence in Lemma~\ref{lemm.tand.tor}.
We consider this next:
\begin{lemm}
\label{lemm.para}
Suppose $\f{tv}\in\THREE$ is a truth-value.
Then:
\begin{enumerate*}
\item\label{item.para.em}
$\ment\f{tv}\tor\tneg\f{tv}$, and as a corollary $\nment\f{tv}$ implies $\ment\tneg\f{tv}$.
\item\label{item.para.para}
$\ment\f{tv}\tand\tneg\f{tv}$ is possible, and as corollaries: $\ment\f{tv}$ does not necessarily imply $\nment\tneg\f{tv}$; and it is possible that $\ment\f{tv}$ and $\ment\tneg\f{tv}$ both hold.
\end{enumerate*}
\end{lemm}
\begin{proof}
We consider each part in turn:
\begin{enumerate}
\item
Using Lemma~\ref{lemm.tand.tor}(3) it would suffice to show that $\ment\f{tv}$ or $\ment\tneg\f{tv}$.
So suppose $\nment\f{tv}$.
By Definition~\ref{defn.tv.ment}(\ref{item.tv.ment}) $\f{tv}\not\in\threeValid$, thus $\f{tv}=\tvF$ and by Figure~\ref{fig.3} $\tneg\f{tv}=\tvT\in\threeValid$ and so $\ment\tneg\f{tv}$.
\item
We just take $\f{tv}=\tvB$ and check from Figure~\ref{fig.3} that $\tneg\tvB=\tvB$.
We use Lemma~\ref{lemm.tand.tor}(1\&4).
\qedhere\end{enumerate}
\end{proof}

\begin{rmrk}
Some comments on Lemma~\ref{lemm.para}:
\begin{itemize}
\item
Lemma~\ref{lemm.para}(\ref{item.para.para}) notes that our notion of validity in three-valued logic is naturally \deffont{paraconsistent}, meaning that it is possible for a thing, and the negation of that thing, to both be valid.
This is a feature of $\tvB\in\THREE$, not a bug, and we will make good use of it later when we use $\tvB$ %
as a return value for `I return no correct truth-value'. 
\item
Lemma~\ref{lemm.para}(\ref{item.para.em}) asserts validity of the \emph{excluded middle} (`$\phi$-or-negation-of-$\phi$'), but its corollaries implicitly illustrate another fact of three-valued logic, by what they do \emph{not} say: 
looking at the proof of Lemma~\ref{lemm.para}(\ref{item.para.em}), it is clearly also true that $\ment\f{tv}\tor\modF\f{tv}$, and that $\ment\f{tv}\tor\modT\tneg\f{tv}$, and that $\nment\f{tv}$ implies $\ment\modF\f{tv}$, and also that $\nment\f{tv}$ implies $\ment\modT\tneg\f{tv}$.  

These are stronger assertions, so why did we not state them in the lemma instead?
Because we will not need the result in those forms later.  

In two-valued logic, if something is true then it is probably a useful lemma.
In contrast, in three-valued logic, there is so much more structure with the third truth-value that more things are \emph{true} than are necessarily \emph{useful}.
So, we have to pick and choose our lemmas, and it is not always the case that the strongest form of a result is the most useful.
\end{itemize}
\end{rmrk}

Lemmas~\ref{lemm.tv.ment.TF} and~\ref{lemm.TF.correct} are easy, but they spell out some useful facts about truth-values and validity judgements:
\begin{lemm}
\label{lemm.tv.ment.TF}
Suppose $\f{tv}\in\THREE$ is a truth-value.
Then:
\begin{enumerate*}
\item\label{item.tv.ment.TF.1}
$\threeCorrect\cap\threeValid=\threeTrue$.

In words: a truth-value is true if and only if it is correct and valid.
\item\label{item.tv.ment.TF.2b}
$\ment \modT\f{tv}$ if and only if $\f{tv}\in\threeTrue$. 

In words: `$\modT\f{tv}$ is valid' holds if and only if $\f{tv}$ is true.
\item\label{item.tv.ment.TF.2b.F}
$\ment \modF\f{tv}$ if and only if $\f{tv}\in\threeFalse$. 

In words: `$\modF\f{tv}$ is valid' holds if and only if $\f{tv}$ is false.
\item\label{item.tv.ment.TF.TB}
$\ment \modTB\f{tv}$ if and only if $\f{tv}\in\threeValid$ if and only if $\ment\f{tv}$. 

In words: `$\modTB\f{tv}$ is valid' if and only if $\f{tv}$ is valid.\footnote{Note that $\modTB\f{tv}$ is not necessarily \emph{literally equal} to $\f{tv}$, in the sense of $\modT\f{tv}\tequiv\f{tv}$ from Definition~\ref{defn.equivalent}.  For example, $\modTB\tvB=\tvT\neq\tvB$.}
\item\label{item.tv.ment.TF.2}
$\ment \modTF\f{tv}$ if and only if $\f{tv}\in\threeCorrect$. 

In words: `$\modTF\f{tv}$ is valid' if and only if $\f{tv}$ is correct.
\end{enumerate*}
\end{lemm}
\begin{proof}
Part~\ref{item.tv.ment.TF.1} is a fact of sets.
Parts~\ref{item.tv.ment.TF.2b} and~\ref{item.tv.ment.TF.2b.F} and~\ref{item.tv.ment.TF.TB} and~\ref{item.tv.ment.TF.2} are just from the truth-tables for $\modT$, $\modF$, $\modTB$, and $\modTF$ in Figure~\ref{fig.3} and Definition~\ref{defn.tv.ment}. 
\end{proof}

\begin{lemm}
\label{lemm.TF.correct}
Suppose $\f{tv}\in\THREE$ is a truth-value.
Then the following are all equivalent:
$$
\ment\modTF\f{tv}
\ \liff\ 
\modTB\f{tv}=\modT\f{tv}
\ \liff\ 
\modT\f{tv}=\f{tv} 
\ \liff\ 
\modFB\f{tv}=\modF\f{tv}
\ \liff\ 
\modF\f{tv}=\tneg\f{tv} 
$$ 
Informally --- we do not build the machinery to make this formal, but it may still be a helpful observation --- we can write that the modality $\modTF$ is equivalent to the modal equivalences $\modTB=\modT=\f{id}$ and $\modFB=\modF=\tneg$.
\end{lemm}
\begin{proof}
By Lemma~\ref{lemm.tv.ment.TF}(\ref{item.tv.ment.TF.2}) and Definition~\ref{defn.tv.ment}(\ref{item.tv.ment.correct}),
$\ment\modTF\f{tv}$ when 
$$
\f{tv}=\tvT \ \lor\ \f{tv}=\tvF
\quad\text{and when}\quad
\f{tv}\neq\tvB .
$$
We check the truth-tables in Figure~\ref{fig.3}, and see that all of the other statements above precisely characterise the same condition.
\end{proof}

\subsection{The rich implication and equivalence structure of $\THREE$}
\label{subsect.3.imp}

\subsubsection{Weak and strong implications, and others}

Three-valued logic supports many kinds of implication, as has already been noted~\cite[note~5, page~22]{arieli:idepl}.
We discuss two in Definition~\ref{defn.three.imp}, and mention two more in Remark~\ref{rmrk.rich.implication}:
\begin{defn}
\label{defn.three.imp}
\leavevmode
\begin{enumerate*}
\item\label{item.three.elementary.weak.implication}
Define \deffont{weak implication} $\tnotor$, using $\tneg$ and $\tor$, as follows:
$$
\f{tv}\tnotor\f{tv}'=(\tneg\f{tv})\tor\f{tv}'.
$$
This `not-or' equality identifies $\tnotor$ as what some would call \emph{material implication}; it is the natural generalisation of Boolean `not-or' implication to the three-valued setting.
\item\label{item.three.elementary.strong.implication}
Define \deffont{strong implication} $\timpc$, using $\tnotor$ and $\modT$, as follows:
$$
\f{tv}\timpc\f{tv}'=\f{tv}\tnotor \modT\f{tv}'. 
$$
Strong implication appears to be original to this paper, so far as we know. 
The closest explicit analogue we can find in the literature is that it is equal to the contrapositive of a connective $p\supset q$ from page~278 of~\cite{Avron1991Natural}; that is, $p\timpc q$ in our notation is equal to $\tneg q\supset \tneg p$ in Avron's notation.
\end{enumerate*}
These and other encodings are assembled in Figure~\ref{fig.3.derived}.
\end{defn}

We motivate the weak and strong implications from Definition~\ref{defn.three.imp}(\ref{item.three.elementary.weak.implication}\&\ref{item.three.elementary.strong.implication}):
\begin{rmrk}[How to interpret $\tnotor$ and $\timpc$]
\label{rmrk.interpret.implication}
\leavevmode
For this Remark, let us interpret $\f{tv}\in\THREE$ as the result of a computation process which either succeeds and returns true (truth-value $\tvT$), or it succeeds and returns false (truth-value $\tvF$), or the process does not succeed (e.g. it crashes) and it returns a broken value / no value (represented by the truth-value $\tvB$).
Recall also from Definition~\ref{defn.tv.ment}(\ref{item.tv.ment}) that $\ment\f{tv}$ means $\f{tv}\in\{\tvT,\tvB\}$.
Then by this reading:
\begin{enumerate*}
\item
$\ment\f{tv}$ means \emph{``the process $\f{tv}$ returns $\tvT$, if valid''}. 
\item
$\ment\f{tv}\tnotor\f{tv}'$ means \emph{``\emph{if} process~$\f{tv}$ returns true \emph{then} process~$\f{tv}'$ returns $\tvT$, if valid''}.

Implication does not imply sequentiality: we are just making a logical assertion about outcomes. 
\item
$\ment\f{tv}\timpc\f{tv}'$ means \emph{``\emph{if} process~$\f{tv}$ returns true \emph{then} process~$\f{tv}'$ is \emph{also} valid, and it returns $\tvT$''}.
\end{enumerate*}
\end{rmrk}

We can sum this up as a \emph{modus ponens} result, which relates the validity $\ment$ from Definition~\ref{defn.tv.ment}, the weak implication $\tnotor$ and strong implication $\timpc$ from Figure~\ref{fig.3}, to real-world implication:
\begin{prop}[Weak and strong modus ponens]
\label{prop.mp.for.tnotor}
Suppose $\f{tv},\f{tv}'\in\THREE$ are truth-values.
Then:
\begin{enumerate*}
\item\label{item.mp.for.tnotor}
$\ment\f{tv}\tnotor\f{tv}'$ if and only if\ \ $\ment\modT\f{tv}$ implies $\ment\f{tv}'$. 

In words: `$\f{tv}$ weakly implies $\f{tv}'$' is valid, if and only if $\f{tv}$ is true implies $\f{tv}'$ is valid (but $\f{tv}'$ need not be true). 

We may call this logical equivalence %
\deffont{weak modus ponens}.
\item\label{item.mp.for.timpc}
$\ment\f{tv}\timpc \f{tv}'$ if and only if $\ment\modT\f{tv}$ implies $\ment\modT\f{tv}'$. 

In words: `$\f{tv}$ strongly implies $\f{tv}'$' is valid if and only if $\f{tv}$ is true implies $\f{tv}'$ is true. 

We may call this logical equivalence \deffont{strong modus ponens}. 
\end{enumerate*}
\end{prop}
\begin{proof}
We consider each part in turn:
\begin{enumerate}
\item
By direct inspection of the truth-table for $\tnotor$ in Figure~\ref{fig.3}, using Definition~\ref{defn.tv.ment}(\ref{item.tv.ment}).
\item
By a similar direct inspection.
\qedhere\end{enumerate}
\end{proof}

\begin{rmrk}[Explicit and implicit implications]
\label{rmrk.rich.implication}
So far, the structure we have seen (the lattice structure, the negation, the modalities, and the two implications $\tnotor$ and $\timpc$) is reflected explicitly in the syntax of our logic in Figures~\ref{fig.predicate.syntax} and~\ref{fig.3.derived}.
For instance, Figure~\ref{fig.predicate.syntax} has a unary negation connective $\tneg$ whose denotation as per Figure~\ref{fig.3.phi.f} is precisely the negation operator $\tneg$ in Figure~\ref{fig.3}.
So far, so explicit.

However, $\THREE$ has further implication structure, which we do not mention explicitly in our logic, but which is still there, is important, and which appears implicitly in the proofs.
Notably: 
\begin{enumerate*}
\item\label{item.rich.implication.timp}
An implication $\timp$ defined by $\f{tv}\timp \f{tv}'=\tneg\f{tv}'\timpc\tneg\f{tv}=(\modTB\f{tv})\tnotor\f{tv}'$ has a truth-table spelled out in Figure~\ref{fig.3} (bottom left).
This can be read as ``if $\f{tv}$ is valid then $\f{tv}'$ is valid''.

$\timp$ is clearly implicit in results like Proposition~\ref{prop.someone.implies.someoneAll} -- see all the parts whose statement uses logical implication $\limp$ --- and in Proposition~\ref{prop.tv.ment.TF.model}.
\item
Another implication is implicit in the use of $\leq$ in results like Lemma~\ref{lemm.positive.monotonicity} and Corollary~\ref{corr.recent.everyone.commute}.
That symbol refers to the lattice order on $\THREE$, but this could easily be internalised in our logic --- e.g. we can set $\f{tv}\tleq \f{tv'}$ to be $(\f{tv}\timpc\f{tv}')\tand(\f{tv}\timp\f{tv}')$, to obtain the connective $\tleq$ also illustrated in Figure~\ref{fig.3} (bottom right).
\end{enumerate*}
So in fact, we use all four of $\timp$ and $\tleq$, and $\tnotor$ and $\timpc$.

The reason that $\tnotor$ and $\timpc$ make it into our logic as (sugared) syntax, and $\timp$ and $\leq$ do not, is that $\tnotor$ and $\timpc$ naturally arise in our axioms (see most of the axioms in Figure~\ref{fig.logical.paxos}) %
whereas we do not need $\timp$ or $\tleq$ in the same way for the axiomatisations in this paper.
Different axiomatisations of different algorithms might make different demands on our logic, and it would be straightforward to include (sugar for) connectives accordingly. 
\end{rmrk}

\begin{rmrk}
A word of warning: $\tnotor$ curries with itself as the reader might expect; $\timpc$ displays currying behaviour too, but via interactions with $\tnotor$.
In more detail, the following yield equal truth-tables:
$$
\begin{array}{c@{\ \ }c@{\ \ }c@{\ \ }c@{\ \ }c}
p\tnotor(q\timpc r) &\text{equals}& (p\tand q)\timpc r
\\
p\timpc (q\timpc r) &\text{equals}& p\timpc (q\tnotor r) &\text{equals}& (p\tand \modTB q)\timpc r 
\\
p\tnotor(q\tnotor r) &\text{equals}& (p\tand q)\tnotor r 
\end{array}
$$
All three lines represent distinct truth-tables.
Note in particular that $p \timpc (q\timpc r)$ is not equal to $(p\tand q)\timpc r$.
\end{rmrk}

\subsubsection{Equivalence $\tequiv$}

Our three-valued domain of truth-values $\THREE$ has rich notions of equivalence, to match its rich notions of implication.
We mention three of them:
\begin{defn}[Notions of equivalence in $\THREE$]
\label{defn.equivalent}
Define $\phi\tlatticeiff\phi'$ and $\phi\tiffcc\phi'$ and $\phi\tequiv\phi'$ by:
$$
\begin{array}{r@{\ }l}
\phi\tlatticeiff\phi' =& (\phi\tnotor\phi')\tand(\phi'\tnotor\phi)
\\
\phi\tiffcc\phi' =& (\phi\timpc\phi')\tand(\phi'\timpc\phi)
\\
\phi\tequiv\phi' =& (\modT\phi\tand\modT\phi')\tor(\modB\phi\tand\modB\phi')\tor(\modF\phi\tand\modF\phi')
\end{array}
$$
Truth-tables for $\tlatticeiff$, $\tiffcc$, and $\tequiv$ are given in Figure~\ref{fig.3.iff}, and we unpack this further in Remark~\ref{rmrk.equiv}.
\end{defn}

\begin{rmrk}
\label{rmrk.equiv}
Figure~\ref{fig.3.iff} gives us three notions of equivalence: $\f{tv}\tlatticeiff \f{tv}'$, $\f{tv}\tiffcc\f{tv}'$, and $\f{tv}\tequiv \f{tv}'$. 
Of the three, $\tequiv$ will be most useful to us later.
We will consider it in our predicate language in Subsection~\ref{subsect.equivalence.of.predicates}, and (for instance) it will help us to express notions of extensional equivalence in Figure~\ref{fig.easy.equivalences}.

We will call $\tequiv$ \deffont{literal equivalence}, and $\tequiv$ appears in the sugared syntax of Figure~\ref{fig.3.derived}.
As per the discussion in Remark~\ref{rmrk.rich.implication}, there would be no barrier to including $\tlatticeiff$ or $\tiffcc$ but we do not need to for now.
\end{rmrk}

\section{\QLogic (\QLogicShort)}
\label{sect.quorum.logic}

\subsection{A preliminary: semitopologies}

\begin{rmrk}
A semitopology is like a topology, but generalised to not require that intersections of open sets necessarily be opens.

The intuition of an open set is an \emph{actionable coalition}; a set of participants with the voting power (or other resources) to make updates to their local state.
The reader might be familiar with the notion of \emph{quorum} from literature like~\cite{DBLP:journals/tocs/Lamport98,DBLP:journals/dc/MalkhiR98,naor:loacaq}, or the notion of \emph{winning coalition} in social choice theory (the theory of voting)~\cite[Item~5, page~40]{riker:thepc}.
Semitopologies are similar to this, except that open sets are always closed under arbitrary (possibly empty) unions, whereas 
\begin{itemize*}
\item
quorums may be closed under arbitrary nonempty unions (in which case the correspondence is precise, except that $\varnothing$ is an open set but is not a quorum), or 
\item
quorums may not be closed under arbitrary unions (but usually it is harmless to suppose that they are).\footnote{The reader who comes from a distributed systems background might be used to \emph{threshold} quorum systems, where quorums have the form `more than $k$ participants' for some $k$. 
These quorum systems are closed under unions, but they are also \emph{up-closed}; if $Q$ is a quorum and $Q\subseteq Q'$ then $Q'$ is a quorum.

This is an important special case, but there \emph{are} natural systems of actionable coalitions --- quorums, in the reader's language --- that are not up-closed.  
For example, consider a quorum of Tezos nodes $T$, and a quorum of Ethereum nodes $E$, and a single bridging node $b$ whose job is to facilitate token transfers between the two.  Then $T$ is a quorum, and $T\cup\{b\}\cup E$ is a quorum of the combined system, but $T\cup\{b\}$ and $\{b\}\cup E$ are \emph{not} quorums because the bridging node requires one quorum from each system to agree to a token transfer before this transfer is considered actionable.

To take a more domestic example, consider the situation of sampling a packet of Pringles crisps in a new flavour (UK English; in US English they are called chips).
You may be disciplined and have no crisps, or you may eat the top crisp from the packet and not like the flavour --- but if you do have that one crisp, and you \emph{do} like the flavour, then the only possible remaining action is to eat the entire packet.

Having (say) three crisps and stopping there is not an option.
There is even a word for this: more-ish.
We can analyse this situation mathematically as the \deffont{more-ish semitopology}, which is non-up-closed, containing the empty set $\varnothing$, a singleton set $\{x\}$ (for some $x$ which we call the \emph{first element}), and then the whole set of elements.
}
\end{itemize*}
Our presentation of semitopologies will be quite brief --- we just define the parts we need to give semantics to the modalities in \QLogicShort --- and the reader who wants to know more is referred to previous work~\cite{gabbay:semdca}. 
\end{rmrk} 

\begin{defn}
\label{defn.semitopology}
A \deffont{semitopology} is a pair $(\Pnt,\opens)$ where:
\begin{enumerate*}
\item
$\Pnt$ is a set of \deffont{points}, which following terminology from distributed systems one might also call \deffont{participants}, \deffont{processes}, \deffont{parties}, or \deffont{nodes}.
\item
$\opens\subseteq\powerset(\Pnt)$ is a set of \deffont{open sets} such that:
\begin{enumerate*}
\item
If $\mathcal O\subseteq\opens$ then $\bigcup\mathcal O\in\opens$.
In words: $\opens$ is closed under arbitrary (possibly empty) unions.
\item
$\Pnt\in\opens$.
\end{enumerate*} 
\end{enumerate*} 
We may write $\opensne$ for the set of nonempty open sets, so
$$
\opensne = \{O\in\opens \mid O\neq\varnothing\} .
$$
\end{defn}

\begin{nttn}
\label{nttn.empty.semitopology}
Call the semitopology $(\varnothing,\{\varnothing\})$ the \deffont{empty semitopology}.
Call a semitopology that is not the empty semitopology, a \deffont{nonempty semitopology}.
It is easy to check that a semitopology is nonempty if and only if its set of points is nonempty. 
\end{nttn}

\begin{defn}[All-but-$f$-many semitopology]
\label{defn.allbut}
\leavevmode
\begin{enumerate*}
\item
Suppose $X$ is a finite set and $N\in\Ngeqz$.
Define $\powerset_{\geq N}(X)$ to be the set of subsets of $X$ that contain at least $N$ elements.
In symbols:
$$
\powerset_{\geq N}(X) = \{O\subseteq X \mid \mycard O\geq N\} .
$$
\item\label{item.allbut}
Suppose $N,f\in\Ngeqz$ and $N\geq f$.
Define a semitopology $\AllBut{N}{f}$ by
$$
\AllBut{N}{f} = \bigl(\Nset{N}, \{\varnothing\}\cup\powerset_{\geq N\minus f}(\Nset{N}\bigr)\bigr).
$$
Thus: points are numbers $0\leq n<N$; and nonempty open sets are sets of points of cardinality at least $N\minus f$.
\end{enumerate*}
\end{defn}

\begin{lemm}
Suppose $N,f\in\Ngeqz$ and $N\geq f$.
Then:
\begin{enumerate*}
\item
$\AllBut{N}{f}$ is indeed a semitopology.
\item
$\AllBut{N}{f}$ is a topology precisely when $f\in\{N,N\minus 1,0\}$.
\end{enumerate*}
\end{lemm}
\begin{proof}
We consider each part in turn:
\begin{enumerate}
\item
It is routine to check that the conditions on open sets in Definition~\ref{defn.semitopology} are satisfied.
The key observation is just the elementary fact that if $\mycard O\geq N\minus f$ and $\mycard O'\geq N\minus f$ then $\mycard(O\cup O')\geq N\minus f$.
\item
For $\AllBut{N}{f}$ to be a topology, open sets need to be closed under intersections.
It is routine to check that this can happen only when $\AllBut{N}{f}$ is the full powerset (so $f=N$ or $f=N\minus 1$) or when open sets are just $\varnothing$ and $\Pnt$ (so $f=0$).
\qedhere\end{enumerate}
\end{proof}

\begin{rmrk}
The original Paxos paper by Lamport~\cite{DBLP:journals/tocs/Lamport98} used semitopologies of the form $\AllBut{N}{f}$ where $N=2n$ and $f=n\minus 1$, for $n\geq 1$ (i.e. open sets are strict majorities of points).
This invites a question: why bother with semitopologies at all?
\begin{enumerate*}
\item
Because we can. 
More specifically, our axiomatisation is parametric over the underlying semitopology.
As good mathematicians, if we can generalise at no cost, then we should.\footnote{\dots and this applies not just in mathematics papers.  Software engineers follow a similar principle and have many names for it: e.g. encapsulation, APIs, modules, classes, typeclasses, and so on.} 
\item
Semitopology not only generalises; it simplifies.
The only properties of $\AllBut{N}{f}$ that we will need, are precisely that it is a semitopology.
If all we need about $\AllBut{N}{f}$ is its semitopological structure, then we should work with that structure because that is what actually matters.
\item
Viewing $\AllBut{N}{f}$ as a semitopology, rather than just as a collection of quorums, adds real meaning and understanding.
In particular, the notions of `dense subset' and `nonempty open interior' will be relevant to understanding the axioms, and these arise naturally in the predicates for declarative axiomatisations of distributed algorithms.
These are clearly native topological notions, so the mathematics is speaking to us here, and what it has to say is inherently topological. 
\end{enumerate*} 
\end{rmrk}

\subsection{Syntax of the logic}

We now develop our three-valued modal fixpoint logic for reasoning about distributed systems.
We do this in the following order: \emph{signatures} $\Sigma$, \emph{models} $\mathcal M$, \emph{contexts} $\acontext$, and finally \emph{syntax}. 
We develop denotation for this syntax in Subsection~\ref{subsect.denotation}.
For quick reference, the reader can just look at the core syntax and denotation in Figures~\ref{fig.predicate.syntax} and~\ref{fig.3.phi.f}.

\begin{defn}[Signatures]
\label{defn.signature}
\leavevmode
\begin{enumerate*}
\item\label{item.signature.varsymb}
Fix $\tf{VarSymb}$ a countably infinite set of \deffont{variable symbols} $\thea,\theb,\thec\in\tf{VarSymb}$.
This set will never change and we may assume it henceforth.
\item
Fix a \deffont{fixpoint variable symbol} $\tf X$ (we will only need one of these, though there is no inherent obstacle to having more).
\item\label{item.predicate.syntax.signature}
A \deffont{signature} $\Sigma$ is a pair 
$$
\Sigma=(\PredSymb,\arity)
$$ 
where:
\begin{enumerate*}
\item\label{item.predsymb}
$\PredSymb$ is a (possibly finite) set of \deffont{predicate symbols} $\tf P,\tf P'\in\PredSymb$.
\item
$\arity$ is an \deffont{arity function} $\arity(\tf P)\in\Ngeqz$ mapping each predicate symbol in $\PredSymb$ to a nonnegative natural number (which is intuitively its number of arguments).
\end{enumerate*}
\item\label{item.predicate.syntax.signature.list.form}
We may write a signature $\Sigma$ in list form as 
$$
\Sigma=[\tf P,\tf P':i,\ \tf Q:j,\ \dots]
$$ 
meaning that $\tf P,\tf P',\tf Q\in\PredSymb$ and $\arity(\tf P)=\arity(\tf P')=i$ and $\arity(\tf Q)=j$ and so on.
\item\label{item.predicate.syntax.empty.signature}
If the signature in part~\ref{item.predsymb} above contains no predicate symbols then we call it \deffont{empty}; otherwise we call it nonempty.
There is only one empty signature, which we will call \deffont{the empty signature}.\footnote{It is not particularly useful.} 
\end{enumerate*}
\end{defn}

\begin{nttn}
\label{nttn.nullary}
We mention some standard terminology:
\begin{enumerate*}
\item\label{item.nullary.ps}
If $\arity(\tf P)=0$ then we may call $\tf P$ \deffont{nullary} or a \deffont{predicate constant symbol}.\footnote{$\tvB$ is a predicate constant symbol, but it is special so we give it a distinct identity in the syntax.}
\item\label{item.unary.ps}
If $\arity(\tf P)=1$ then we may call $\tf P$ a \deffont{unary} predicate symbol.
\item\label{item.atomic.predicate}
We may call a predicate of the form $\tf P(v_1,\dots,v_n)$ or $\tneg\tf P(v_1,\dots,v_n)$ an \deffont{atomic predicate}. 
\end{enumerate*}
\end{nttn}

\begin{defn}[Models]
\label{defn.model}
A \deffont{model} $\mathcal M$ is a tuple $\mathcal M=(\Pnt,\opens,\Val)$ where:
\begin{enumerate*}
\item
$(\Pnt,\opens)$ is a nonempty semitopology (Definition~\ref{defn.semitopology} and Notation~\ref{nttn.empty.semitopology}), which we may call the \deffont{underlying semitopology} of the model.\footnote{We insist on nonemptiness to avoid that $n\mentval\phi$ in Definition~\ref{defn.validity.judgement}(\ref{item.ment.n}), meaning $\Forall{p\in\Pnt}\Forall{O\in\opensne}(n,p,O\mentval\phi)$, might be vacuously satisfied because there are no points.  

Specifically, Lemma~\ref{lemm.dense.char}(\ref{item.dense.char.someoneAll}\&\ref{item.dense.char.QuorumBox}) and Proposition~\ref{prop.pointwise.dense.char}(\ref{item.dense.char.pointwise.someoneAll}\&\ref{item.dense.char.pointwise.QuorumBox}) fail on the edge case of the empty semitopology.  Rather than carry a side-condition in these results and their consequences, it is convenient to just exclude the edge case.  

Just to really spell this out, we would also have $n\mentval\tbot$ on the empty semitopology.  This is not wrong --- $n,p\mentval\tbot$ \emph{does} hold for every $p\in\varnothing$, vacuously and regardless of what $n,p\mentval\tbot$ even means! --- but this is arguably pathological.  Again, we just exclude this case.} 

We may call $p\in\Pnt$ a \deffont{point} or a \deffont{participant}; also \deffont{party}, \deffont{process}, or \deffont{node} would be consistent with the literature.
For our purposes, these are all synonyms: the terminology \emph{point} is from topology; \emph{participant}, \emph{party}, and \emph{process} are from distributed systems; and \emph{node} is from graph theory (thinking of a network as a graph of communicating processes).
\item\label{item.model.val}
$\Val$ is a nonempty set of \deffont{values}.\footnote{We assume $\Val$ is nonempty to avoid the degenerate case where a universal quantification over values holds just because there are no values on which to test it.  See in particular the argument in Proposition~\ref{prop.intertwined.char}(\ref{item.intertwined.char.reverse}).}

For our purposes we will not care what values are, so long as some exist, so we do not put any type system on values or endow them with structure; but we could easily do so if required. 
\end{enumerate*} 
\end{defn}

\begin{defn}[Contexts]
\label{defn.contexts}
Suppose $\Sigma=(\PredSymb,\arity)$ is a signature, and $\mathcal M=(\Pnt,\opens,\Val)$ is a model. 
\begin{enumerate*}
\item\label{item.context.stage}
A \deffont{stage} is a nonnegative number $n\in\Time$.
\item\label{item.contexts.valuation}
A \deffont{valuation} $\avaluation$ inputs a predicate symbol $\tf P\in\PredSymb$ and a stage $n\in\Time$ and a point $p\in\Pnt$ and returns a function $\avaluation(\tf P)_{n,p}:\Val^{\arity(\tf P)}\to\THREE$.

We write the set of valuations $\Valuation(\Sigma,\mathcal M)$.
\item\label{item.contexts}
A \deffont{context} $\acontext$ is a tuple $(n,p,O,\avaluation)$ where $n\in\Time$ is a stage, $p\in\Pnt$ is a point, $O\in\opensne$ is a nonempty open set, and $\avaluation\in\tf{Valuation}(\Sigma,\mathcal M)$ is a valuation. 

We write the set of contexts $\tf{Ctx}(\Sigma,\mathcal M)$; if $\Sigma$ and $\mathcal M$ are understood, we may write this just $\tf{Ctx}$.
\item
A \deffont{context valuation} is a function $\varkappa\in \THREE^{\tf{Ctx}}$.
For reference, context valuations will inherit a bounded lattice structure from $\THREE$ in the natural way; see Figure~\ref{fig.chi.3.lattice}.
\end{enumerate*}
\end{defn}

\begin{figure} %
$$
\begin{array}{r@{\ }l}
t ::=& 
(a\in\tf{VarSymb}) \mid (v\in\Val) %
\\
\phi ::=&  
\tvB 
\mid t\teq t \mid (\tf P(\overbrace{t,\dots,t}^{\scalebox{0.65}{$\arity(\tf P)$}}) : \tf P\in\PredSymb)
\\
\mid &\tneg\phi \mid \modT\phi \mid \somewhere\phi \mid \Quorum\phi \mid \texi a.\phi
\\
\mid & \phi\tor\phi'  
\\
\mid & \yesterday \phi \mid \tomorrow \phi 
\\
\mid & \tf X \mid \mu\tf X.\phi \mid (\varkappa \in \THREE^{\tf{Ctx}})
\end{array}
$$ 
\caption{Term and predicate syntax of \QLogicShort (Definition~\ref{defn.predicate.syntax})}
\label{fig.predicate.syntax}
\end{figure}

\begin{defn}[Syntax]
\label{defn.predicate.syntax}
Suppose $\Sigma$ is a signature (Definition~\ref{defn.signature}) and $\mathcal M$ is a model (Definition~\ref{defn.model}).
\begin{enumerate*}
\item 
Define the syntax of \deffont{\QLogicShort} over $\Sigma$ and $\mathcal M$ as in Figure~\ref{fig.predicate.syntax}.

In this syntax, $\texi\thea.\phi$ binds $\thea$ in $\phi$, and $\mu\tf X.\phi$ binds $\tf X$ in $\phi$.
We call a term or predicate with no free variable symbols or fixpoint variable symbols \deffont{closed}.
\item
One special case will be of particular interest: we will write 
\begin{quote}
`suppose $\texi a.\phi$ is a closed predicate',
\end{quote}
and often we do not really care about $\texi a.\phi$ \emph{per se}, except that when this is a closed predicate it means that $a\in\tf{VarSymb}$ and $\phi$ is a predicate whose set of free variables is at most $\{a\}$.
See for example Remark~\ref{rmrk.unpack.texiunique} and Lemma~\ref{lemm.unique.affine.existence}.
\item
We define substitution of variable symbols / fixpoint variable symbols by values / context valuations as usual, and write these $\phi[a\ssm v]$ and $\phi[\tf X\ssm \varkappa]$ respectively. %
\end{enumerate*}
\end{defn}

\begin{rmrk}
We apply \QLogicShort to Paxos in this document, but we see this logic as having independent interest and a broader relevance for specifying and reasoning about the kinds of algorithms that get written for distributed systems.
Their hallmark is a requirement for distributed collaborative action --- action by quorums, or actionable coalitions, or winning coalitions; these things have many names --- to progress an algorithm to its next step. 
\end{rmrk}

\begin{xmpl}
We give some example predicates and suggest intuitive meanings for them.
We do not yet have the denotation (Figure~\ref{fig.3.phi.f}) nor the rich derived syntax (Figure~\ref{fig.3.derived}) that we build later, so these examples are simple and informal, but we hope they will be helpful.

Assume our signature has a unary predicate symbol $\tf{output}$ (representing the output of some process), and assume $v,v'\in\Val$ are values.
\begin{enumerate}
\item
$\tf{output}(v)$ means `where we are, today, if the process succeeds, then it outputs $v$' and $\tneg\tf{output}(v')$ means `where we are, if the process succeeds, then it does not output $v'$'.
\item
$\modT\tf{output}(v)$ means `where we are, today, the process succeeds and outputs $v$'. 
\item
$\somewhere\tf{output}(v)$ means `there is a place where, today, if the process succeeds, then it outputs $v$'.
\item
$\tomorrow\modT\tf{output}(v)$ means `where we are, tomorrow, the process succeeds and outputs $v$'.
\item
$\yesterday\modT\tf{output}(v)$ means `where we are, yesterday, the process succeeded and output $v$'.
This assertion is subtly stronger than it might seem; our notion of time starts at day~0, at which by convention all assertions about yesterday are false.
See Remark~\ref{rmrk.initial.comments.denotation}(\ref{item.interpret.yesterday}).
\item
$\tvB$ means `where we are, today, the process (which we have not named with a predicate symbol) crashes'.
\end{enumerate}
More sophisticated examples will appear later; notably in the derived syntax in Figure~\ref{fig.3.derived}, in Remark~\ref{rmrk.comments.on.semantics}, and of course when we apply our logic to axiomatise distributed algorithms.
\end{xmpl}

\begin{rmrk}[Three pedant points]
\label{rmrk.pedant.point}
The thing about being pedantic is that it does not matter, until it does.
We now discuss three pedant points, which may not matter \dots until perhaps a moment arrives when they do:
\begin{enumerate}
\item
We may treat $\Sigma$ and $\mathcal M$ as fixed parameters, and elide them.  
In practice this means that we may write `suppose $\phi$ is a closed predicate' instead of `suppose $\phi$ is a closed predicate in the syntax of \QLogicShort over $\Sigma$ and $\mathcal M$'.
But note next point.
\item\label{item.model.matters}
Our syntax from Definition~\ref{defn.predicate.syntax} and Figure~\ref{fig.predicate.syntax} contains denotational elements: $v\in\Val$ and $\varkappa\in\THREE^{\tf{Ctx}(\Sigma,\mathcal M)}$.

In case the reader has not seen denotations added directly to syntax before, this is a standard technical device, in the style for example of Shoenfeld~\cite{Shoenfield1967} (page~18, Subsection~2.5 ``for each individual $a\in\mathcal A$ we choose a new constant, called the \emph{name} of $a$'') or Smullyan~\cite{Smullyan1968} (page~46, `formulas with constants in $\mathsf U$').\footnote{The alternative is augment our notion of context in Definition~\ref{defn.contexts}(\ref{item.contexts}) with two valuations, e.g. $\sigma$ mapping $a\in\tf{VarSymb}$ to $v\in\Val$, and some element $\varkappa$ representing a value for $\tf X$ respectively.  We prefer not to do this because our notion of context is already a 4-tuple $(n,p,O,\avaluation)$, and $(n,p,O,\avaluation,\sigma,\varkappa)$ would be too much of a notational handful.}
It is also key to standard proofs of well-known mathematical results, like the \emph{Lindenbaum-Tarski algebras} used for proving completeness, and for proving the upward and downward L\"owenheim-Skolem theorems~\cite{chang:modt}.

In this paper, it means that when we write `suppose $\phi$ is a closed predicate' then this statement is perfectly correct and well-defined but it \emph{really does matter} what the signature and model are (even if the \emph{precise} choice of model will not matter to the theorems we prove), because closed predicates $\phi$ can include elements from the denotation such as $\varkappa\in\THREE^{\tf{Ctx}(\Sigma,\mathcal M)}$ (this function-space is determined by both the signature $\Sigma$ and the model $\mathcal M$) and $v\teq v'$ for some $v,v'\in\Val$ ($\Val$ is determined by the choice of model $\mathcal M$).

Note also that the choice of model matters for validity (as it must!).
For example: in Definition~\ref{defn.phi.equivalent} we define equivalence of predicates (and we explicitly mention $\Sigma$ and $\mathcal M$);
if $\mathcal M$ has only a single value $v$, then $\tall a,b.a\teq b$ and $\tall a.a\teq v$ are equivalent to $\tvT$; and if $\mathcal M$ has more than one value then $\tall a,b.a\teq b$ is equivalent to $\tvF$ --- $\tall a.\phi$ is shorthand for $\tneg\texi a.\tneg\phi$; see Figure~\ref{fig.3.derived}.
In contrast: in Lemma~\ref{lemm.equiv} we elide $\Sigma$ and $\mathcal M$ for brevity, and just write `suppose $\phi$ and $\phi'$ are closed predicates and $\acontext$ is a context', without specifying over what signature or model, because this is clear in context.

The reader can mentally add $\Sigma$ and $\mathcal M$ everywhere if they like.
We will always be precise and explicit where it counts.
\item\label{item.pedant.texi.v}
We may quantify over $v\in\Val$ in our formal syntax of predicates, writing things like $\texi v.\texi v'.v\teq v'$ as predicates.
Technically this is wrong and is not syntax! 
Figure~\ref{fig.predicate.syntax} permits e.g. $\texi a.\texi a'.a\teq a'$ where $a,a'\in\tf{VarSymb}$ (Definition~\ref{defn.signature}(\ref{item.signature.varsymb})); it does not permit $\texi v$ in the formal syntax of predicates where $v\in\Val$, so writing e.g. $\texi v.\texi v'.v\teq v'$ as if it were formal syntax of \QLogicShort is an abuse of notation.

However, $a$ and $a'$ get instantiated to values $v$ and $v'$ in our denotation (see the line for $\texi$ in Figure~\ref{fig.3.phi.f}).
It seems not worth the disruption of maintaining a distinction between $\texi a.\phi$ (which is syntax, assuming $\phi$ is) and $\texi v.\phi[a\ssm v]$ (which looks like syntax but technically is not).
So in practice this seems to us one example of an abuse of formal notation that is positively useful and clarifies meaning.

If this bothers the reader, then every time they see $\texi v.\phi$ (and $\tall v.\phi$, and similar) in our text, they can either declare that $v$ here is \emph{actually} a variable symbol in $\tf{VarSymb}$ which we happen to have called $v$ (and not a value in $\Val$ after all), or they can pick a fresh variable symbol $a$ and replace every unbound $v$ in $\phi$ with $a$.
Either way they will get well-formed syntax and everything will be as it should be.

Notation exists to serve clarity, and not the other way around, so we will simply ignore this henceforth and carry on writing $\texi v.\phi$, just whenever we feel like it, and this will denote just what it says: `there exists a $v$ such that $\phi$'.
Meaning should always be clear. 

Note that an unbound value $v$ in a predicate is legal syntax as per Figure~\ref{fig.predicate.syntax}; 
e.g. in `$\texi a.a\teq v$', $v$ is a value, and this is fine; the (constructive and useful) abuse of notation comes when we write $\texi v$. 
\end{enumerate}
\end{rmrk}

\subsection{Denotation (semantics) of the syntax}
\label{subsect.denotation}

\begin{figure}
$$
\begin{array}{r@{\ }l}
\tbot =& \lambda\acontext.\tvF
\\
\ttop =& \lambda\acontext.\tvT
\\
\bigvee\mathcal X =& \lambda\acontext.\bigvee\{ \varkappa\,\acontext \mid \varkappa\in\mathcal X\}
\\
\bigwedge\mathcal X =& \lambda\acontext.\bigwedge\{ \varkappa\,\acontext \mid \varkappa\in\mathcal X\}
\\[2ex]
\varkappa\leq\varkappa' \liff& \Forall{\acontext} \varkappa\,\acontext\leq\varkappa'\,\acontext
\end{array}
$$
Above, $\mathcal X\subseteq \THREE^{\tf{Ctx}}$ and $\varkappa,\varkappa'\in\THREE^{\tf{Ctx}}$. 
\caption{Bounded complete lattice structure on the function space $\THREE^{\tf{Ctx}}$}
\label{fig.chi.3.lattice}
\end{figure}

We start with an elementary observation about $\THREE^{\tf{Ctx}}$, that it inherits a bounded lattice structure from $\THREE$ in the standard way:
\begin{defn}
\label{defn.three.ctx.lattice}
Recall from Remark~\ref{rmrk.three.elementary}(\ref{item.three.elementary.bounded.lattice}) that $\THREE$ with ordering $\tvF<\tvB<\tvT$ is a finite lattice.
As is standard, the function space $\THREE^{\tf{Ctx}}$ (functions from $\tf{Ctx}$ to $\THREE$) inherits this ordering and lattice structure \deffont{pointwise}, as indicated in Figure~\ref{fig.chi.3.lattice}.
\end{defn}

Definition~\ref{defn.three.ctx.lattice} makes $\THREE^{\tf{Ctx}}$ into a \deffont{bounded complete lattice}, meaning that it has a least and greatest element, and arbitrary joins (least upper bounds) and meets (greatest lower bounds).\footnote{Having a least and greatest element follows from having arbitrary joins and meets; the least element is the join of the empty set, and the greatest element is the meet of the empty set.  So in fact a complete lattice is already a bounded complete lattice.} 
We will need arbitrary meets for the case of $\mu\tf X$ in Figure~\ref{fig.3.phi.f}.

\begin{figure} %
$$
\begin{array}{r@{\ }l@{\quad}l}
\modellabel{\tvB}{n,p,O,\avaluation} =&\tvB
\\
\modellabel{v\teq v'}{n,p,O,\avaluation} =&
\begin{cases}
\tvT & v=v'
\\
\tvF & v\neq v'
\end{cases}
\\
\modellabel{\tf P(v_1,\dots,v_{\arity(\tf P)})}{n,p,O,\avaluation} =&\avaluation(\tf P)_{n,p}(v_1,\dots,v_{\arity(\tf P)}) %
\\[2ex]
\modellabel{\phi\tor\phi'}{n,p,O,\avaluation} =& \modellabel{\phi}{n,p,O,\avaluation}\tor\modellabel{\phi'}{n,p,O,\avaluation}
\\
\modellabel{\tneg\phi}{n,p,O,\avaluation} =& \tneg\,\modellabel{\phi}{n,p,O,\avaluation}  
\\
\modellabel{\modT\phi}{n,p,O,\avaluation} =& \modT\modellabel{\phi}{n,p,O,\avaluation}  
\\
\modellabel{\somewhere\phi}{n,p,O,\avaluation} =& \bigvee_{p'\in O} \modellabel{\phi}{n,p',O,\avaluation} 
\\
\modellabel{\Quorum\phi}{n,p,O,\avaluation} =& \bigvee_{O'\in\opensne} \modellabel{\phi}{n,p,O',\avaluation} 
\\
\modellabel{\texi\thea.\phi}{n,p,O,\avaluation} =& \bigvee_{v{\in}\Val} \modellabel{\phi[\thea\ssm v]}{n,p,O,\avaluation}
\\[2ex]
\modellabel{\yesterday\phi}{n,p,O,\avaluation} =&
\bigvee_{\f{max}(0,n\minus 1)\leq n'<n}\modellabel{\phi}{n',p,O,\avaluation}
\\
\modellabel{\tomorrow\phi}{n,p,O,\avaluation} =&
\modellabel{\phi}{n\plus 1,p,O,\avaluation}
\\[2ex]
\modellabel{\mu\tf X.\phi}{n,p,O,\avaluation}=&
\bigl(\bigwedge\{ \varkappa\in\THREE^{\tf{Ctx}}\mid \model{\phi[\tf X\ssm\varkappa]}\leq\varkappa\}\bigr)(n,p,O,\avaluation)
\figskip
\modellabel{\varkappa}{n,p,O,\avaluation}=&
\varkappa(n,p,O,\avaluation)
\end{array}
$$ 
\caption{Denotational semantics of \QLogicShort (Definition~\ref{defn.denotation})}
\label{fig.3.phi.f}
\end{figure}

\begin{figure} %
$$
\begin{array}{r@{\ }l@{\ \ }r@{\ }l}
\tvF =& \modT\tvB
&
\modellabel{\tvF}{\acontext} =&\tvF
\figskip
\tvT =& \tneg\tvF
&
\modellabel{\tvT}{\acontext} =&\tvT
\\[2ex]
\modF\phi =& \modT\tneg\phi
&
\modellabel{\modF\phi}{\acontext} =&\modF\modellabel{\phi}{\acontext}
\figskip
\modB\phi =& \tneg\modT\phi\tand\tneg\modF\phi
&
\modellabel{\modB\phi}{\acontext} =&\modB\modellabel{\phi}{\acontext}
\figskip
\modTB\phi =& \tneg\modT\tneg\phi
&
\modellabel{\modTB\phi}{\acontext} =&\modTB\modellabel{\phi}{\acontext}
\figskip
\modTF\phi =& \modT\phi\tor\modF\phi
&
\modellabel{\modTF\phi}{\acontext} =&\modTF\modellabel{\phi}{\acontext}
\figskip
\modFB\phi =& \modF\phi\tor\modB\phi
&
\modellabel{\modFB\phi}{\acontext} =&\modFB\modellabel{\phi}{\acontext}
\\[2ex]
\phi\tand\phi' =& \tneg(\tneg\phi\tor \tneg\phi')
&
\modellabel{\phi\tand\phi'}{\acontext} =&\modellabel{\phi}{\acontext}\tand\modellabel{\phi'}{\acontext}
\figskip
\phi\tnotor\phi' =& \tneg\phi\tor\phi'
&
\modellabel{\phi\tnotor\phi'}{\acontext} =&\modellabel{\phi}{\acontext}\tnotor\modellabel{\phi'}{\acontext}
\figskip
\phi\timpc\phi' =& \phi\tnotor\modT\phi'
&
\modellabel{\phi\timpc\phi'}{\acontext} =&\modellabel{\phi}{\acontext}\timpc\modellabel{\phi'}{\acontext}
\figskip
\phi\tequiv\phi' =& (\modT\phi\tand\modT\phi')\tor(\modB\phi\tand\modB\phi') \tor(\modF\phi\tand\modF\phi') \hspace{-5em}
\\
&&
\modellabel{\phi\tequiv\phi'}{\acontext} =&\modellabel{\phi}{\acontext}\tequiv\modellabel{\phi'}{\acontext}
\\[2ex]
\tall a.\phi =& \tneg\texi a.\tneg\phi
&
\modellabel{\tall a.\phi}{\acontext} =&\bigwedge_{v{\in}\Val}\modellabel{\phi[a\ssm v]}{\acontext}
\figskip
\everyone \phi =& \tneg\somewhere\tneg\phi
&
\modellabel{\everyone\phi}{n,p,O,\avaluation} =&\bigwedge_{p'{\in}O}\modellabel{\phi}{n,p',O,\avaluation}
\figskip
\Contraquorum\phi =& \tneg\Quorum\tneg\phi
&
\modellabel{\Contraquorum\phi}{n,p,O,\avaluation} =&\bigwedge_{O'{\in}\opensne}\modellabel{\phi}{n,p,O',\avaluation}
\\[2ex]
\someOallP\phi =& \Quorum\everyone\phi
&
\modellabel{\someOallP\phi}{n,p,O,\avaluation} =&\bigvee_{O'{\in}\opensne}\bigwedge_{p'\in O'}\modellabel{\phi}{n,p',O',\avaluation}
\figskip
\allOsomeP\phi =& \Contraquorum\somewhere\phi
&
\modellabel{\allOsomeP\phi}{n,p,O,\avaluation} =&\bigwedge_{O'{\in}\opensne}\bigvee_{p'\in O'}\modellabel{\phi}{n,p',O',\avaluation}
\figskip
\someOsomeP\phi =& \Quorum\somewhere\phi
&
\modellabel{\someOsomeP\phi}{n,p,O,\avaluation} =&\bigvee_{O'{\in}\opensne}\bigvee_{p'\in O'}\modellabel{\phi}{n,p',O',\avaluation}
\figskip
\allOallP\phi =& \Contraquorum\everyone\phi
&
\modellabel{\allOallP\phi}{n,p,O,\avaluation} =&\bigwedge_{O'{\in}\opensne}\bigwedge_{p'\in O'}\modellabel{\phi}{n,p',O',\avaluation}
\\[2ex]
\sometime\phi =& \mu\tf X.\tomorrow(\phi\tor \tf X)
&
\modellabel{\sometime\phi}{n,p,O,\avaluation} =&\bigvee_{n'> n}\modellabel{\phi}{n',p,O,\avaluation}
\figskip
\forever\phi =& \tneg\sometime\tneg\phi
&
\modellabel{\forever\phi}{n,p,O,\avaluation} =&\bigwedge_{n'> n}\modellabel{\phi}{n',p,O,\avaluation}
\figskip
\infinitely\phi =& \forever\,\sometime\,\phi
&
\modellabel{\infinitely\phi}{n,p,O,\avaluation} =&\bigwedge_{n'> n}\bigvee_{n''>n'}\modellabel{\phi}{n'',p,O,\avaluation}
\figskip
\final\phi =& \sometime\allOallP\forever\,\phi
\\
&\rlap{$\modellabel{\final\phi}{n,p,O,\avaluation} =\bigvee_{n'> n}\bigwedge_{O'{\in}\opensne,\ p'\in O'}\bigwedge_{n''>n'}\modellabel{\phi}{n'',p',O',\avaluation}$}
\figskip
\recent\phi =& \mu\tf X.\yesterday(\phi\tor\tf X)
&
\modellabel{\recent\phi}{n,p,O,\avaluation} =&\bigvee_{0\leq n'<n}\modellabel{\phi}{n',p,O,\avaluation}
\figskip
\urecent\phi =& \mu\tf X.(\phi\tor\yesterday\tf X)
&
\modellabel{\urecent\phi}{n,p,O,\avaluation} =&\bigvee_{0\leq n'\leq n}\modellabel{\phi}{n',p,O,\avaluation}
\end{array}
$$
$$
\begin{array}{r@{\ }l}
\mrup{a}{\phi}{v} =&
\mu\tf X.(\yesterday(\modT\phi[a\ssm v] \tor ((\tneg\modT\texi a.\phi) \tand \tf X)))
\figskip
\modellabel{\mrup{a}{\phi}{v}}{n,p,O,\avaluation} =&
\begin{cases}
\tvF 
&\Forall{0\leq n'<n} \modellabel{\texi a.\phi}{n',p,O,\avaluation}\neq\tvT
\\
\modT\modellabel{\phi[\thea\ssm v]}{\f{max},p,O,\avaluation}
&\Exists{0\leq n'<n} \modellabel{\texi a.\phi}{n',p,O,\avaluation}=\tvT
\end{cases}
\\
&\qquad\text{where}\quad
\f{max}=\bigvee\{0\leq n'<n \mid \modellabel{\texi a.\phi}{n',p,O,\avaluation}=\tvT\}
\end{array}
$$ 
\caption{Derived syntactic sugar expressions (Definition~\ref{defn.3.derived})}
\label{fig.3.derived}
\end{figure}

We are now ready to give a denotation to our syntax:
\begin{defn}[Denotation]
\label{defn.denotation}
Suppose $\Sigma=(\PredSymb,\arity)$ is a signature (Definition~\ref{defn.signature}(\ref{item.predicate.syntax.signature})), and $\mathcal M$ is a model (Definition~\ref{defn.model}). 
Suppose $\phi$ is a closed predicate (Definition~\ref{defn.predicate.syntax}(1)).
Define the \deffont{denotation} of $\phi$ 
$$
\model{\phi}\in\THREE^{\tf{Ctx}(\Sigma,\mathcal M)}
\quad\text{such that}\quad 
\acontext \longmapsto \modellabel{\phi}{\acontext}
$$
by the rules in Figure~\ref{fig.3.phi.f}.

We will usually treat $\Sigma$ and $\mathcal M$ as fixed parameters, so we may write $\tf{Ctx}(\Sigma,\mathcal M)$ just as $\tf{Ctx}$.
\end{defn}

\begin{rmrk}
Usually when we talk about denotation as defined in Definition~\ref{defn.denotation} above, the signature $\Sigma$ and model $\mathcal M$ are either understood or do not matter, so we may not mention them explicitly.
Meaning will always be clear.
\end{rmrk}

\begin{rmrk}[Initial comments on the denotation]
\label{rmrk.initial.comments.denotation}
We make some initial comments on Figure~\ref{fig.3.phi.f}:
\begin{enumerate*}
\item\label{item.first.mu}
$\mu\tf X.\phi$ acts as a \emph{fixpoint operator} where $\tf X$ is positive in $\phi$.\footnote{However, $\mu\tf X.\phi$ is well-formed syntax even without positivity, and we will make good use of this generality in axiom \rulefont{LdrUniq} in Figure~\ref{fig.ThyPaxOne}.  For more details on this, see Remark~\ref{rmrk.paxos.axioms.discussion}(\ref{item.discuss.LdrExistUniq}) and Subsection~\ref{subect.ldrunique}.}
Including a fixpoint operator is a standard method to include the power of recursion / iteration in a logic; an instance of $\tf X$ inside $\mu\tf X.\phi$, where $\tf X$ is positive in $\phi$, can be thought of as a recursive call back to $\phi$.
For such $\phi$,\ $\mu\tf X.\phi$ is what would be written $\mathrm{LFP}_{\tf X}\phi$ in~\cite[Chapter~4]{immerman:desc} (a Wikipedia article is also accessible and clear~\cite{wiki:Fixedpoint_logic}).
We will discuss this in Subsection~\ref{subsection.muX}.
\item
In the clause $\model{\tvB}=\tvB$ in Figure~\ref{fig.3.phi.f}, the $\tvB$ on the left is the syntax symbol from Figure~\ref{fig.predicate.syntax}, and the $\tvB$ on the right is $\tvB$-the-truth-value from Definition~\ref{defn.THREE}.
\item
In the clause for $\model{v\teq v'}$ in Figure~\ref{fig.3.phi.f}, the $\teq$ on the left is the syntax symbol from Figure~\ref{fig.predicate.syntax}, and the $\tvT$, $\tvF$, and $=$ on the right are truth-values and (actual) equality.

We really do mean $\model{v\teq v'}$ in this clause (equality of a pair of values $v,v'\in\Val$) rather than $\model{t\teq t'}$ (equality of a pair of terms) because we assume $\phi$ is closed, and the reader can easily check from Figure~\ref{fig.predicate.syntax} that the only closed terms are values. 
\item
Similarly for the clauses for $\model{\tneg\phi}$ and $\model{\phi\tor\phi'}$ and $\model{\modT\phi}$; the symbols $\tneg$, $\tor$, and $\modT$ on the left are syntax, and the symbols $\tneg$, $\tor$, and $\modT$ on the right refer to functions on truth-values as per Figure~\ref{fig.3}.
\item
$\somewhere$ resembles the usual modal operator `somewhere', but note that $\somewhere$ looks for somewhere \emph{in $O$}, where $O$ is the nonempty open set in the context.
\item
$\Quorum$ is a modality read `for some quorum' (where quorum = nonempty open set).
\item\label{item.interpret.yesterday}
The clause 
$\modellabel{\yesterday\phi}{n,p,O,\avaluation} = \bigvee_{\f{max}(0,n\minus 1)\leq n'<n}\modellabel{\phi}{n',p,O,\avaluation}$ is just a one-line encoding of
$$
\modellabel{\yesterday\phi}{n,p,O,\avaluation}=
\begin{cases}
\tvF & n=0
\\
\modellabel{\phi}{n\minus 1,p,O,\avaluation} & n>0 . 
\end{cases}
$$
\item
Intuitively, we see that $\tneg$ is anti-monotone.
We note further that it is very special in this respect: it is the only naturally \emph{non-monotone} operator in the primitive syntax from Figure~\ref{fig.predicate.syntax}.
Any non-monotonicity in the derived syntax from Figure~\ref{fig.3.derived}, comes from $\tneg$. 

This is deliberate and it is the design feature which powers Lemma~\ref{lemm.positive.monotonicity}.
\end{enumerate*} 
\QLogicShort is an expressive language and we will explore it with plenty of exposition.
Meaning will always be precise.
\end{rmrk}

\subsection{Derived expressions: simple programming in \QLogicShort}
\label{subsect.derived.expressions}

\begin{defn}[Derived expressions in Figure~\ref{fig.3.derived}]
\label{defn.3.derived}
The syntax of \QLogicShort in Figure~\ref{fig.predicate.syntax} and Definition~\ref{defn.predicate.syntax} is very expressive.
In the left-hand column of Figure~\ref{fig.3.derived} we define some derived expressions.
We can think of this as syntactic sugar, or as some simple programming in \QLogicShort.

For clarity and convenience we unpack the relevant denotation in the right-hand column.
\end{defn}

\begin{rmrk}[On the complexity of the derived expressions]
\label{rmrk.fig.3.easy}
Figure~\ref{fig.3.derived} has a lot in it.

However, we do not need to understand it all at once. 
We can think of \QLogicShort as the core of an abstract descriptive logic programming language; then Figure~\ref{fig.3.derived} collects the useful subroutines (in roughly increasing order of complexity). 
The reader can read and examine these directly, just like they might dip into the standard libraries of a programming language, or just wait to see how they are used.

In principle we could fully unpack our sugar down to the primitive syntax from Figure~\ref{fig.predicate.syntax}, but in practice we work directly with Figure~\ref{fig.3.derived} because this is more convenient.

For reference, further discussions of syntax are in Remarks~\ref{rmrk.initial.comments.denotation} and~\ref{rmrk.comments.on.semantics}.
\end{rmrk}

\begin{nttn}
\label{nttn.logic.terminology}
We take a moment to recall some standard conventions and terminology:
\begin{enumerate*}
\item
Implications $\tnotor$ and $\timpc$ associate to the right, so that $\phi\timpc\psi\timpc\chi$ means $\phi\timpc(\psi\timpc\chi)$.
\item
Unary connectives (like $\tneg$, $\modT$, and $\somewhere$) bind more strongly than binary connectives (like $\tand$ and $\tnotor$), so that $\somewhere\tneg\phi\tand\psi$ means $(\somewhere\tneg\phi)\tand\psi$.
\item
The scope of a quantification like $\texi a$ and $\tall a$ extends as far to the right as possible. 
Thus $\tall\thea.\phi\tand\psi$ means $\tall\thea.(\phi\tand\psi)$.
Note however that $\somewhere\texi\thea.\phi\tand\psi$ still means $(\somewhere(\texi\thea.\phi))\tand\psi$, because the $\texi$ quantifier is trapped in the scope of the $\somewhere$. 
\item
All the above said, we will (as is usual) use spacing and some common sense to convey scope.
For example 
$$
\texi\thea.\phi\mathbin{\tnotor}\psi
\quad\text{means}\quad
\texi\thea.(\phi\tnotor\psi), 
$$
and
$$
\texi\thea.\phi \ \tnotor \psi
\quad\text{means}\quad
(\texi\thea.\phi)\tnotor\psi. 
$$
In practice, meaning will always be clear.
\end{enumerate*}
\end{nttn}

\begin{rmrk}
The reader may not be surprised that quantifier ordering can matter: $\tall a.\texi b.a\teq b$ is completely different from $\texi b.\tall a.a\teq b$.
Modalities are also (weak) quantifiers, so the reader may not be surprised that $\everyone\Quorum$ (a universal modality followed by an existential one) is different from $\Quorum\everyone$ (existential followed by universal); the interested reader may enjoy unpacking these to check.

Less immediately obvious is that modality ordering can matter \emph{even between modalities with the same universal/existential} polarity, like $\everyone$ and $\Contraquorum$.

Consider this, for $p\in\Pnt$ and $O\in\opensne$:
$$
\begin{array}{r@{\ }l}
\modellabel{\allOallP\phi}{n,p,O,\avaluation} 
=\modellabel{\Contraquorum\everyone\phi}{n,p,O,\avaluation} 
=&\bigwedge_{O'{\in}\opensne}\bigwedge_{p'\in O'}\modellabel{\phi}{n,p',O',\avaluation}
\\
\modellabel{\everyone\Contraquorum\phi}{n,p,O,\avaluation} 
=&\bigwedge_{p'\in O}\bigwedge_{O'{\in}\opensne}\modellabel{\phi}{n,p',O',\avaluation}
\end{array}
$$
In the first case, the quantification is over every $O'\in\opensne$ and every $p'\in O'$.
In the second case, the quantification is over every $O'\in\opensne$ and every $p'\in O$; that is a completely different thing. 
\end{rmrk}

\subsection{The operator $\mu \tf X$}
\label{subsection.muX}

\begin{rmrk}
\label{rmrk.comments.on.semantics}
Our operator $\mu\tf X.\phi$ is a powerful construct.
As per Remark~\ref{rmrk.initial.comments.denotation}(\ref{item.first.mu}), 
via $\phi$ such that $\tf X$ is positive in $\phi$, $\mu$ inserts the power of recursion / iteration over $\phi$ in the logic, in the following sense:
a mention of $\tf X$ inside $\phi$ in $\mu\tf X.\phi$ can be thought of as a recursive call back to $\phi$, and we can think of $\mu\tf X.\phi$ as encoding an iterative loop that repeatedly executes a computation encoded in $\phi$, and loops if it encounters $\tf X$ in $\phi$.

(There is also a useful non-positive use of $\mu$ in our axioms; see Subsection~\ref{subect.ldrunique}. But in this Remark we focus on the positive, fixpoint uses.) 

With this intuition, we can look at some of the derived syntax in Figure~\ref{fig.3.derived} that uses $\mu$:
\begin{enumerate}
\item
$\recent\phi=\mu\tf X.\yesterday(\phi\tor\tf X)$ can be imagined as the following loop:
\begin{enumerate*}
\item
Set an accumulator to $\tvF$.\footnote{We mean `accumulator' in the sense of an `accumulator variable' inside a machine that (we imagine) is running this pseudocode.  This is not a technical term from logic.}
\item\label{recent.step.0}
If we are at day~0 then terminate and return the value of the accumulator; otherwise go to yesterday.
\item
Compute the truth-value of $\phi$, and $\tor$ this truth-value with the accumulator.
\item
Loop back to step~\ref{recent.step.0}.
\end{enumerate*} 
\item
The difference between $\recent\phi$ and $\urecent\phi$ is that $\recent\phi$ returns the greatest truth-value that $\phi$ has taken in the past, whereas $\urecent\phi$ returns the greatest truth-value that $\phi$ has taken today or in the past.

The notation $\urecent$ (`$\recent$', underlined) is a pun\footnote{A bad one perhaps, and the reader may or may not find this pun-acceptable.} on $<$ and $\leq$ (`$<$', underlined).
\item
$\sometime\phi=\mu\tf X.\tomorrow(\phi\tor\tf X)$ can be viewed as executing the following loop:
\begin{enumerate*}
\item
Set an accumulator to $\tvF$.
\item\label{sometime.step.0}
Go to tomorrow.
\item
Compute the truth-value of $\phi$ and $\tor$ with the accumulator.
\item
Loop back to step~\ref{sometime.step.0}.
\item
After infinite time (see Remark~\ref{rmrk.oo} about why that makes sense), return the limit value of the accumulator.
\end{enumerate*}
\item
$\forever\phi=\tneg\mu\tf X.\tomorrow(\tneg\phi\tor\tf X)$ can be viewed as the following loop:
\begin{enumerate*}
\item
Set an accumulator to $\tvT$.
\item\label{forever.step.0}
Go to tomorrow.
\item
Compute the truth-value of $\phi$, and $\tand$ this truth-value with the accumulator.
\item
Loop back to step~\ref{forever.step.0}.
\item
After infinite time return the limit value of the accumulator.
\end{enumerate*} 
\item
$\final\phi=\sometime\allOallP\forever\phi$ can be understood as follows: it iterates forward in time looking for a time when everywhere (= at every open set $O$, at every point $p\in O$) and forever, $\phi$ holds.
\item\label{item.mru.program}
The clause 
$\mrup{a}{\phi}{v} = \mu\tf X.(\yesterday(\modT\phi[a\ssm v] \tor ((\tneg\modT\texi a.\phi) \tand \tf X)))$
looks more imposing, but the idea is simply to 
\begin{itemize*}
\item
iterate back in time looking for when $\texi a.\phi$ has truth-value $\tvT$, i.e. when \emph{some} value makes $\phi(a)$ be true, and 
\item
if we find such a time then we check if \emph{at that time} $\phi[a\ssm v]$ has truth-value $\tvT$ and 
\item
only if both steps above succeed, do we return $\tvT$.
\end{itemize*}
Thus, we return $\tvT$ when $v$ is one of the values that \emph{most recently} make $\phi(a)$ be true.

In detail:
\begin{enumerate*}
\item\label{mu.step.0}
If we are at day~0, terminate with value $\tvF$; otherwise go to yesterday.
\item
If $\phi[a\ssm v]$ has truth-value $\tvT$ then terminate with value $\tvT$; otherwise continue to the next step.
\item
If $\texi a.\phi$ has truth-value $\tvT$ then terminate with value $\tvF$; otherwise go to step~\ref{mu.step.0}.
\end{enumerate*}
\end{enumerate}
\end{rmrk}

\begin{rmrk}[Comment on the omniscient observer]
\label{rmrk.oo}
The programs in Remark~\ref{rmrk.comments.on.semantics} above are executed by an omniscient observer who stands outside of the model's notion of time and space.
This is why (for example) $\forever$ can run an infinite loop --- the omniscient observer can observe it because the observer is looking at the model from outside.\footnote{Let's make this concrete.  Consider the natural numbers $\Ngeqz$.  We can easily observe this and note that $n\plus 1\gneq n$ for every $n\in\Ngeqz$; whereas a local observer trapped on the number line and constrained to perceive this line as the passage of time can make a similar observation, but from the point of view of that observer, the calculation involved is never-ending.  

Note that even if we could speak with this local observer, it would be no good to appeal to it to `just realise that it is on the number line', because in order to check this to a mathematical certainty the observer would have to \dots traverse all of its local notion of time to \emph{check}, amongst other things, that $n\plus 1\gneq n$.  After all, the observer might be on $\mathbb Z/p$ for some astronomically enormous $p$, i.e. their time might be on a large loop.  

The reader interested in a branch of mathematics explicitly devoted to studying these questions, and who is willing to possibly end up seriously doubting the nature of their own reality, can start by reading up on nonstandard models of arithmetic.} 
\end{rmrk}

\begin{rmrk}
The rest of this Subsection will show how our denotation for $\mu$ from Figure~\ref{fig.3.phi.f} does indeed yield a fixpoint operator.
The culminating result is Proposition~\ref{prop.mu.fixpoint} (which just says that $\mu$ is indeed a fixpoint operator for positive predicates) but the key technical lemma is Lemma~\ref{lemm.positive.monotonicity} (monotonicity of positive predicates).

The reader is welcome to read the rest of this Subsection, but this is all just a standard use of the Knaster-Tarski fixpoint theorem~\cite[Theorem~2.35, page~50]{priestley:intlo}.

In fact, the real design effort has already happened: we designed the derived syntax of Figure~\ref{fig.3.derived} so that it would respect the positivity condition of Definition~\ref{defn.positive}, without having to fully expand it out.
For instance, we defined $\modTB\phi$ to be $\tneg\modT\tneg\phi$ (which clearly puts $\phi$ in a positive position) instead of $\modT\phi\tor\modB\phi$ (which is logically equivalent, but which does not).
\end{rmrk}

\begin{rmrk}
\label{rmrk.leq.phi}
Our denotation $\model{\text{-}}$ from Definition~\ref{defn.denotation} and Figure~\ref{fig.3.phi.f} maps predicates to functions in $\THREE^{\tf{Ctx}}$.
One particular consequence of the order on $\THREE^{\tf{Ctx}}$ from Definition~\ref{defn.three.ctx.lattice} is that (denotations of) predicates are naturally ordered, and we can write 
$$
\model{\phi}\leq\model{\psi}
\quad\text{which just means}\quad
\Forall{\acontext\in\tf{Ctx}}\modellabel{\phi}{\acontext}\leq\modellabel{\psi}{\acontext}.
$$
We use this fact for example in Lemma~\ref{lemm.positive.monotonicity}.
We can also use this fact, and the complete lattice structure on $\THREE^{\tf{Ctx}}$ from Definition~\ref{defn.three.ctx.lattice}, to simplify the presentation of the clause for $\mu\tf X.\phi$ in Figure~\ref{fig.3.phi.f} as follows:
$$
\model{\mu\tf X.\phi}=
\bigwedge\{ \varkappa\in\THREE^{\tf{Ctx}}\mid \model{\phi[\tf X\ssm\varkappa]}\leq\varkappa\} .
$$
\end{rmrk}

\begin{defn}
\label{defn.positive}
Recall the predicate syntax from Figure~\ref{fig.predicate.syntax} %
and suppose $\phi$ is a predicate.
Then:
\begin{enumerate*}
\item\label{item.positive.occurrence}
Call a free occurrence of $\tf X$ in $\phi$ 
\begin{itemize*}
\item
\deffont{positive} when the occurrence appears under an even number (which may be $0$) of negations $\tneg$, and 
\item
\deffont{negative} when the occurrence appears under an odd number of negations.
\end{itemize*}
For example: all free occurrences of $\tf X$ are positive in $\tf X$ and $\tneg\tneg\tf X$ and $\mu\tf X.\tneg\tf X$ and $\yesterday\tf X$; and the free occurrence of $\tf X$ in $\tneg\tf X$ is negative. 
\item
Call a predicate $\phi$ \deffont{positive} when every free occurrence of $\tf X$ in $\phi$ is positive.
\end{enumerate*}
\end{defn}

\begin{lemm}
\label{lemm.positive.monotonicity}
Suppose $\mu\tf X.\phi$ is a closed predicate.
Then:
\begin{enumerate*}
\item
If $\phi$ is positive then $\varkappa\mapsto\model{\phi[\tf X\ssm\varkappa]}$ is monotone as a function from $\THREE^{\tf{Ctx}}$ with the partial order $\leq$ from Figure~\ref{fig.chi.3.lattice}, to itself.
In symbols, for every $\varkappa,\varkappa'\in\THREE^{\tf{Ctx}}$:
$$
\varkappa\leq\varkappa'
\quad\text{implies}\quad
\model{\phi[\tf X\ssm\varkappa]}\leq\model{\phi[\tf X\ssm\varkappa']} . 
$$
\item
If $\tneg\phi$ is positive then $\varkappa\mapsto\model{\phi[\tf X\ssm\varkappa]}$ is antimonotone as a function from $\THREE^{\tf{Ctx}}$ with the partial order $\leq$ from Figure~\ref{fig.chi.3.lattice}, to itself.
In symbols, for every $\varkappa,\varkappa'\in\THREE^{\tf{Ctx}}$:
$$
\varkappa\leq\varkappa'
\quad\text{implies}\quad
\model{\phi[\tf X\ssm\varkappa']}\leq\model{\phi[\tf X\ssm\varkappa]} . 
$$
\end{enumerate*}
(The meaning of $\model{\text{-}}\leq\model{\text{-}}$ is unpacked in Remark~\ref{rmrk.leq.phi}.)

\end{lemm}
\begin{proof}
Both parts follow by a routine simultaneous induction on the syntax of $\phi$, using the rules in Figure~\ref{fig.3.phi.f}.
\end{proof}

\begin{prop}
\label{prop.mu.fixpoint}
Suppose $\mu\tf X.\phi$ is a closed predicate and $\phi$ is positive.
Then:
\begin{enumerate*}
\item
$\model{\mu\tf X.\phi}$ is a least fixpoint of the function $\varphi:\THREE^{\tf{Ctx}}\to\THREE^{\tf{Ctx}}$ given by
$$
\varphi(\varkappa)=\lambda\acontext.\modellabel{\phi[\tf X\ssm\varkappa]}{\acontext}. 
$$
\item
If $\acontext$ is a context then
$$
\modellabel{\mu\tf X.\phi}{\acontext}=\modellabel{\phi[\tf X\ssm \model{\mu\tf X.\phi}]}{\acontext}.
$$
\end{enumerate*}
\end{prop}
\begin{proof}
By the Knaster-Tarski fixpoint theorem for $\THREE^{\tf{Ctx}}$ as a complete lattice as per Definition~\ref{defn.three.ctx.lattice}.
A concise but exceptionally clear presentation is in~\cite[Theorem~2.35, page~50]{priestley:intlo} (a Wikipedia article~\cite{wiki:Knaster-Tarski_theorem} is also accessible and clear). 
\end{proof}

\subsection{Validity}

We defined in Definition~\ref{defn.tv.ment} what it means for a truth-value to be \emph{valid}; now we extend this notion to predicates-in-context: 
\begin{defn}[Validity]
\label{defn.validity.judgement}
Suppose that:
\begin{itemize*}
\item
$\Sigma$ is a signature (Definition~\ref{defn.signature}(\ref{item.predicate.syntax.signature})).
We take $\Sigma$ to be fixed for the rest of this Definition, which means that we do not explicitly include it (e.g. as a subscript) in the notation.
\item
$\mathcal M$ is a model (Definition~\ref{defn.model}).
Again, we fix $\mathcal M$ as a parameter for the rest of this definition.
\end{itemize*}
Suppose further that:
\begin{itemize*}
\item
$\acontext=(n,p,O,\avaluation)$ is a context in $\tf{Ctx}(\Sigma,\mathcal M)$, so that by Definition~\ref{defn.contexts}(\ref{item.contexts}) $n\in\Time$ and $p\in\Pnt$ and $O\in\opensne$ and $\avaluation$ is a valuation to $\mathcal M$.
\item
$\phi$ is a closed predicate in $\Sigma$ and $\mathcal M$ (cf. Remark~\ref{rmrk.pedant.point}(\ref{item.model.matters})).
\end{itemize*}
Then we define: 
\begin{enumerate*}
\item\label{item.ctx.ment.phi}
$\acontext\ment\phi$ means $\ment\modellabel{\phi}{\acontext}$, which by Definition~\ref{defn.tv.ment}(\ref{item.tv.ment}) means $\modellabel{\phi}{\acontext}\in\threeValid$, which by Definition~\ref{defn.tv.ment}(\ref{item.tv.ment.correct}) means $\modellabel{\phi}{\acontext}=\tvT\lor\modellabel{\phi}{\acontext}=\tvB$.

Because this is a core definition, we sum this up for the reader's clear reference:
$$
\acontext\ment\phi
\quad\liff\quad
{\ment\modellabel{\phi}{\acontext}}
\quad\liff\quad
\modellabel{\phi}{\acontext}\in\threeValid
\quad\liff\quad
\modellabel{\phi}{\acontext}=\tvT\lor\modellabel{\phi}{\acontext}=\tvB.
$$
We read this as \deffont{$\phi$ is valid in the context $\acontext$.}
\item
$n,p\mentval\phi$ means $\Forall{O\in\opensne}(n,p,O\mentval\phi)$.
\item\label{item.ment.n}
$n\mentval\phi$ means $\Forall{p\in\Pnt}\Forall{O\in\opensne}(n,p,O\mentval\phi)$.
\item\label{item.mentval}
$\mentval\phi$ means $\Forall{n\in\Time}\Forall{p\in\Pnt}\Forall{O\in\opensne}(n,p,O\mentval\phi)$.
\item\label{item.ment.phi}
$\ment\phi$ means $\Forall{\acontext\in\tf{Ctx}(\Sigma,\mathcal M)}(\acontext\ment\phi)$.
\end{enumerate*}
\end{defn}

\begin{figure}
$$
\begin{array}{l@{\quad}l@{\quad}l}
n,p,O\mentval\phi &\text{means} & \modellabel{\phi}{n,p,O,\avaluation}\in\{\tvT,\tvB\} 
\figskip
\mentval\phi &\text{means} & \Forall{n\in\Time}\Forall{p\in\Pnt}\Forall{O\in\opensne}(n,p,O\mentval\phi)
\figskip
\mentval\Theta & \text{means} & \Forall{\theta\in\Theta}({\mentval\theta})
\figskip
\Theta\ment\phi & \text{means} &
\Forall{\avaluation\in\Valuation(\Sigma,\mathcal M)}({\mentval\Theta})\limp ({\mentval\phi}) 
\end{array}
$$
\caption{Selected clauses on validity of $\phi$ / validity of $\phi$ in a theory $\Theta$ (Definitions~\ref{defn.validity.judgement} and~\ref{defn.theory.axiom})}
\label{fig.validity}
\end{figure}

Later on, when we start to build declarative axiomatisations of consensus algorithms, we will express these as \emph{theories}.
We include the relevant definition here, since it follows on naturally from Definition~\ref{defn.validity.judgement}: 
\begin{defn}[Validity in a theory]
\label{defn.theory.axiom}
Suppose that:
\begin{itemize*}
\item
$\Sigma$ is a signature and $\mathcal M=(\Pnt,\opens,\Val)$ is a model.
\item
$\phi$ is a closed predicate in $\Sigma$ and $\mathcal M$.
\item
$\Theta$ is a (possibly infinite) set of closed predicates.\footnote{It would be unusual for $\Theta$, which is intended to be a collection of axioms, to mention specific elements $v\in\Val$ or $\varkappa\in\THREE^{\tf{Ctx}(\Sigma,\mathcal M)}$ from a model (cf. Remark~\ref{rmrk.pedant.point}), but the maths does not seem to mind, so we do not exclude this possibility.}
\end{itemize*}
Then:
\begin{enumerate*}
\item\label{item.theory}
A \deffont{theory} is a pair $(\Sigma,\Theta)$ of 
\begin{itemize*}
\item
a signature $\Sigma$, and 
\item
a set $\Theta$ of \deffont{axioms}, which are closed predicates over $\Sigma$ and $\mathcal M$.
\end{itemize*}
\item\label{item.axioms.valid}
Given a valuation $\avaluation\in\Valuation(\Sigma,\mathcal M)$ (Definition~\ref{defn.contexts}(\ref{item.contexts.valuation})), we define 
$$
\mentval\Theta
\quad\text{and}\quad
\Theta\ment\phi
$$ 
as per Figure~\ref{fig.validity}. 
When $\Theta\ment\phi$ holds, we say that $\phi$ is \deffont{valid in the theory $\Theta$} (in $\Sigma$ and $\mathcal M$). 
\end{enumerate*}
Examples of theories are in Figures~\ref{fig.ThyPaxOne} and~\ref{fig.ThySPax}.
\end{defn}

\begin{rmrk}
\label{rmrk.thy.quantifier.scope}
Note the scope of the quantifiers in Definition~\ref{defn.theory.axiom}.
The assertion $\Theta\ment\phi$ means 
$$
\textstyle\Forall{\avaluation}((\bigwedge_{\theta\in\Theta}\mentval\theta) \limp (\mentval\phi)) .
$$
In words: if every $\theta\in\Theta$ holds everywhere in the model under the valuation $\avaluation$, then $\phi$ holds everywhere in the model under the valuation $\avaluation$.
Spelling this out further we obtain:
$$
\Forall{\avaluation}\bigl(\Forall{\theta\in\Theta}\Forall{n,p,O}\,n,p,O\mentval\theta) \limp (\Forall{n,p,O}n,p,O\mentval\phi\bigr),
$$
where above $n\in\Time$, $p\in\Pnt$, and $O\in\opensne$.

This pattern of quantifiers is what makes $\Theta$ a theory, rather than just a set of assumptions: as a theory, $\Theta$ is a restriction on the valuation $\avaluation$.
\end{rmrk}

\begin{rmrk}
One final pedant point: 
in Definition~\ref{defn.theory.axiom}(\ref{item.theory}) we defined a theory to be a pair $(\Sigma,\Theta)$ of a signature and a set of closed predicates over a signature $\Sigma$ \emph{and} over a model $\mathcal M$.
This technical convenience lets us write axioms that directly reference specific values in the model (instead of introducing explicit constant symbols to our signature, and interpreting them as \dots values in the model).

We will exploit this to assume a distinguished `undefined' value $\udfn$ in axiom \rulefont{PaxPropose?} in Figure~\ref{fig.ThyPaxOne}, as per Definition~\ref{defn.ThyPaxOne}(\ref{item.ThyPaxOne.axioms}).
\end{rmrk}

\subsection{Comments on notation: eliding parameters, and functional notation}

\begin{nttn}[Eliding parameters]
\label{nttn.validity.irrelevant.context}
It will be quite common that we make validity judgements where it is clear that part of the context is irrelevant.
For example the reader can check the denotational semantics in Figure~\ref{fig.3.phi.f} and the definition of validity in Definition~\ref{defn.validity.judgement}(\ref{item.ctx.ment.phi}), to see that:
\begin{enumerate*}
\item
$\acontext \ment \tvB$ is valid; we do not need to know the values in $\acontext$. 
\item
The validity of $n,p,O\mentval \tf P(v)$ depends on $\avaluation$, $n$, and $p$ --- but not on $O$.
\item
The validity of $n,p,O\mentval\someOallP\phi$ depends on $n$ and $\avaluation$ (and of course on $\phi$), but it does not depend on the values of $p$ and $O$.
\end{enumerate*}
In such cases, we may omit parameters in the context if they do not matter --- we may even omit the entire context if appropriate.
Thus, we may write $\ment\tvT$, $\ment \thea\teq\thea$, and $n\mentval\someOallP\phi$.

This notation remains within the scope of Definition~\ref{defn.validity.judgement} --- which universally quantifies over any omitted parameter --- so omitting parameters in this way is not an abuse of notation.
However, if the reader sees us omit a parameter, there may be two reasons for it:
\begin{enumerate*}
\item
It may indicate a universal quantification over all possible values.
For example, this is clearly what we intend when we write $\mentval\phi$ in Figure~\ref{fig.validity} (second row).
\item
It may indicate that we have elided irrelevant parameters.
For example, this is clearly what we intend in the example of $\ment\tvT$ above.
\end{enumerate*}
In practice, it will always be clear which is intended and meaning will always be clear.\footnote{Actually, meaning will be \emph{clearer}.  

In early drafts of this paper, we dutifully included all parameters, whether relevant or not.  
But this turned out to be just too confusing: the notational overhead of carrying around irrelevant information was burdensome; but furthermore, it was actually easier to understand the sequents once we dropped irrelevant parameters, because this helped us to keep track of which parameters \emph{did} matter.}
\end{nttn}

\begin{figure}
$$
\begin{array}{l@{\quad}l@{\quad}l@{\quad}l@{\quad}l}
\texiunique  a.\tf{write}(a)
&
\someOallP\texi a.\tf{send}(a)
&
\texi a.\allOallP\tf{decide}(a)
&
\tall a.\modTF \tf{propose}(a)
&
\mrup{a}{(\somewhere\tf{accept}(a))}{v}
\\[1ex]
& \text{may be written as}
\\[1ex]
\texiunique \tf{write} 
&
\someOallP\texi\tf{send} 
&
\texi\allOallP\tf{decide}
&
\tall\modTF\tf{propose} 
&
\mru{(\somewhere\tf{accept})}{v}
\end{array}
$$
\caption{Abbreviations inspired by higher-order logic (Notation~\ref{nttn.hos})}
\label{fig.hos}
\end{figure}

\begin{nttn}[Functional notation]
\label{nttn.hos}
We will frequently apply quantifiers (like $\texi$, $\texiunique$, or $\tall$) to unary predicate symbols, possibly with modalities.
For example:
$$
\tall a.\someOallP\tf{decide}(a)
\quad\text{is}\quad\text{$\tall$ quantifying $a$ in\ $(\someOallP\tf{decide})$\ applied to $a$.}
$$ 
For clarity, we may use a \deffont{functional notation} inspired by higher-order logic (HOL) to elide the variable, as illustrated in Figure~\ref{fig.hos}.
Meaning will always be clear.

See also Remark~\ref{rmrk.functional.notation.in.paxos} for examples of this notation applied in an axiomatisation.
\end{nttn}

\section{Properties of predicates}
\label{sect.predicates}

This Section collects elementary technical results which we will need later.
Some of these results, and their proofs, are more technical than others, but they are not necessarily surprising if we look at what they mean in the model.
Mostly the statements are simple and it is obvious by common sense that they should work (simple is good, because it means greater productivity and fewer errors).

A typical example is Lemma~\ref{lemm.tomorrow.forever.recent}(\ref{item.tomorrow.forever.recent.1}): once we look at what it actually says about the model, this is that ``if $\phi$ is true today, then in the future it will have been true in the past''.
\emph{For saying this, a mathematician earns a living?} --- we might think.

What is less evident is \emph{which} elementary properties will be important later, nor is it obvious what definitions and structures make these properties all fit together nicely.
So the reader should note that the simplicity of this Section is a design feature, which was only attained in retrospect and after some work.

Broadly speaking, we start this Section with the most elementary results, and build complexity from there.
The reader is welcome to read this Section now, or to skip it and refer back to its results as they are used in later proofs.

\subsection{Equivalence of predicates, and some (de Morgan) equivalences}
\label{subsect.equivalence.of.predicates}

For the rest of this Section we make a fixed but arbitrary choice of signature $\Sigma$ (Definition~\ref{defn.signature}(\ref{item.predicate.syntax.signature})) and model $\mathcal M$ (Definition~\ref{defn.model}).
All predicates considered will be over $\Sigma$ and $\mathcal M$, unless otherwise stated.

We generalise equivalence $\tequiv$ (Definition~\ref{defn.equivalent}) from truth-values, to predicates viewed as functions from contexts to truth-values:
\begin{defn}[Equivalence of predicates]
\label{defn.phi.equivalent}
Suppose $\phi$ and $\phi'$ are closed predicates.

Call $\phi$ and $\phi'$ \deffont{extensionally equivalent} when 
$$
\model{\phi}=\model{\phi'} .
$$
Unpacking Definition~\ref{defn.denotation}, this means 
$$
\Forall{\acontext\in\tf{Ctx}(\Sigma,\mathcal M)}(\modellabel{\phi}{\acontext}=\modellabel{\phi'}{\acontext}) .
$$
For brevity, we may write `extensionally equivalent' just as `equivalent' henceforth.
\end{defn}

\begin{lemm}
\label{lemm.equiv}
Suppose $\phi$ and $\phi'$ are closed predicates and $\acontext$ is a context.
Then: 
\begin{enumerate*}
\item\label{item.equiv.1}
$\acontext\ment\phi\tequiv\phi'$ if and only if $\modellabel{\phi}{\acontext}=\modellabel{\phi'}{\acontext}$.
\item\label{item.equiv.2}
$\ment\phi\tequiv\phi'$ if and only if $\model{\phi}=\model{\phi'}$ (i.e. $\phi$ and $\phi'$ are equivalent as per Definition~\ref{defn.phi.equivalent}).
\end{enumerate*}
\end{lemm}
\begin{proof}
We consider each part in turn:
\begin{enumerate}
\item
By Definition~\ref{defn.validity.judgement}(\ref{item.ctx.ment.phi}) $\acontext\ment\phi\tequiv\phi'$ means $\modellabel{\phi\tequiv\phi'}{\acontext}\in\threeValid$, and by the clause for $\tequiv$ in Figure~\ref{fig.3.derived} this is if and only if $\modellabel{\phi}{\acontext}=\modellabel{\phi'}{\acontext}$.
\item
By Definition~\ref{defn.validity.judgement}(\ref{item.ment.phi}), $\ment\phi\tequiv\phi'$ means precisely $\Forall{\acontext}\acontext\ment\phi\tequiv\phi'$.
We use part~\ref{item.equiv.1} of this result.
\qedhere\end{enumerate}
\end{proof}

We will see many equivalences in what follows.
We start by using Definition~\ref{defn.phi.equivalent} and Lemma~\ref{lemm.equiv} to express some standard equivalences in logic, having to do with implication being expressible using negation and disjunction, and with \emph{de Morgan} dualities between modalities and connectives:
\begin{lemm}
\label{lemm.minimal.syntax}
For each of the equivalences $\f{LHS}\tequiv\f{RHS}$ listed in Figure~\ref{fig.easy.equivalences}, we have that $\f{LHS}$ is extensionally equal to $\f{RHS}$ in the sense of Definition~\ref{defn.phi.equivalent}.
\end{lemm}
\begin{proof}
By routine calculations using Figures~\ref{fig.3.phi.f} and~\ref{fig.3.derived} and facts of the lattice $\THREE=(\tvF<\tvB<\tvT)$.
\end{proof}

\begin{figure}
$$
\begin{array}{r@{\ }l@{\quad}r@{\ }l@{\quad}r@{\ }l@{\quad}r@{\ }l}
\phi\tand\phi' \tequiv& \tneg(\tneg\phi\tor\tneg\phi')
&
\phi\tor\phi' \tequiv& \tneg(\tneg\phi\tand\tneg\phi')
&
\tneg\phi\tequiv& \phi\timpc\tvF \tequiv\phi\tnotor\tvF %
\\
\phi\timpc\phi' \tequiv& \phi\tnotor\modT\phi'
&
\phi\tnotor\phi' \tequiv& \tneg\phi\tor\phi'
\\
\everyone\phi\tequiv&\tneg\somewhere\tneg\phi
&
\somewhere\phi\tequiv&\tneg\everyone\tneg\phi
\\
\Contraquorum\phi\tequiv&\tneg\Quorum\tneg\phi
&
\Quorum\phi\tequiv&\tneg\Contraquorum\tneg\phi
\\
\forever\phi\tequiv&\tneg\sometime\tneg\phi
&
\sometime\phi\tequiv&\tneg\forever\tneg\phi
\\
\tall a.\phi\tequiv&\tneg\texi a.\tneg\phi
&
\texi a.\phi\tequiv&\tneg\tall a.\tneg\phi
&
\modTF\phi\tequiv&\modT\phi\tor\modT\tneg\phi
\\
\modB\phi\tequiv&\tneg\modTF\phi
&
\modT\phi\tequiv&\tneg\modFB\phi
&
\modF\phi\tequiv&\tneg\modTB\phi
\end{array}
$$
\caption{Some easy equivalences}
\label{fig.easy.equivalences}
\end{figure}

\subsection{Two commutation lemmas}
\label{subsect.two.commutation.results}

We note some simple commutation lemmas.
These results are all easy facts of our model, but it only takes a few lines to spell them out:

\begin{lemm}[Commuting connectives]
\label{lemm.commuting.connectives}
Suppose $\phi$ and $\phi'$ are closed predicates and suppose
$$
\bullet\in\{\tor,\tand\} 
\ \ \text{and}\ \ 
\tf M\in\{\modT,\modTB\} . 
$$
Then
$$
\ment\tf M(\phi\bullet\phi')\tequiv (\tf M\phi\bullet\tf M\phi') .
$$
\end{lemm}
\begin{proof}
Routine from Lemma~\ref{lemm.equiv}(\ref{item.equiv.2}) and Figures~\ref{fig.3.phi.f} and~\ref{fig.3.derived}, and facts of $\THREE$ as a lattice.
\end{proof}

\begin{lemm}[Commuting modalities]
\label{lemm.someone.commutation}
Suppose $\mathsf Q\phi$ is a closed predicate for 
$$
\textstyle\mathsf Q\in\{\somewhere,\everyone,\Quorum,\Contraquorum,\recent,\urecent,\yesterday,\tomorrow,\forever,\sometime,\infinitely,\final\}
\cup
\bigcup_{x\in\tf{VarSymb}}\{\texi x,\tall x\}.
$$
Then 
$$
\ment\mathsf Q\modT\phi \tequiv \modT\mathsf Q\phi
\quad\text{and}\quad
\ment\mathsf Q\modTB\phi \tequiv \modTB\mathsf Q\phi.
$$
\end{lemm}
\begin{proof}
Routine from Lemma~\ref{lemm.equiv}(\ref{item.equiv.2}) and Figures~\ref{fig.3.phi.f} and~\ref{fig.3.derived}, and facts of $\THREE$ as a lattice.
\end{proof}

\subsection{Lattice inequalities}

Corollary~\ref{corr.recent.everyone.commute} is in the same spirit as Lemmas~\ref{lemm.commuting.connectives} and~\ref{lemm.someone.commutation}. 
In practice we usually just use it as a direct fact of the model, but we spell it out anyway.
It relies on this simple fact about lattices: 
\begin{lemm} 
\label{lemm.quantifier.swap}
Suppose $(\mathcal L,\tand,\tor)$ is a complete lattice and $X$ and $Y$ are sets and $f:(X\times Y)\to\mathcal L$ is any function. 
Then:
\begin{enumerate*}
\item\label{item.quantifier.swap.1}
$\bigvee_{x}\bigwedge_{y}f(x,y)\leq \bigwedge_{y}\bigvee_{x}f(x,y)$, 
\item\label{item.quantifier.swap.1b}
It is not in general the case that 
$\bigwedge_{y}\bigvee_{x}f(x,y) \leq \bigvee_{x}\bigwedge_{y}f(x,y)$.
\item\label{item.quantifier.swap.2}
$\bigvee_{x}\bigvee_{y}f(x,y) = \bigvee_{y}\bigvee_{x}f(x,y)$.
\end{enumerate*}
\end{lemm}
\begin{proof}
By easy lattice calculations. 
\end{proof}

\begin{corr}
\label{corr.recent.everyone.commute}
Suppose $\mathsf Q\phi$ is a closed predicate, where 
$$
\mathsf Q\in\{\texi a,\tall a \mid a{\in}\tf{VarSymb}\}\cup\{\somewhere,\everyone,\Quorum,\Contraquorum\},
$$
and suppose $\acontext$ is a context.
Then:
\begin{enumerate*}
\item\label{item.recent.everyone.commute.1}
$\modellabel{\recent\mathsf Q\phi}{\acontext}\leq \modellabel{\mathsf Q\recent\phi}{\acontext}$. 
\item
$\modellabel{\mathsf Q\recent\phi}{\acontext}\leq \modellabel{\recent\mathsf Q\phi}{\acontext}$ for 
$\mathsf Q\in\{\texi a\mid a{\in}\tf{VarSymb}\}\cup\{\somewhere,\Quorum\}$.
\item
It need not be the case that $\modellabel{\mathsf Q\recent\phi}{\acontext}\leq \modellabel{\recent\mathsf Q\phi}{\acontext}$ for 
$\mathsf Q\in\{\tall a\mid a{\in}\tf{VarSymb}\}\cup\{\everyone,\Contraquorum\}$.
\end{enumerate*}
\end{corr}
\begin{proof}
Suppose $n\in\Time$ and $p\in\Pnt$ and $O\in\opensne$ and $\avaluation$ is a valuation.
We freely use \strongmodusponens.
\begin{enumerate}
\item
We consider $\mathsf Q=\Contraquorum$.
We use the clauses for $\recent$ and $\Contraquorum$ in Figure~\ref{fig.3.derived}, %
and then Lemma~\ref{lemm.quantifier.swap}(\ref{item.quantifier.swap.1}): 
$$
\bigvee_{0{\leq}n'{<}n}\bigwedge_{O'{\in}\opensne}\modellabel{\phi}{n',p,O',\avaluation}
\stackrel{L\ref{lemm.quantifier.swap}(\ref{item.quantifier.swap.1})}{\leq}
\bigwedge_{O'{\in}\opensne}\bigvee_{0{\leq}n'{<}n}\modellabel{\phi}{n',p,O',\avaluation} .
$$
Now we consider $\mathsf Q=\tall a$:
$$
\bigvee_{0{\leq}n'{<}n}\bigwedge_{v{\in}\Val}\modellabel{\phi[a\ssm v]}{n',p,O,\avaluation}
\stackrel{L\ref{lemm.quantifier.swap}(\ref{item.quantifier.swap.1})}{\leq}
\bigwedge_{v{\in}\Val}\bigvee_{0{\leq}n'{<}n}\modellabel{\phi[a\ssm v]}{n',p,O,\avaluation} .
$$
The other cases are precisely similar.
\item
We consider just $\mathsf Q=\somewhere$; the other cases are precisely similar.
We just use Lemma~\ref{lemm.quantifier.swap}(\ref{item.quantifier.swap.2}): 
$$
\bigvee_{p'{\in}O}\bigvee_{0{\leq}n'{<}n}\modellabel{\phi}{n',p',O,\avaluation}
\stackrel{L\ref{lemm.quantifier.swap}(\ref{item.quantifier.swap.2})}{=}
\bigvee_{0{\leq}n'{<}n}\bigvee_{p'{\in}O}\modellabel{\phi}{n',p',O,\avaluation} .
$$
\item
We consider just $\mathsf Q=\tall a$; the other cases are precisely similar.
We note by Lemma~\ref{lemm.quantifier.swap}(\ref{item.quantifier.swap.1b}) that 
$$
\bigwedge_{v{\in}\Val}\bigvee_{0{\leq}n'_v{<}n}\modellabel{\phi[a\ssm v]}{n'_v,p,O,\avaluation}
\leq
\bigvee_{0{\leq}n'{<}n}\bigwedge_{v\in\Val}\modellabel{\phi[a\ssm v]}{n',p,O,\avaluation} 
\ \ \ \text{does not hold in general}.
$$
\qedhere\end{enumerate}
\end{proof}

\subsection{The compound modalities $\someOsomeP$, $\allOallP$, $\someOallP$, and $\allOsomeP$}
\label{subsect.compound.modalities}

\begin{rmrk}
\label{rmrk.compound.modalities}
We will make much use of the modalities $\someOsomeP$, $\allOallP$, $\someOallP$, and $\allOsomeP$.

We call these four combinations \deffont{compound modalities}, because they are built out of $\Quorum$ and $\Contraquorum$ (which quantify over open sets) and $\somewhere$ and $\everyone$ (which quantify over points in an open set). 
In this Subsection we study some of their properties.
\end{rmrk}

Lemma~\ref{lemm.dense.char} just unpacks definitions, where $\between$ is notation from~\cite{gabbay:semtad,gabbay:semdca} meaning `has nonempty intersection with':
\begin{lemm}
\label{lemm.dense.char}
Suppose $\phi$ is a closed predicate and $n\in\Time$ and $\avaluation$ is a valuation.
Then: 
\begin{enumerate*}
\item\label{item.dense.char.someoneAll}
$n\mentval\someOsomeP\phi$ if and only if $\Exists{O\in\opensne}O\between\{p\in\Pnt \mid n,p,O\mentval\phi\}$.
\item\label{item.dense.char.everyoneAll}
$n\mentval\allOallP\phi$ if and only if $\Forall{O\in\opensne}O\subseteq\{p\in\Pnt \mid n,p,O\mentval\phi\}$.
\item\label{item.dense.char.QuorumBox}
$n\mentval\someOallP\phi$ if and only if $\Exists{O\in\opensne}O\subseteq\{p\in\Pnt \mid n,p,O\mentval\phi\}$.
\item\label{item.dense.char.CoquorumDiamond}
$n\mentval\allOsomeP\phi$ if and only if $\Forall{O\in\opensne}O\between\{p\in\Pnt \mid n,p,O\mentval\phi\}$.
\end{enumerate*} 
\end{lemm}
\begin{proof}
We consider each part in turn:
\begin{enumerate}
\item
We reason as follows: 
$$
\begin{array}{r@{\ }l@{\quad}l}
n\mentval\someOsomeP\phi
\liff&
\Exists{O{\in}\opensne}\Exists{p{\in}O}n,p,O\mentval\phi
&\text{Figure~\ref{fig.3.derived}}
\\
\liff&
\Exists{O{\in}\opensne}O\between \{p\in\Pnt \mid n,p,O\mentval\phi\}
&\text{Fact}
\end{array}
$$
\item
We reason as follows: 
$$
\begin{array}{r@{\ }l@{\quad}l}
n\mentval\allOallP\phi
\liff&
\Forall{O{\in}\opensne}\Forall{p{\in}O}n,p,O\mentval\phi
&\text{Figure~\ref{fig.3.derived}}
\\
\liff&
\Forall{O{\in}\opensne}O\subseteq\{p\in\Pnt \mid n,p,O\mentval\phi\}
&\text{Fact}
\end{array}
$$
\item
We reason as follows:
$$
\begin{array}{r@{\ }l@{\quad}l}
n\mentval\someOallP\phi
\liff&
\Exists{O{\in}\opensne}\Forall{p{\in}O}n,p,O\mentval\phi
&\text{Figure~\ref{fig.3.derived}}
\\
\liff&
\Exists{O{\in}\opensne}O\subseteq\{p\in\Pnt \mid n,p,O\mentval\phi\}
&\text{Fact}
\end{array}
$$
\item
We reason as follows:
$$
\begin{array}[b]{r@{\ }l@{\quad}l}
n\mentval\allOsomeP\phi
\liff&
\Forall{O{\in}\opensne}\Exists{p{\in}O}n,p,O\mentval\phi
&\text{Figure~\ref{fig.3.derived}}
\\
\liff&
\Forall{O{\in}\opensne}O\between\{p\in\Pnt \mid n,p,O\mentval\phi\}
&\text{Fact}
\end{array}
\qedhere$$
\end{enumerate}
\end{proof}

Proposition~\ref{prop.someone.implies.someoneAll} packages up some simple corollaries of Lemma~\ref{lemm.dense.char}, in forms that will be useful for our later proofs: 
\begin{prop}
\label{prop.someone.implies.someoneAll}
Suppose $\phi$ and $\psi$ are closed predicates and $\acontext$ is a context.
Then:
\begin{enumerate*}
\item\label{item.someone.implies.someoneAll.4}
$\acontext\ment\someOallP\phi\tand\allOsomeP\psi\limp\acontext\ment\someOsomeP(\phi\tand\psi)$.
\item\label{item.someone.implies.someoneAll.4b}
$\acontext\ment\someOallP\phi\tand\allOallP\psi\limp\acontext\ment\someOallP(\phi\tand\psi)$.
\item\label{item.someone.implies.someoneAll.5}
$\acontext\ment\modT(\someOallP\phi\tand\allOsomeP\psi)\limp\acontext\ment\modT\someOsomeP(\phi\tand\psi)$,\ and
\\
$\acontext\ment(\someOallP\phi\tand\allOsomeP\psi)\timpc \someOsomeP(\phi\tand\psi)$.
\item\label{item.someone.implies.someoneAll.2}
$\acontext\ment\modT\someOallP\phi\limp\acontext\ment\modT\someOsomeP\phi$.
\item\label{item.someone.implies.someoneAll.2b}
$\acontext\ment\modT\allOsomeP\phi\limp\acontext\ment\modT\someOsomeP\phi$.
\end{enumerate*}
\end{prop}
\begin{proof}
We consider each part in turn:
\begin{enumerate}
\item
Routine from Lemma~\ref{lemm.dense.char}(\ref{item.dense.char.QuorumBox}\&\ref{item.dense.char.CoquorumDiamond}).
\item
Routine from Lemma~\ref{lemm.dense.char}(\ref{item.dense.char.QuorumBox}\&\ref{item.dense.char.everyoneAll}).
\item
From part~\ref{item.someone.implies.someoneAll.4} of this result taking $\modT\phi$ and $\modT\psi$ and the commutation Lemmas~\ref{lemm.someone.commutation} and~\ref{lemm.commuting.connectives}.
Then we use \strongmodusponens.
\item
From part~\ref{item.someone.implies.someoneAll.5} of this result, taking $\psi=\tvT$.
\item
From part~\ref{item.someone.implies.someoneAll.5} of this result, taking $\phi=\tvT$.
\qedhere\end{enumerate}
\end{proof}

\subsection{Some comments on time}

\begin{rmrk}
This Subsection notes Lemmas~\ref{lemm.tomorrow.forever.recent} and~\ref{lemm.persist.char}.
They both have to do with time, and underneath all their complicated-looking symbols they are both straightforward:
\begin{itemize*}
\item
Lemma~\ref{lemm.tomorrow.forever.recent} reflects a fact that if something holds now, then in the future it will have held.
\item
Lemma~\ref{lemm.persist.char} reflects the fact that something repeats forever if and only if it happens infinitely often.
\end{itemize*}
\end{rmrk}

\begin{lemm}
\label{lemm.tomorrow.forever.recent}
Suppose $\acontext=(n,p,O,\avaluation)$ is a context and $\phi$ is a closed predicate.
Then:
\begin{enumerate*}
\item\label{item.tomorrow.forever.recent.1}
$n,p,O\mentval\modT\phi$ implies $\Forall{n'{\gneq}n}\,n',p,O\mentval\modT\recent\phi$.
\item\label{item.rur.r}
$n,p,O\mentval\modT\urecent\phi$ implies $\Forall{n'{\geq} n}\,n',p,O\mentval\modT\urecent\phi$.
\item
As corollaries, 
$$
\acontext\ment\phi\timpc \forever\recent\phi
\quad\text{and}\quad
\acontext\ment\urecent\phi\timpc\forever\urecent\phi.
$$
\end{enumerate*}
\end{lemm}
\begin{proof}
This is just from the clauses for $\recent$ and $\urecent$ in Figure~\ref{fig.3.derived}. 
The corollaries follow using \strongmodusponens.
\end{proof}

$\infinitely=\forever\sometime$ has a characterisation as `infinitely often', with the convenient feature that its denotation is independent of the stage at which it is evaluated (called $n'$ in the statement of the Lemma below):
\begin{lemm}
\label{lemm.persist.char}
Suppose $\phi$ is a closed predicate and $p\in\Pnt$ and $O\in\opensne$ and $\avaluation$ is a valuation.
Suppose further that $n'\in\Time$ and $\f{tvs}\in\{\threeTrue,\threeValid\}$.
Then we have
$$
\modellabel{\infinitely\phi}{n',p,O,\avaluation}\in\f{tvs}
\quad\liff\quad
\{n\in\Time \mid \modellabel{\phi}{n,p,O,\avaluation}\in\f{tvs}\}\text{ is infinite} . 
$$
\end{lemm}
\begin{proof}
By routine arguments from the clause for $\infinitely$ in Figure~\ref{fig.3.derived}.
\end{proof}

\subsection{Unique existence and affine existence in three-valued logic}
\label{subsect.unique.affine.existence}

\subsubsection{The definition, and main result}

\begin{rmrk}
Distributed algorithms often require an implementation to make (at most) one choice from (zero or more) possibilities.
In our axioms, we model this using unique existence and affine existence:
\begin{itemize*}
\item
\emph{Unique existence} means `there exists precisely one'; we write this $\exists_{1}$, for uniformity with affine existence, defined in the next item.
(The standard symbol is $\exists!$.)\footnote{There is another reason to not write $\exists!$ for unique existence: to avoid a mild notation clash with our forward rules, which will have names mentioning $!$ (for `forward') and $01$ (for affine), and `Uniq' (for uniqueness).  In context, writing $\exists_{1}$ and $\exists_{01}$ seems the cleaner and more systematic option.} 
\item
\emph{Affine existence} means `there exist zero or one'; we write this $\exists_{01}$.
(We are not aware of a standard symbol for affine existence.)
\end{itemize*}
See for example \rulefont{PaxPropose01} and \rulefont{PaxWrite01} in Figure~\ref{fig.logical.paxos} (a discussion of which is in Remark~\ref{rmrk.paxos.axioms.discussion}(\ref{item.01.discussion})).

So we need to think about what $\exists_{1}$ and $\exists_{01}$ even mean in a three-valued setting.
This turns out to be not completely straightforward, because we have a third truth-value that is neither true nor false.

For concreteness, suppose $\phi$ is a predicate with one free variable $a$, and suppose $\acontext$ is a context, and suppose that 
$\modellabel{\phi[a\ssm v]}{\acontext}=\tvB$ for all $v\in\Val$.
Then what truth-values should $\modellabel{\texiaffine a.\phi}{\acontext}$ and $\modellabel{\texiunique a.\phi}{\acontext}$ have?
 
Now suppose $v,v',v''\in\Val$, and suppose that $\modellabel{\phi[a\ssm v''']}{\acontext}=\tvF$ for all $v'''\in\Val\setminus\{v,v',v''\}$ and
$$
\modellabel{\phi[a\ssm v]}{\acontext}=\tvT
\quad\land\quad
\modellabel{\phi[a\ssm v']}{\acontext}=\tvB
\quad\land\quad
\modellabel{\phi[a\ssm v'']}{\acontext}=\tvB .
$$
\noindent Then what truth-values should $\modellabel{\texiaffine a.\phi}{\acontext}$ and $\modellabel{\texiunique a.\phi}{\acontext}$ have?

Now suppose that $\modellabel{\phi[a\ssm v''']}{\acontext}=\tvF$ for all $v'''\in\Val\setminus\{v\}$ and
$$
\modellabel{\phi[a\ssm v]}{\acontext}=\tvB .
$$
\noindent Then what truth-values should $\modellabel{\texiaffine a.\phi}{\acontext}$ and $\modellabel{\texiunique a.\phi}{\acontext}$ have?

The answer is not completely automatic: it is a real design decision that we have to make.
\end{rmrk}

\begin{defn}
\label{defn.unique.exist}
Suppose $\phi$ is a predicate.
Define \deffont{affine existence} $\texiaffine a.\phi$ and \deffont{unique existence} $\texiunique  a.\phi$ by:
$$
\begin{array}{r@{\ }l}
\texiaffine a.\phi =& \tall a,a'.(\phi\tand\phi[a\ssm a']) \tnotor a\teq a'
\\
\texiunique  a.\phi =& (\texi a.\phi) \tand (\texiaffine a.\phi)
\end{array}
$$
Above, $a'\in\tf{VarSymb}$ is a variable symbol chosen fresh for $\phi$, and $[a\ssm a']$ denotes the replacement of every instance of $a\in\tf{VarSymb}$ with $a'$.\footnote{Secretly, $[a\ssm a']$ is a nominal swapping permutation $(a'\ a)$~\cite{gabbay:newaas-jv}.} 
\end{defn}

\begin{rmrk}
\label{rmrk.unpack.texiunique}
Suppose $\texi a.\phi$ is a closed predicate (i.e. $a\in\tf{VarSymb}$ and $\phi$ is a predicate with free variable symbols at most $\{a\}$), and suppose $\acontext$ is a context.
We unpack the denotations of $\texiunique$ and $\texiaffine$ from Definition~\ref{defn.unique.exist}: 
\begin{enumerate*}
\item
If $\exists v,v'\in\Val.v\neq v'\land\modellabel{\phi[a\ssm v]}{\acontext}=\modellabel{\phi[a\ssm v']}{\acontext}=\tvT$, then 
$$
\modellabel{\texiaffine a.\phi}{\acontext}
=\tvF=
\modellabel{\texiunique a.\phi}{\acontext}
.
$$
\item
If $\exists_{1}v\in\Val.\modellabel{\phi[a\ssm v]}{\acontext}=\tvT$ and $\Forall{v\in\Val}\modellabel{\phi[a\ssm v]}{\acontext}\in\{\tvT,\tvF\}$, then 
$$
\modellabel{\texiaffine a.\phi}{\acontext}
=\tvT=
\modellabel{\texiunique a.\phi}{\acontext}
.
$$
\item
If $\exists_{1}v\in\Val.\modellabel{\phi[a\ssm v]}{\acontext}=\tvT$ and $\Exists{v\in\Val}\modellabel{\phi[a\ssm v]}{\acontext}=\tvB$, then 
$$
\modellabel{\texiaffine a.\phi}{\acontext}
=\tvB=
\modellabel{\texiunique a.\phi}{\acontext}
.
$$
\item\label{item.explain.existsunique.4}
If $\neg\Exists{v\in\Val}\modellabel{\phi[a\ssm v]}{\acontext}=\tvT$ and $\exists_{1} v\in\Val.\modellabel{\phi[a\ssm v]}{\acontext}=\tvB$, then 
$$
\modellabel{\texiaffine a.\phi}{\acontext}=\tvT
\quad\text{and}\quad
\modellabel{\texiunique a.\phi}{\acontext}=\tvB
. 
$$
\item\label{item.two.b}
If $\neg\Exists{v\in\Val}\modellabel{\phi[a\ssm v]}{\acontext}=\tvT$ and $\exists v,v'\in\Val.v\neq v'\land \modellabel{\phi[a\ssm v]}{\acontext}=\modellabel{\phi[a\ssm v']}{\acontext}=\tvB$, then 
$$
\modellabel{\texiaffine a.\phi}{\acontext} 
=\tvB=
\modellabel{\texiunique a.\phi}{\acontext}
.
$$
\item
If $\Forall{v\in\Val}\modellabel{\phi[a\ssm v]}{\acontext}=\tvF$, then
$$
\modellabel{\texiaffine a.\phi}{\acontext}=\tvT
\quad\text{and}\quad
\modellabel{\texiunique a.\phi}{\acontext}=\tvF
.
$$
\end{enumerate*}
\end{rmrk}

\begin{rmrk}
\label{rmrk.discuss.texiunique}
We make two elementary observations on Remark~\ref{rmrk.unpack.texiunique}:
\begin{enumerate} 
\item
The equality $\teq$ in our logic is two-valued: $\model{v\teq v'}$ returns $\tvT$ or $\tvF$, and never $\tvB$.\footnote{In the terminology of Definition~\ref{defn.correct}(\ref{item.correct.phi}) below, equality predicates are always \emph{correct}.}
\item
Two notions of `unique existence' are in play in this discussion:
\begin{itemize*}
\item
`$\exists_{1}v\in\Val$' refers to the actual unique existence quantifier in the real world (whatever that is).
This computes actual truth-values `true' and `false' and returns true when there actually exists a unique value.
\item
`$\texiunique a$' is part of the (derived) syntax of \QLogicShort, and is defined in Definition~\ref{defn.unique.exist}.
Its denotation computes elements of $\THREE=\{\tvF,\tvB,\tvT\}$.  
As is standard (though confusingly) we \emph{also} call elements of $\THREE$ truth-values; but they are truth-values of \QLogicShort, rather than actual truth-values in the real world (whatever that is).
\end{itemize*}
So: $\exists_{1}$ and $\texiunique$ are both unique existence operators, but they live at different levels and they are not the same thing.\footnote{This kind of thing is far from unusual.  A book on formal logic might use a notion of implication and existence in its proofs, which is related to but distinct from implication and existential predicate-formers in the formal logic under consideration.  Nor is this limited to logic: a programming language may have multiple versions of the number $1$ (8-bit, floating point, 64-bit, and infinite precision) all of which are called `one', but all of which are different from each other, and also different from the actual $1\in\mathbb N$ (whatever that is).}
\end{enumerate}
\end{rmrk}

We will use unique/affine existence in our proofs via Lemma~\ref{lemm.unique.affine.existence}: 
\begin{lemm}
\label{lemm.unique.affine.existence}
Suppose $\acontext$ is a context and $\texi a.\phi$ is a closed predicate and $v,v'\in\Val$ are values.
Then:
\begin{enumerate*}
\item\label{item.unique.affine.existence.01}
$\acontext\ment\texiaffine a.\phi$ if and only if 
$\exists_{01} v\in\Val.\modellabel{\phi[a\ssm v]}{\acontext}=\tvT$.
\item\label{item.unique.affine.existence.11}
$\acontext\ment\texiunique a.\phi$ if and only if 
$\exists v\in\Val.\modellabel{\phi[a\ssm v]}{\acontext}\in\threeValid \land \exists_{01} v\in\Val.\modellabel{\phi[a\ssm v]}{\acontext}=\tvT$.
\item\label{item.unique.affine.existence.01.implies}
If $\acontext\ment\texiaffine a.\phi$ and $\acontext\ment\modT\phi[a\ssm v]$ and $\acontext\ment\modT\phi[a\ssm v']$, then $v=v'$.
\item\label{item.unique.affine.existence.11.implies}
If $\acontext\ment\texiunique a.\phi$ and $\acontext\ment\modT\phi[a\ssm v]$ and $\acontext\ment\modT\phi[a\ssm v']$, then $v=v'$.
\end{enumerate*}
\end{lemm}
\begin{proof}
We unpack Definition~\ref{defn.validity.judgement}(\ref{item.ctx.ment.phi}) %
and check the possibilities in Remark~\ref{rmrk.unpack.texiunique}.
\end{proof}

\subsubsection{Comments on an edge case}

Lemma~\ref{lemm.texiunique.all.b} notes a (perhaps surprising, at first glance) fact about our denotation, that 
if a unary predicate symbol $\tf P$ takes truth-value $\tvB$ when applied to all values (as would model a participant that has entirely crashed), then $\texiunique\tf P$ and $\texiaffine\tf P$ are valid --- even though $\tf P(v)$ is valid on every value.
This reflects the common side-condition in the literature `if the participant is correct, then \emph{stuff}', which is trivially true if the participant is \emph{not} correct: 
\begin{lemm}
\label{lemm.texiunique.all.b}
Suppose $\acontext$ is a context and $\tf P$ is a unary predicate symbol.\footnote{There is nothing mathematically special about using a \emph{unary} predicate symbol here, but in practice this is the case that we will use.}
Then
$$
\acontext\ment \tall\modB\tf P \limp (\acontext\ment \texiunique \tf P \land \acontext\ment \texiaffine \tf P).  
$$
The above statement uses Notation~\ref{nttn.hos}.
Unpacking this, we obtain:
$$
\acontext\ment \tall a.\modB\tf P(a) \limp (\acontext\ment \texiunique a.\tf P(a)\land \acontext\ment \texiaffine a.\tf P(a)) . 
$$  
\end{lemm}
\begin{proof}
Suppose $\acontext\ment \tall\modB\tf P$.
Unpacking Definition~\ref{defn.validity.judgement}(\ref{item.ctx.ment.phi}) this means $\modellabel{\tall\modB\tf P}{\acontext}\in\threeValid$, and by Figure~\ref{fig.3.derived} it follows that $\Forall{v\in\Val}\modellabel{\tf P(v)}{\acontext}=\tvB$.

By 
Remark~\ref{rmrk.unpack.texiunique}(\ref{item.explain.existsunique.4}) (if $\Val=\{v\}$ for some $v$) 
or Remark~\ref{rmrk.unpack.texiunique}(\ref{item.two.b}) (if there exist $v\neq v'\in \Val$) 
it follows that 
$\modellabel{\texiaffine\tf P}{\acontext}=\tvT$ and $\modellabel{\texiunique\tf P}{\acontext}=\tvB$,\  
or 
$\modellabel{\texiaffine\tf P}{\acontext}=\tvB=\modellabel{\texiunique\tf P}{\acontext}$, 
respectively.

Both $\tvT$ and $\tvB$ are \emph{valid} by Definition~\ref{defn.validity.judgement}(\ref{item.ctx.ment.phi}), and thus $\acontext\ment\texiaffine\tf P$ and $\acontext\ment\texiunique\tf P$ as required. 
\end{proof}

\subsection{Pointwise predicates (a well-behavedness property)}

For this subsection we fix a signature $\Sigma$ and a model $\mathcal M=(\Pnt,\opens,\Val)$.

\subsubsection{The definition, and basic properties}

\begin{rmrk}
\emph{Being pointwise} is a well-behavedness property on predicates, that their denotation in a context $\acontext=(n,p,O,\avaluation)$ %
does not depend on the open set parameter $O$ of the context.

Many important syntactic classes of predicates are naturally pointwise, including any predicate of the form $\tf P(v_1,\dots,v_n)$ (Lemma~\ref{lemm.atomic.O}).

Being pointwise is desirable for us because it simplifies denotation to holding `at a point';\footnote{Actually: at a stage, point, and valuation, but intuitively we tend to view the stage and valuation as fixed, and the point $p$ and open set $O$ as varying.} and also because several nice properties depend on this (see for example Proposition~\ref{prop.pointwise.dense.char}); and finally also because there is an easy syntactic condition (given in Lemma~\ref{lemm.atomic.O}(\ref{item.atomic.pointwise.2})) that guarantees that a predicate is pointwise.
\end{rmrk}

\begin{defn}[Pointwise predicates]
\label{defn.pointwise.meaning}
Suppose $\phi$ is a closed predicate.
\begin{enumerate*}
\item
If $n\in\Time$ and $\avaluation\in\Valuation(\Sigma,\mathcal M)$ is a valuation then call $\phi$ \deffont{pointwise in $(n,\avaluation)$} when
$$
\Forall{p\in\Pnt}\Forall{O,O'{\in}\opensne}\modellabel{\phi}{n,p,O,\avaluation}=\modellabel{\phi}{n,p,O',\avaluation} .
$$ 
In words, $\phi$ is pointwise in $(n,\avaluation)$ when at every point $p\in\Pnt$, the open set parameter $O\in\opensne$ does not affect the value of the denotation $\modellabel{\phi}{n,p,O,\avaluation}$. 
\item
Call $\phi$ \deffont{absolutely pointwise} when it is pointwise for every $n$ and $\avaluation$.
\end{enumerate*}
\end{defn}

\begin{defn}
\label{defn.pointwise.pred.nttn}
Suppose $\phi$ is closed, and $n\in\Time$ and $p\in\Pnt$ and $\avaluation$ is a valuation.
Then define $\modellabel{\phi}{n,p,\avaluation}$ by
$$
\modellabel{\phi}{n,p,\avaluation}=\bigwedge_{O\in\opensne}\modellabel{\phi}{n,p,O,\avaluation}.
$$
\end{defn}

In a manner of speaking we have already seen $\modellabel{\phi}{n,p,\avaluation}$ from Definition~\ref{defn.pointwise.pred.nttn} in our logic:
\begin{lemm}
\label{lemm.Coquorum.denote}
Suppose $\phi$ is closed and $n\in\Time$ and $p\in\Pnt$ and $O\in\opensne$ and $\avaluation$ is a valuation.
Then
$$
\modellabel{\phi}{n,p,\avaluation}=\modellabel{\Contraquorum\phi}{n,p,O,\avaluation}.
$$
\end{lemm}
\begin{proof}
Direct from the clause for $\Contraquorum$ from Figure~\ref{fig.3.derived}.
\end{proof}

\begin{lemm}
\label{lemm.pointwise.pred}
Suppose $\phi$ is a closed predicate and $n\in\Time$ and $\avaluation$ is a valuation.
Then $\phi$ is pointwise in $(n,\avaluation)$ in the sense of Definition~\ref{defn.pointwise.meaning} 
if and only if
$$
\Forall{p\in\Pnt}\Forall{O\in\opensne}\modellabel{\phi}{n,p,O,\avaluation}=\modellabel{\phi}{n,p,\avaluation}.
$$
\end{lemm}
\begin{proof}
By a routine calculation from Definitions~\ref{defn.pointwise.meaning} and~\ref{defn.pointwise.pred.nttn}.
\end{proof}

\begin{corr}
\label{corr.pointwise.pred}
Suppose $\phi$ is a closed predicate and $n\in\Time$ and $p\in\Pnt$ and $\avaluation$ is a valuation.
Suppose $\phi$ is pointwise in $(n,\avaluation)$.
Then
$$
\Forall{O\in\opensne}n,p,O\mentval\phi\liff n,p\mentval\phi. 
$$
\end{corr}
\begin{proof}
Routine from Definition~\ref{defn.validity.judgement} and Lemma~\ref{lemm.pointwise.pred}. %
\end{proof}

Being pointwise is particularly nice because of two simple syntactic criteria that guarantee it:
\begin{lemm}[Conditions for being pointwise]
\label{lemm.atomic.O}
Suppose $\phi$ is a closed predicate.
Then:
\begin{enumerate*}
\item\label{item.atomic.pointwise}
If $\phi$ is atomic (Notation~\ref{nttn.nullary}(\ref{item.atomic.predicate}); having the form $\tf P(v_1,\dots,v_n)$ or $\tneg\tf P(v_1,\dots,v_n)$) then it is absolutely pointwise (Definition~\ref{defn.pointwise.meaning}).
\item\label{item.atomic.pointwise.2}
If $\phi$ %
does not mention $\somewhere$ or $\everyone$ or $\varkappa$ or $\mu$ except within the scope of a $\Quorum$ or $\Contraquorum$ modality, then $\phi$ is absolutely pointwise.\footnote{In particular, $\somewhere$ and $\everyone$ are within the scope of $\Quorum$ and $\Contraquorum$ in the compound modalities $\someOsomeP$, $\allOallP$, $\someOallP$, and $\allOsomeP$.}
\end{enumerate*} 
\end{lemm}
\begin{proof}
For part~\ref{item.atomic.pointwise}, we just unpack the definition of $\modellabel{\phi}{}$ in Figure~\ref{fig.3.phi.f} and note that the `$\avaluation(\tf P)_{n,p}$' used in the clause for a predicate symbol $\tf P$ depends just on the values of $n$, $p$, and $\avaluation$, but not on the open sets $O$ and $O'$. 

Part~\ref{item.atomic.pointwise.2} follows by a routine induction on predicate syntax, using part~\ref{item.atomic.pointwise} as the base case.\footnote{This result could be generalised to permit $\mu$ not in the scope of $\Quorum$ or $\Contraquorum$, subject to suitable positivity conditions (Definition~\ref{defn.positive}).  But we will not need the generality so we do not bother.}
\end{proof}

\subsubsection{Pointwise, as a modality in our logic}

In fact, we can express the property of being pointwise within our logic:
\begin{defn}
\label{defn.pointwise.pred}
Suppose $\phi$ is a predicate.
Then define a predicate $\pointwise{\phi}$ by
$$
\pointwise{\phi} \ =\ \allOallP(\Quorum\phi \tequiv \Contraquorum\phi) .
$$
\end{defn}

\begin{lemm}
\label{lemm.pointwise.heart}
Suppose $\phi$ is a closed predicate and $n\in\Time$ and $p\in\Pnt$ and $\avaluation$ is a valuation.
Then 
$$
n,p\mentval\Quorum\phi \tequiv \Contraquorum\phi
\quad\text{if and only if}\quad
\bigvee_{O'\in\opensne}\modellabel{\phi}{n,p,O',\avaluation}=\bigwedge_{O'\in\opensne}\modellabel{\phi}{n,p,O',\avaluation}.
$$
(The predicate $\Quorum\phi \tequiv \Contraquorum\phi$ itself satisfies the conditions for being pointwise from Lemma~\ref{lemm.atomic.O}(\ref{item.atomic.pointwise.2}), so we can sensibly write $n,p\mentval \Quorum\phi \tequiv \Contraquorum\phi$ above.)
\end{lemm}
\begin{proof}
From Lemma~\ref{lemm.equiv}(\ref{item.equiv.1}) and the clauses for $\Quorum$ in Figure~\ref{fig.3.phi.f}, and $\Contraquorum$ from Figure~\ref{fig.3.derived} (or Lemma~\ref{lemm.Coquorum.denote}). 
\end{proof}

\begin{lemm}
Suppose $\phi$ is a closed predicate and $n\in\Time$ and $p\in\Pnt$ and $\avaluation$ is a valuation.
Then $\phi$ is pointwise in $(n,\avaluation)$ if and only if $n,p\mentval\pointwise{\phi}$. 
\end{lemm}
\begin{proof}
By routine calculations using Lemma~\ref{lemm.pointwise.heart}, and the clause for $\allOallP$ in Figure~\ref{fig.3.derived}, noting the elementary fact that for every $O\in\opensne$,
$$
\bigwedge_{O'\in\opensne}\modellabel{\phi}{n,p,O',\avaluation}
\leq
\modellabel{\phi}{n,p,O,\avaluation}
\leq
\bigvee_{O'\in\opensne}\modellabel{\phi}{n,p,O',\avaluation} .
$$
\end{proof}

\subsubsection{An interlude: dense sets, and sets with a nonempty open interior}

We recall a pair of standard topological definitions which will turn up repeatedly later, most notably in the context of pointwise predicates as in Proposition~\ref{prop.pointwise.dense.char} immediately below, and again later in Lemma~\ref{lemm.is}:
\begin{defn}[Dense subset]
\label{defn.dense.subset}
Suppose $(\Pnt,\opens)$ is a semitopology and suppose $P\subseteq\Pnt$ is a (not necessarily open) set of points.
\begin{enumerate*}
\item\label{item.dense.subset}
Say that $P$ is \deffont{dense} when $\Forall{O{\in}\opensne}O\intersectswith P$.
\item\label{item.noi.subset}
Say that $P$ \deffont{has a nonempty open interior} when $\Exists{O{\in}\opensne}O\subseteq P$.
\item\label{item.dense.noi}
Write $\dense(P)$ when $P$ is dense, and $\noi(P)$ when $P$ has a nonempty open interior. 
\end{enumerate*} 
\end{defn}

\begin{lemm}
\label{lemm.dense.up-closed}
Suppose $(\Pnt,\opens)$ is a semitopology and $P,P'\subseteq\Pnt$.
Then:
\begin{enumerate*}
\item\label{item.dense.up-closed} 
If $P\subseteq P'$ then $\dense(P)\limp \dense(P')$ (being dense is \emph{up-closed}).
\item\label{item.noi.up-closed} 
If $P\subseteq P'$ then $\noi(P)\limp \noi(P')$ (having a nonempty open interior is \emph{up-closed}).
\item\label{item.noi.noi} 
If $P\in\opensne$ then $\noi(P)$ (a nonempty open set has a nonempty open interior).
\end{enumerate*}
\end{lemm}
\begin{proof}
By routine set calculations from Definition~\ref{defn.dense.subset}.
\end{proof}

\begin{lemm}
Suppose $(\Pnt,\opens)$ is a semitopology and $P\subseteq\Pnt$.
Then
$$
\dense(P) \liff \neg\noi(\Pnt\setminus P) 
\quad\text{and}\quad 
\noi(P) \liff \neg\dense(\Pnt\setminus P) . 
$$
In this sense, $\noi$ and $\dense$ form a \emph{de Morgan dual pair} of operators --- just as $\forall=\neg\exists\neg$ (as noted in Figure~\ref{fig.easy.equivalences}) and $\exists=\tneg\forall\tneg$.
\end{lemm}
\begin{proof}
We sketch the (easy) argument:
\begin{itemize}
\item
Suppose $P$ is not dense; by Definition~\ref{defn.dense.subset} this means that there exists $O\in\opensne$ such that $P\cap O=\varnothing$, thus $O\subseteq\Pnt\setminus P$, thus $\Pnt\setminus P$ has a nonempty open interior.
\item
Suppose $\Pnt\setminus P$ has a nonempty open interior $\varnothing\neq O\subseteq\Pnt\setminus P$.
Then clearly $P$ is not dense.
\qedhere\end{itemize}
\end{proof}

More on how `being dense' works in semitopologies is in~\cite[Chapter~11]{gabbay:semdca}.
The notion behaves somewhat differently than it does in topologies, but not in ways that will affect us here.

\subsubsection{Application to pointwise predicates}

Technically, Proposition~\ref{prop.pointwise.dense.char} is a corollary of Lemma~\ref{lemm.dense.char} and Corollary~\ref{corr.pointwise.pred}.
The result links the four logical modalities $\someOsomeP$, $\allOallP$, $\someOallP$, and $\allOsomeP$ to elementary topological properties of the sets of points at which these modalities are valid --- being nonempty, being full, having a nonempty open interior, and being dense.
We will use Proposition~\ref{prop.pointwise.dense.char}, both directly and via other lemmas that depend on it.

The most important use-case of Proposition~\ref{prop.pointwise.dense.char} will be when $\phi$ is atomic (Notation~\ref{nttn.nullary}(\ref{item.atomic.predicate})) --- having the form 
$$
\tf P(v_1,\dots,v_n)
\quad\text{or}\quad
\tneg\tf P(v_1,\dots,v_n)
$$
--- since then by Lemma~\ref{lemm.atomic.O} $\phi$ is automatically pointwise:
\begin{prop}
\label{prop.pointwise.dense.char}
Suppose $n\in\Time$ and $\avaluation$ is a valuation and $\phi$ is a closed predicate that is pointwise in $(n,\avaluation)$ (Definition~\ref{defn.pointwise.meaning}). 
Then: 
\begin{enumerate*}
\item\label{item.dense.char.pointwise.someoneAll}
$n\mentval\someOsomeP\phi$ if and only if $\varnothing\neq\{p\in\Pnt \mid n,p\mentval\phi\}$. 
\item\label{item.dense.char.pointwise.everyoneAll}
$n\mentval\allOallP\phi$ if and only if $\Pnt=\{p\in\Pnt \mid n,p\mentval\phi\}$. 
\item\label{item.dense.char.pointwise.QuorumBox}
$n\mentval\someOallP\phi$ if and only if $\noi(\{p\in\Pnt \mid n,p\mentval\phi\})$.
\item\label{item.dense.char.pointwise.CoquorumDiamond}
$n\mentval\allOsomeP\phi$ if and only if $\dense(\{p\in\Pnt \mid n,p\mentval\phi\})$.
\end{enumerate*} 
Recall that:
\begin{itemize*}
\item
$\dense$ and $\noi$ are from Definition~\ref{defn.dense.subset}.
\item
The notations $n,p\mentval\phi$ and $n\mentval\phi$ are from Definition~\ref{defn.validity.judgement}.
\end{itemize*}
\end{prop}
\begin{proof}
If $\phi$ is pointwise then by Corollary~\ref{corr.pointwise.pred} $n,p,O\mentval\phi$ if and only if $n,p\mentval\phi$, for any $O\in\opensne$, and we just simplify parts~\ref{item.dense.char.someoneAll} to~\ref{item.dense.char.CoquorumDiamond} of Lemma~\ref{lemm.dense.char}:
\begin{enumerate} 
\item
We reason as follows:
$$
\begin{array}{r@{\ }l@{\quad}l}
n\mentval\someOsomeP\phi
\liff&
\Exists{O\in\opensne}O\between\{p\in\Pnt \mid n,p\mentval\phi\}
&\text{Lemma~\ref{lemm.dense.char}(\ref{item.dense.char.someoneAll})} 
\\
\liff&
\Pnt\between \{p\in\Pnt \mid n,p\mentval\phi\}
&\text{Fact}
\\
\liff&
\{p\in\Pnt \mid n,p\mentval\phi\}\neq\varnothing
&\text{Fact}
\end{array}
$$
\item
We reason as follows:
$$
\begin{array}{r@{\ }l@{\quad}l}
n\mentval\allOallP\phi
\liff&
\Forall{O\in\opensne}O\subseteq\{p\in\Pnt \mid n,p\mentval\phi\}
&\text{Lemma~\ref{lemm.dense.char}(\ref{item.dense.char.everyoneAll})} 
\\
\liff&
\Pnt\subseteq\{p\in\Pnt \mid n,p\mentval\phi\}
&\text{Fact}
\\
\liff&
\{p\in\Pnt \mid n,p\mentval\phi\}=\Pnt
&\text{Fact}
\end{array}
$$
\item
We reason as follows:
$$
\begin{array}{r@{\ }l@{\quad}l}
n\mentval\someOallP\phi
\liff&
\Exists{O\in\opensne}O\subseteq\{p\in\Pnt \mid n,p\mentval\phi\}
&\text{Lemma~\ref{lemm.dense.char}(\ref{item.dense.char.QuorumBox})} 
\\
\liff&
\noi(\{p\in\Pnt \mid n,p\mentval\phi\})
&\text{Definition~\ref{defn.dense.subset}}
\end{array}
$$
\item
We reason as follows:
$$
\begin{array}[b]{r@{\ }l@{\quad}l}
n\mentval\allOsomeP\phi
\liff&
\Forall{O\in\opensne}O\between\{p\in\Pnt \mid n,p\mentval\phi\}
&\text{Lemma~\ref{lemm.dense.char}(\ref{item.dense.char.CoquorumDiamond})} 
\\
\liff&
\dense(\{p\in\Pnt \mid n,p\mentval\phi\})
&\text{Definition~\ref{defn.dense.subset}}
\end{array}
\qedhere$$
\end{enumerate}
\end{proof}
See also Lemma~\ref{lemm.allbut.concretely}, which instantiates Proposition~\ref{prop.pointwise.dense.char}(\ref{item.dense.char.pointwise.QuorumBox}\&\ref{item.dense.char.pointwise.CoquorumDiamond}) to a special case
familiar from the literature on distributed systems.

\subsection{Correct predicates (another well-behavedness property)}

\begin{rmrk}
We defined $\threeCorrect=\{\tvT,\tvF\}$ in Definition~\ref{defn.tv.ment}(\ref{item.tv.ment.correct}) and a corresponding modality $\modTF$ in Figure~\ref{fig.3} (which returns $\tvT$ when its argument is in $\threeCorrect$ and $\tvF$ otherwise), and we studied some of their properties, most notably in Lemma~\ref{lemm.TF.correct}.
 
We now extend this to notions of \emph{correctness} on predicates $\phi$ and predicate symbols $\tf P$.

Correctness turns out to be highly relevant to how we will use our logic.
Partly this is because when a predicate is correct, we are returned to the familiar two-valued world of Boolean logic --- but more subtly, correctness is interesting because of how this concept interacts with the predicate and modal parts of our logic; in certain circumstances we can use a combination of correctness \emph{somewhere} and validity \emph{everywhere} to deduce information about truth (see Remark~\ref{rmrk.why.exciting}) and this turns out to be a good logical model for how reasoning on distributed algorithms is often done in practice.

We start with the definition and some basic lemmas, and develop from there:
\end{rmrk}

\begin{defn}[Correctness]
\label{defn.correct}
Suppose $\phi$ is a predicate and $a\in\tf{VarSymb}$ is a variable symbol. 
Then:
\begin{enumerate*}
\item\label{item.correct.phi}
Consistent with the definition of $\threeCorrect=\{\tvT,\tvF\}$ from Definition~\ref{defn.tv.ment}(\ref{item.tv.ment.correct}),
we define $\correct{\phi}$, read \deffont{$\phi$ is correct}, by
$$
\correct{\phi}=\modTF\phi .
$$
Note that $\phi$ being correct does not mean that $\phi$ returns $\tvT$; it means that $\phi$ returns $\tvT$ or $\tvF$ and in particular it does \emph{not} return $\tvB$.
\item\label{item.correct}\label{item.crashed}
Suppose $\tf P\in\PredSymb$ is a unary predicate symbol.
Define a predicate $\correct{\tf P}$ by 
$$
\correct{\tf P}\ =\ \tall a.\correct{\tf P(a)}\ \stackrel{N\ref{nttn.hos}}{=}\ \tall\modTF\tf P
 .
$$
\item
If $\tf P_1,\dots,\tf P_n$ are unary predicate symbols then we may write 
$$
\correct{\tf P_1,\dots,\tf P_n}
\quad\text{for}\quad 
\correct{\tf P_1}\tand\dots\tand\correct{\tf P_n}.
$$
(See for example axiom \rulefont{LdrCorrect} in Figure~\ref{fig.logical.paxos}.) 
\item\label{item.P-correct}
Suppose $n\in\Time$ and $p\in\Pnt$ and $\avaluation$ is a valuation, and suppose $\tf P$ is a unary predicate symbol.
Then when $n,p\mentval\correct{\tf P}$ we may call $p$ \deffont{$\tf P$-correct}. 
\end{enumerate*}
\end{defn}

\begin{lemm}
\label{lemm.correct.valid}
Suppose $\acontext$ is a context and $\phi$ is a closed predicate.
Then 
$$
\acontext\ment\correct{\phi}
\quad\text{if and only if}\quad
\modellabel{\phi}{\acontext}\in\threeCorrect.
$$
\end{lemm}
\begin{proof}
This just re-states Lemma~\ref{lemm.tv.ment.TF}(\ref{item.tv.ment.TF.2}).
\end{proof}

\begin{rmrk}
\label{rmrk.tf.synonyms}
Definition~\ref{defn.correct}(\ref{item.correct.phi}) gives us two synonyms for the modality expressing that the truth-value returned by a predicate $\phi$ is either $\tvT$ or $\tvF$ (not $\tvB$): $\modTF$ and $\tf{correct}$.

Although these are synonymous, there is a difference in emphasis:
\begin{itemize*}
\item
$\modTF$ is a logical modality and just means `true or false'.
\item
The literature on distributed systems uses a term `correctness', which indicates that a process succeeds and returns a correct value. 
In our logical model this corresponds to returning $\tvT$ or $\tvF$ --- and thus not returning $\tvB$. 
So when we write $\tf{correct}$, the emphasis is on `not $\tvB$' with emphasis on a process `not failing'.
\end{itemize*}
We use whichever notation seems most natural in context. 
For example, in Figure~\ref{fig.logical.paxos}:
\begin{enumerate*}
\item
In \rulefont{LdrCorrect} we write $\correct{\tf{propose},\tf{write},\tf{decide}}$ because intuitively we are asserting that at a particular time the leader \emph{does not crash} when proposing, writing, or deciding.
\item
In \rulefont{LdrExist} we write $\modTF\tf{leader}$, because $\tf{leader}$ is a predicate that identifies whether a participant is the leader at a particular time, and this is either true or false.\footnote{In practice determining who the leader is requires a participant to compute a function, but the computation is not distributed, in the sense that prior agreement amongst participants is assumed on what the function is and what seed data it uses, and each participant computes it locally to get (the same) definitive yes/no answer.}
\end{enumerate*} 
If the reader prefers, they can safely read $\modTF$ every time they see $\tf{correct}$, and they will be precisely correct (pun intended). 
\end{rmrk}

Many predicates of interest are correct, including correctness itself, our truth-modalities, and equality and equivalence:
\begin{lemm}
\label{lemm.equality.correct}
Suppose $\acontext$ is a context and $\phi$ and $\phi'$ are closed predicates and $v,v'\in\Val$.
Then:
\begin{enumerate*}
\item
$\modellabel{\correct{\correct{\phi}}}{\acontext}=\tvT$.
\quad We can (perhaps cheekily) read this as `correctness is correct'.
\item
$\modellabel{\correct{\tf Q\phi}}{\acontext}=\tvT$ for $\tf Q\in\{\modT,\modF,\modB,\modTB,\modTF,\modFB\}$. 
\item\label{item.equality.correct}
$\modellabel{\correct{v\teq v'}}{\acontext}=\tvT$.
In words: equality is correct.
\item\label{item.equivalence.correct}
$\modellabel{\correct{\phi\tequiv \phi'}}{\acontext}=\tvT$.
In words: equivalence is correct.
\end{enumerate*}
\end{lemm}
\begin{proof}
By routine checking from Figures~\ref{fig.3}, \ref{fig.3.phi.f}, \ref{fig.3.derived}, and (for equivalence $\tequiv$)~\ref{fig.3.iff}.
\end{proof}

Lemma~\ref{lemm.correct.P.Pv} and Corollary~\ref{corr.modTF.texi} are straightforward, but we spell it out:
\begin{lemm}
\label{lemm.correct.P.Pv}
Suppose $\tf P\in\PredSymb$ is a unary predicate symbol and $\acontext$ is a context.
Then 
\begin{enumerate*}
\item
$\modellabel{\correct{\tf P}}{\acontext} = \bigwedge_{v{\in}\Val}\modellabel{\correct{\tf P(v)}}{\acontext}$
and
\item
$\modellabel{\correct{\tf P}}{\acontext}\leq \modellabel{\correct{\tf P(v)}}{\acontext}$ for every value $v\in\Val$.
\end{enumerate*}
\end{lemm}
\begin{proof}
By Definition~\ref{defn.correct}(\ref{item.correct}) and Figure~\ref{fig.3.derived} 
$$
\modellabel{\correct{\tf P}}{\acontext} 
= 
\modellabel{\tall a.\correct{\tf P(a)}}{\acontext} 
= 
\bigwedge_{v{\in}\Val} \modellabel{\correct{\tf P(v)}}{\acontext} 
\leq 
\modellabel{\correct{\tf P(v)}}{\acontext} .
\qedhere$$
\end{proof}

\begin{corr}
\label{corr.modTF.texi}
Suppose $\acontext$ is a context and $\tf P\in\PredSymb$ is a unary predicate symbol and $a\in\tf{VarSymb}$ is a variable symbol. 
Then 
$$
\acontext\ment\correct{\tf P}
\quad\text{implies}\quad
\acontext\ment\correct{\texi\tf P}\ \land\ \acontext\ment\correct{\tall\tf P}.
$$
\end{corr}
\begin{proof}
Suppose $\acontext\ment\correct{\tf P}$.
By Lemma~\ref{lemm.correct.P.Pv} we see that $\Forall{v\in\Val}\acontext\ment\correct{\tf P(v)}$.
By Lemma~\ref{lemm.correct.valid} $\Forall{v\in\Val}\modellabel{\tf P(v)}{\acontext}\in\threeCorrect$.
Now the clauses for $\texi$ in Figure~\ref{fig.3.phi.f} and $\tall$ in Figure~\ref{fig.3.derived} take a supremum and infimum respectively of $\{\modellabel{\tf P(v)}{\acontext} \mid v\in\Val\}\subseteq\threeCorrect=\{\tvT,\tvF\}$, and this supremum or infimum is also clearly in $\threeCorrect=\{\tvT,\tvF\}$.
We just use Lemma~\ref{lemm.correct.valid}. 
\end{proof}

Proposition~\ref{prop.tv.ment.TF.model}(\ref{item.tv.ment.TF.phi.lr}) is a wrapper for Lemma~\ref{lemm.TF.correct}.
This deceptively simple equivalence turns out to be a key way that correctness can be usefully applied:
\begin{prop}
\label{prop.tv.ment.TF.model}
Suppose $\phi$ is a closed predicate and $\tf P$ is a unary predicate symbol and $\acontext$ is a context.
Then:
\begin{enumerate*}
\item\label{item.tv.ment.TF.phi.lr}
$\acontext\ment\correct{\phi}$ if and only if $\acontext\ment\phi\tequiv\modT\phi$. 

In words: $\phi$ is correct precisely when `$\phi$ is valid' is equivalent to `$\phi$ is true'.
\item\label{item.tv.ment.TF.P}
$\acontext\ment\correct{\tf P}$ if and only if $\Forall{v\in\Val}(\acontext\ment\tf P(v)\tequiv\modT\tf P(v))$. 
\item\label{item.modTF.texi.P.valid}
$\acontext\ment\correct{\tf P}$ implies $\acontext\ment\texi\tf P\tequiv \modT\texi\tf P$ and $\acontext\ment\tall\tf P\tequiv\modT\tall\tf P$. 
\end{enumerate*}
\end{prop}
\begin{proof}
We consider each part in turn:
\begin{enumerate}
\item
Suppose $\acontext\ment\correct{\phi}$.
By Lemma~\ref{lemm.TF.correct} $\modT\modellabel{\phi}{\acontext}=\modellabel{\phi}{\acontext}$.
By Lemma~\ref{lemm.equiv}(\ref{item.equiv.1}) and the clause for $\modT$ in Figure~\ref{fig.3.phi.f}, $\acontext\ment \phi\tequiv\modT\phi$.
The reverse implication just reverses the same reasoning.
\item
We combine Lemma~\ref{lemm.correct.P.Pv} with part~\ref{item.tv.ment.TF.phi.lr} of this result for $\phi=\tf P(v)$ for every $v\in\Val$.
\item
From Corollary~\ref{corr.modTF.texi} and Proposition~\ref{prop.tv.ment.TF.model}(\ref{item.tv.ment.TF.phi.lr}).
\qedhere\end{enumerate}
\end{proof}

Proposition~\ref{prop.correct.cobox.boxT} explores how correctness interacts with our compound modalities from Subsection~\ref{subsect.compound.modalities}.
\begin{prop}
\label{prop.correct.cobox.boxT}
Suppose $\phi$ is a closed predicate and $\tf P\in\PredSymb$ is a unary predicate symbol. 
Suppose $n\in\Time$ and $\avaluation$ is a valuation.
Then:
\begin{enumerate*}
\item\label{item.correct.cobox.boxT.1}
$n\mentval\someOallP\correct{\phi}\tand\allOallP\phi$ implies $n\mentval \modT\someOallP\phi$.\footnote{Alert: note that we write $\allOallP\phi$ in the right-hand conjunct, not $\allOsomeP\phi$.}
\item\label{item.correct.cobox.boxT.2}
$n\mentval\someOallP\correct{\tf P}\tand\allOallP\texi\tf P$ implies 
$n\mentval\modT\someOallP\texi\tf P$.
\end{enumerate*}
\end{prop}
\begin{proof}
We consider each part in turn:
\begin{enumerate}
\item
Suppose $n\mentval\someOallP\correct{\phi}\tand\allOallP\phi$.
By Proposition~\ref{prop.someone.implies.someoneAll}(\ref{item.someone.implies.someoneAll.4b})
$n\mentval\someOallP(\correct{\phi}\tand\phi)$.
By Proposition~\ref{prop.tv.ment.TF.model}(\ref{item.tv.ment.TF.phi.lr}) and Lemma~\ref{lemm.someone.commutation} (for $\someOallP$)
we have $n\mentval\modT\someOallP\phi$ as required.
\item
We combine part~\ref{item.correct.cobox.boxT.1} of this result with Corollary~\ref{corr.modTF.texi}.
\qedhere\end{enumerate}
\end{proof}

\begin{rmrk}[Why this is interesting]
\label{rmrk.why.exciting}
The statement of Proposition~\ref{prop.correct.cobox.boxT} is quite heavy on the symbols.
For clarity, let us re-read it, but let us replace $\ment\someOallP\tf{correct}$ with the phrase `mostly correct', and $\ment\allOallP$ with the phrase `everywhere valid', and $\texi$ with the phrase `for some value'.
Then we get:
\begin{enumerate*}
\item
If $\phi$ is \emph{mostly correct} and \emph{everywhere valid}, then it is \emph{mostly true}.
\item
If $\tf P$ is \emph{mostly correct} and \emph{everywhere valid} for some value, then it is \emph{mostly true} for some value.
\end{enumerate*}
So what we see is that we can promote weak statements about validity everywhere, with stronger but more limited statements about correctness, to obtain limited statements about truth.

The reader familiar with how distributed algorithms work may recognise this style of reasoning.
So: here it is beginning to appear in our logic.
\end{rmrk}

\subsection{$\mrup{a}{\phi}{v}$: a most recent $v$ that makes $\phi[a\ssm v]$ true}
\label{subsect.mrup}

\begin{rmrk}
\label{rmrk.explain.mrup}
We see from Figure~\ref{fig.3.derived} and Remark~\ref{rmrk.comments.on.semantics}(\ref{item.mru.program}) that $\mrup{a}{\phi}{v}$ looks back into the past, searching for a most recent time when $\model{\texi a.\phi}=\tvT$, and checks if at that most recent time, $\model{\phi[a\ssm v]}=\tvT$.
$\mrup{a}{\phi}{v}$ returns $\tvT$ if this search succeeds, and $\tvF$ if this fails.
It never returns $\tvB$.

In words: $\mrup{a}{\phi}{v}$ checks whether $v$ is a value that has most recently made $\phi$ true.

We need $\mrup{a}{\phi}{v}$ for our axiomatisation of Paxos in Figure~\ref{fig.logical.paxos}, where 
\begin{itemize*}
\item
in \rulefont{PaxSend?} we write $\mru{\tf{accept}}{v}$\ (short for $(\mrup{a}{\tf{accept}(a))}{v}$, as per Figure~\ref{fig.hos}), and 
\item
in \rulefont{PaxWrite?} we write $\mru{(\somewhere\tf{accept})}{v}$\ (short for $\mrup{a}{(\somewhere\tf{accept}(a))}{v}$).
\end{itemize*} 
We read $\mrup{a}{\phi}{v}$ as 
\begin{quote}
\deffont{$v$ is a value that most recently makes $\phi[a\ssm v]$ true} .
\end{quote}
In the common case that $\phi=\tf P(a)$ for some unary predicate symbol $\tf P$, then we might just say
\begin{quote}
\deffont{$v$ is a value that most recently makes $\tf P$ true}.
\end{quote} 
In this Subsection we study the properties of this construct in more detail.
\end{rmrk}

\begin{lemm}
\label{lemm.mru.unambivalent}
Suppose $\texi a.\phi$ is a closed predicate and $v\in\Val$ and $\acontext$ is a context.
Then:
\begin{enumerate*}
\item
$\acontext\ment\correct{\mrup{a}{\phi}{v}}$.
\item
$\acontext\ment\mrup{a}{\phi}{v}\tequiv\modT\mrup{a}{\phi}{v}$.
\end{enumerate*}
\end{lemm}
\begin{proof}
We consider each part in turn:
\begin{enumerate}
\item
By Definition~\ref{defn.correct}(\ref{item.correct.phi}) $\correct{\mrup{a}{\phi}{v}}=\modTF\mrup{a}{\phi}{v}$, and this means that $\mrup{a}{\phi}{v}$ always returns $\tvT$ or $\tvF$ (but not $\tvB$).
We just check the relevant clause in Figure~\ref{fig.3.derived} and see that this is so.
\item
We just use part~1 of this result and Proposition~\ref{prop.tv.ment.TF.model}(\ref{item.tv.ment.TF.phi.lr}).
\qedhere\end{enumerate}
\end{proof}

\begin{rmrk}
\label{rmrk.explain.mru}
\leavevmode
\begin{enumerate}
\item\label{item.explain.mru.T}
Lemma~\ref{lemm.mru.unambivalent}(2) %
notes that $\mrup{a}{\phi}{v}$ and $\modT(\mrup{a}{\phi}{v})$ are extensionally equivalent (Definition~\ref{defn.phi.equivalent}).

We may write one or the other, depending on what seems most clear in a proof.
In particular, if we are using \strongmodusponens then we might retain a $\modT$ modality for clarity.
\item 
We try to head off an easy error: note that $\mrup{a}{\phi}{v}$ is not equivalent to $\modT\recent\phi[a\ssm v]$, nor the (by Lemma~\ref{lemm.someone.commutation}) equivalent form $\recent\modT\phi[a\ssm v]$.
\begin{itemize*}
\item
$\modT\recent\phi[a\ssm v]$ returns $\tvT$ when $\phi[a\ssm v]$ has returned $\tvT$ at some stage in the past.
\item
$\mrup{a}{\phi}{v}$ (`$v$ is a value that most recently makes $\phi[a\ssm v]$ true') returns $\tvT$ when $\phi[a\ssm v]$ has returned $\tvT$ in \emph{the most recent stage} such that $\texi a.\phi$ returned $\tvT$ (i.e. such that $\phi[a\ssm v']$ returned $\tvT$ for \emph{some} $v'$).
\end{itemize*}
These are not the same thing.
\end{enumerate}
\end{rmrk}

Remark~\ref{rmrk.explain.mru} notes that $\mrup{a}{\phi}{v}$ is different from $\modT\recent\phi[a\ssm v]$.
That said, there is a close link between them, which will turn out to be important for understanding our axiomatisation of Paxos (see Lemma~\ref{lemm.evaporates.simplify} and surrounding discussion): 
\begin{lemm}
\label{lemm.evaporates}
Suppose $\texi a.\phi$ is a closed predicate and $\acontext=(n,p,O,\avaluation)$ is a context. 
Then
$$
\acontext\ment(\texi v.\mrup{a}{\phi}{v})\tequiv \modT\recent\texi a.\phi.
$$
\end{lemm}
\begin{proof}
First, note from Lemma~\ref{lemm.someone.commutation} that $\ment\modT\recent\texi a.\phi\tequiv \recent\texi a.\modT\phi$.
And, note that $\texi a.\phi$ and $\texi v.\phi[a\ssm v]$ denote the same syntax, as per Remark~\ref{rmrk.pedant.point}(\ref{item.pedant.texi.v}).

We unpack the meaning of $\modellabel{\texi v.\mrup{a}{\phi}{v}}{\acontext}$ in Figure~\ref{fig.3.derived} and note that it returns $\tvT$ when there exist $v\in\Val$ and $0\leq n'<n$ such that $n',p,O\mentval \modT\phi[a\ssm v]$, and $\tvF$ otherwise.
But this is also what $\modellabel{\recent\texi v.\modT(\phi[a\ssm v])}{\acontext}$ does, so by Lemma~\ref{lemm.equiv}(\ref{item.equiv.1}) we are done.
\end{proof}

Lemma~\ref{lemm.mru.today.yesterday} connects $\mrup{a}{\phi}{v}$ to the fact that our notion of time is discrete and so has a well-defined notion of a `tomorrow' (i.e. a `next timestep'):
\begin{lemm}
\label{lemm.mru.today.yesterday}
Suppose $\texi a.\phi$ is a closed predicate and $v\in\Val$ and $\acontext$ is a context. 
Then
$$
\acontext\ment \tomorrow(\mrup{a}{\phi}{v}) \tequiv (\modT\phi[a\ssm v] \tor ((\tneg\modT\texi a.\phi) \tand \mrup{a}{\phi}{v})).
$$
\end{lemm}
\begin{proof}
It is a fact of Figure~\ref{fig.3.derived} that $\acontext\ment\tomorrow\yesterday\psi\tequiv\psi$ for any closed $\psi$.
The result follows from this observation and from the clause for $\mrup{a}{\phi}{v}$ in Figure~\ref{fig.3.derived}. 
\end{proof}

\begin{corr}
\label{corr.mru.today.to.yesterday}
Suppose $\texi a.\phi$ is a closed predicate and $v\in\Val$ and $\acontext$ is a context. 
Then:
\begin{enumerate*}
\item\label{item.mru.today.implies.tomorrow}
$\acontext\ment \phi[a\ssm v] \timpc \tomorrow(\mrup{a}{\phi}{v})$.
\item\label{item.mru.tomorrow.options}
$\acontext\ment \tomorrow(\mrup{a}{\phi}{v}) \timpc (\phi[a\ssm v]\tor \mrup{a}{\phi}{v})$. 
\end{enumerate*}
\end{corr}
\begin{proof}
Elementary from Lemma~\ref{lemm.mru.today.yesterday} using \strongmodusponens.
\end{proof}

\begin{lemm}
\label{lemm.mru.to.recent}
Suppose $\texi a.\phi$ is a closed predicate %
and $v,v'\in\Val$ and $\acontext$ is a context. 
Then:
\begin{enumerate*}
\item\label{item.mru.to.recent.2}
$\acontext\ment(\mrup{a}{\phi}{v}\tand\mrup{a}{\phi}{v'}) \timpc \recent(\phi[a\ssm v]\tand\phi[a\ssm v'])$.
\item\label{item.mru.to.recent.1}
$\acontext\ment(\mrup{a}{\phi}{v}) \timpc \recent(\phi[a\ssm v])$.

If $\phi=\modality{\tf Q}\tf P(a)$ for a unary predicate symbol $\tf P\in\PredSymb$ and some prefix of modalities $\modality{\tf Q}$, then we can elide the variable symbol $a$ as per Notation~\ref{nttn.hos} and Figure~\ref{fig.hos}, and this becomes
$$
\acontext\ment\mru{(\modality{\tf Q}\tf P)}{v} \timpc \recent\modality{\tf Q}\tf P(v).
$$
\end{enumerate*}
\end{lemm}
\begin{proof}
We consider each part in turn:
\begin{enumerate}
\item
Unpacking the clause for $\mruarrow$ in Figure~\ref{fig.3.derived}, if $v$ and $v'$ are values that most recently make $\phi[a\ssm\text{-}]$ be $\tvT$, then certainly $\recent(\phi[a\ssm v]\tand\phi[a\ssm v'])$.
\item
Follows from part~\ref{item.mru.to.recent.2} as the special case that $v'=v$.
\qedhere\end{enumerate}
\end{proof}

\section{\lowercase{n}-twined semitopologies}
\label{sect.n-twined}

\subsection{Antiseparation properties}

We are particularly interested in models over semitopologies whose open sets have certain \emph{antiseparation properties}, by which we mean properties that involve the existence of nonempty intersections between open sets.\footnote{For comparison, topology is often concerned with \emph{separation} properties, like Hausdorff separation, which have to do with existence of \emph{empty} intersections between open sets.  In this sense, the interest in this Section is the precise dual to where the emphasis traditionally lies in topology.}

The design space of antiseparation properties is rich, as the first author explored in~\cite{gabbay:semdca} --- but for our needs in this paper we will consider just one canonical property which we will call being \emph{$n$-twined}.
We define the $n$-twined property in Definition~\ref{defn.twined}, we explain why it matters in Remark~\ref{rmrk.motivate.n-twined}, and then we study its properties.

\begin{defn}
\label{defn.twined}
\leavevmode
Suppose $(\Pnt,\opens)$ is a semitopology and $n\in\Ngeqz$.
Then call $(\Pnt,\opens)$ \deffont{$n$-twined} when every collection of $n$ nonempty open sets has a nonempty intersection.
In symbols:
$$
\Forall{O_1,\dots,O_n\in\opensne}O_1\cap\dots\cap O_n\neq\varnothing.
$$
\end{defn}

\begin{lemm}
\label{lemm.1-twined}
We note two edge cases:
\begin{enumerate*}
\item
A semitopology is $0$-twined if and only if it is nonempty (Notation~\ref{nttn.empty.semitopology}).
\item
Every semitopology is $1$-twined.
\end{enumerate*}
\end{lemm}
\begin{proof}
We consider each part in turn:
\begin{enumerate}
\item
If $n=0$ then we need to show that $\Pnt$ (the intersection of zero nonempty sets) is not equal to $\varnothing$.
By Notation~\ref{nttn.empty.semitopology} this is true precisely when the semitopology is nonempty. 
\item
By Definition~\ref{defn.twined}, we must show that $O\neq\varnothing$ for every $O\in\opensne$.
This is true.
\qedhere\end{enumerate}
\end{proof}

\begin{figure}[t]
\centering
\providecommand{\DiagramHatchSpacingFactor}{0.50}
\begin{tikzpicture}[baseline={(0,0.06)},x=1cm,y=1cm,
                    line cap=round,line join=round, scale=0.7]
  \begin{scope}
    \clip (-1.35,0) ellipse [x radius=2.17,y radius=1.60];
    \foreach \k in {-12,...,12} {
      \draw[line width=0.42pt]
        ({0.70*\DiagramHatchSpacingFactor*\k-1.85},-2.15)
        -- ({0.70*\DiagramHatchSpacingFactor*\k+1.85},2.15);
    }
  \end{scope}
  \begin{scope}
    \clip (1.35,0.12) ellipse [x radius=2.17,y radius=1.60];
    \foreach \k in {-12,...,12} {
      \draw[line width=0.42pt]
        ({0.72*\DiagramHatchSpacingFactor*\k+2.10},-2.15)
        -- ({0.72*\DiagramHatchSpacingFactor*\k-2.10},2.15);
    }
  \end{scope}

  \draw[line width=0.85pt] (-1.35,0)
    ellipse [x radius=2.17,y radius=1.60];
  \draw[line width=0.85pt] (1.35,0.12)
    ellipse [x radius=2.17,y radius=1.60];

  \fill (0.08,0.10) circle [radius=1.4pt];
  \node[font=\small,inner sep=0pt] at (0,2.55)
    {$v$ must equal $v'$};
  \draw[-{Stealth[length=1.6mm,width=1.1mm]},line width=0.55pt]
    (0.05,2.25) .. controls (0.42,1.53) and (0.27,0.74) .. (0.11,0.37);

  \node[font=\small,inner sep=0pt] at (-2.00,-2.04) {``Output is $v$''};
  \node[font=\small,inner sep=0pt] at ( 2.00,-2.04) {``Output is $v'$''};
\end{tikzpicture}
\qquad
\providecommand{\DiagramWaveSpacingFactor}{0.4}
\begin{tikzpicture}[baseline={(0,-0.15)},x=1cm,y=1cm,
                    line cap=round,line join=round, scale=0.7]
  \pgfmathsetmacro{\WavePitch}{0.71*\DiagramWaveSpacingFactor}
  \pgfmathtruncatemacro{\WaveHalfCount}{floor(1.80/\WavePitch)}
  \pgfmathtruncatemacro{\WaveFirstRow}{-\WaveHalfCount}
  \begin{scope}
    \clip (0,2.00) ellipse [x radius=1.90,y radius=1.90];
    \foreach \row in {\WaveFirstRow,...,\WaveHalfCount} {
      \pgfmathsetmacro{\WaveY}{1.96+\row*\WavePitch}
      \draw[line width=0.48pt]
        plot[domain=-2.10:2.10,samples=151,variable=\x]
        (\x,{\WaveY+0.085*sin(360*\x/0.65)});
    }
  \end{scope}

  \begin{scope}
    \clip (-1.35,-0.15) ellipse [x radius=2.03,y radius=1.95];
    \foreach \k in {-12,...,12} {
      \draw[line width=0.42pt]
        ({0.70*\DiagramHatchSpacingFactor*\k-1.85},-2.50)
        -- ({0.70*\DiagramHatchSpacingFactor*\k+1.85},2.10);
    }
  \end{scope}
  \begin{scope}
    \clip (1.35,-0.15) ellipse [x radius=2.03,y radius=1.95];
    \foreach \k in {-12,...,12} {
      \draw[line width=0.42pt]
        ({0.72*\DiagramHatchSpacingFactor*\k+2.10},-2.50)
        -- ({0.72*\DiagramHatchSpacingFactor*\k-2.10},2.10);
    }
  \end{scope}

  \draw[line width=0.85pt] (0,2.00)
    ellipse [x radius=1.90,y radius=1.90];
  \draw[line width=0.85pt] (-1.35,-0.15)
    ellipse [x radius=2.03,y radius=1.95];
  \draw[line width=0.85pt] (1.35,-0.15)
    ellipse [x radius=2.03,y radius=1.95];

  \node[font=\small,inner sep=0pt] at (0,4.39) {Byzantine / broken};
  \node[font=\small,inner sep=0pt] at (-2.00,-2.53) {``Output is $v$''};
  \node[font=\small,inner sep=0pt] at ( 2.00,-2.53) {``Output is $v'$''};

  \fill[white] (0,-0.75) circle [radius=2.7pt];
  \fill (0,-0.75) circle [radius=1.4pt];
  \node[font=\small,inner sep=0pt] at (0,-3.14)
    {$v$ must equal $v'$};
  \draw[-{Stealth[length=1.6mm,width=1.1mm]},line width=0.55pt]
    (0,-2.89) .. controls (0.16,-2.26) and (0.09,-1.42) .. (0,-1.01);
\end{tikzpicture}
\caption{Compatibility in 2- and 3-twined semitopologies (Remark~\ref{rmrk.motivate.n-twined})}
\label{fig.n-twined}
\end{figure}

\begin{rmrk}[Discussion of the special cases $n=2$ and $n=3$]
\label{rmrk.motivate.n-twined}
Lemma~\ref{lemm.1-twined} shows that being $0$-twined and $1$-twined are straightforward properties.
In contrast, the $2$- and $3$-twined semitopologies will be very interesting to us in this paper.
We give a high-level discussion of why this is:
\begin{itemize}
\item
Suppose we are in a distributed system running on a $2$-twined semitopology.
We assume that all participants are honest.
Suppose further that we have 
\begin{itemize*}
\item
a nonempty open set of participants $O\in\opensne$ who agree that some predicate $\phi$ is true (respectively: valid), and 
\item
a nonempty open set of participants $O'\in\opensne$ who agree that some predicate $\phi'$ is true (respectively: valid).
\end{itemize*}
Then because the system is $2$-twined, there exists some $p\in O\cap O'$ who agrees that both $\phi$ and $\phi'$ are true (respectively: valid).\footnote{We will render this in our logic later on, in Proposition~\ref{prop.supertwined.qcbt}(\ref{item.intertwined.qcbt.2}).}

This simple observation underlies guarantees that participants generate compatible outputs (for some protocol-dependent notion of compatibility) in many distributed algorithms, because it guarantees that $\phi$ and $\phi'$ are not logically conflicting --- e.g. that we do not have something like $\phi'=\modF\phi$ --- because if they were, then an honest $p$ could not believe both.

In the special case that $\phi$ is `the output of our consensus algorithm is $v$' and $\phi'$ is `the output of our consensus algorithm is $v'$', the $2$-twined property makes the strong compatibility guarantee that where consensus algorithms succeed (i.e. where participants do not crash), the output is compatible in the sense that two participants will agree on the output of the consensus algorithm (assuming outputs are supposed to be unique).

This is illustrated in the left-hand diagram in Figure~\ref{fig.n-twined}.
\item
Suppose we are in a distributed system running on a $3$-twined semitopology.
Some participants may be dishonest (also called Byzantine or broken), and they might not follow the protocol.

Suppose further that we have nonempty open sets $O,O'\in\opensne$ of participants who agree on the truth (respectively: validity) of predicates $\phi$ and $\phi'$ respectively.
Suppose further that we have a nonempty open set $O''\in\opensne$ of participants who are honest (meaning that we bound dishonest participants to at most some \emph{closed} set). 
Because the system is $3$-twined, there exists some $p\in O\cap O'\cap O''$ --- this $p$ is a participant who agrees on the truth (respectively: validity) of $\phi$ and $\phi'$ and is also honest.\footnote{We will render this in our logic twice: in Proposition~\ref{prop.supertwined.qcbt}(\ref{item.supertwined.qcbt.2}) as a general result, and then as a useful special case in Corollary~\ref{corr.two.quorums.intersect.correct}(\ref{item.two.quorums.intersect.correct.supertwined}).}

In this case, $\phi\tand\phi'$ must be consistent, because if it were not then the honest participant $p$ could not believe both $\phi$ and $\phi'$.
Such a situation is illustrated in the right-hand diagram in Figure~\ref{fig.n-twined}.

Just as for the $2$-twined case, this simple argument underlies many correctness arguments in distributed systems in the presence of dishonest participants.
\end{itemize}
There is no barrier to considering $n$-twined semitopologies for $n\geq 4$, and this would just give us stronger guarantees of intersection.

A brief survey of intersection properties in the literature is in Remark~\ref{rmrk.ntwined.and.the.literature}.
\end{rmrk}

\subsection{Being $n$-twined and being dense / having nonempty open interior}

The property of being $n$-twined is closely related to the topological notion of being a \emph{dense set} and having a \emph{nonempty open interior} (Definition~\ref{defn.dense.subset}).
We explore this in Lemma~\ref{lemm.is} and Proposition~\ref{prop.noi.dense}; then in Proposition~\ref{prop.allbut.supertwined} we specialise the discussion to the important special case of $\AllBut{N}{f}$.

\begin{lemm}
\label{lemm.is}
Suppose $(\Pnt,\opens)$ is a semitopology and $n\in\mathbb N_{\geq 1}$.
Then:
\begin{enumerate*}
\item\label{item.is.supertwined}
$(\Pnt,\opens)$ is $n$-twined if and only if every collection of $n\minus 1$ nonempty open sets has a dense intersection: 
$$
\Forall{O_1,\dots O_{n\minus 1}\in\opensne}\dense(O_1\cap\dots\cap O_{n\minus 1}).
$$
\item\label{item.is.intertwined}
As a corollary, $(\Pnt,\opens)$ is \xxtwined if and only if every nonempty open set is dense: 
$$
\Forall{O\in\opensne}\dense(O).
$$
\end{enumerate*}
\end{lemm}
\begin{proof}
Both parts follow from Definitions~\ref{defn.dense.subset}(\ref{item.dense.subset}) and~\ref{defn.twined}, just by pushing around definitions.

One corner case merits brief mention: if $n=1$ then we must show that every collection of $0$ sets (i.e. $\varnothing$) is dense.
So, we need to show that $\bigcap\varnothing$ is dense.
By convention $\bigcap\varnothing=\Pnt$, and this is indeed dense, so we are done. 
\end{proof}

\begin{prop}
\label{prop.noi.dense}
Suppose $(\Pnt,\opens)$ is a semitopology and $n\in\mathbb N_{\geq 1}$.
Then the following are equivalent:
\begin{enumerate*}
\item\label{item.noi.dense.1}
$(\Pnt,\opens)$ is $n$-twined. 
\item\label{item.noi.dense.2}
$\Forall{P_1,\dots,P_{n\minus 1}\subseteq\Pnt}(\bigwedge_{1{\leq}i{\leq}n\minus 1}\noi(P_i)) \limp \dense(\bigcap_{1{\leq}i{\leq}n\minus 1}P_i)$.
\end{enumerate*}
\end{prop}
\begin{proof}
Part~\ref{item.noi.dense.2} trivially implies part~\ref{item.noi.dense.1} because we can take $P_1,\dots,P_{n\minus 1}\in\opensne$, so that by Definition~\ref{defn.dense.subset} $\noi(P_i)$ for $1\leq i\leq n\minus 1$ holds by construction by Lemma~\ref{lemm.dense.up-closed}(\ref{item.noi.noi}), and we use Lemma~\ref{lemm.is}(\ref{item.is.supertwined}).

Showing that part~\ref{item.noi.dense.1} implies part~\ref{item.noi.dense.2} is only slightly less trivial.
Suppose $(\Pnt,\opens)$ is $n$-twined and $P_1,\dots,P_{n\minus 1}\subseteq\Pnt$ and $\noi(P_i)$ for each $1\leq i\leq n\minus 1$, meaning by Definition~\ref{defn.dense.subset} that for each such $i$ there exists an $O_i\in\opensne$ such that $O_i\subseteq P_i$.
By Lemma~\ref{lemm.is}(\ref{item.is.supertwined}) $O_1\cap\dots\cap O_{n\minus 1}$ is dense, and by Lemma~\ref{lemm.dense.up-closed}(\ref{item.dense.up-closed}) so is $P_1\cap\dots\cap P_{n\minus 1}$.
\end{proof}

Proposition~\ref{prop.allbut.supertwined} provides the canonical example of $n$-twined semitopologies.
It mildly generalises known properties in the literature; the case $n=2$ in Proposition~\ref{prop.allbut.supertwined}(\ref{item.allbut.intertwined}) corresponds to the $Q^2$ property from~\cite[Theorem~1]{DBLP:conf/asiacrypt/DamgardDFN07}, and the case $n=3$ corresponds to the $Q^3$ property of~\cite[Definition 2]{DBLP:journals/dc/AlposCTZ24}.
The topology-flavoured presentations in Proposition~\ref{prop.allbut.supertwined}(\ref{item.allbut.NOI.dense}\&\ref{item.allbut.NOI.dense.3}) are new.
\begin{prop}
\label{prop.allbut.supertwined}
Suppose $0\leq f<N\in\Ngeqz$ and $P\subseteq\Nset{N}$ and $n\in\mathbb N_{\geq 1}$.
Then:
\begin{enumerate*}
\item\label{item.allbut.intertwined}\label{item.allbut.supertwined}
$N>n*f$ if and only if $\AllBut{N}{f}$ is $n$-twined.
\item\label{item.allbut.NOI.dense}
$\noi(P)$ (and equivalently $P\in\opensne$) if and only if $\mycard P\geq N\minus f$.\footnote{In general semitopologies, $\noi(P)$ ($P$ having a nonempty open interior) is not the same as $P\in\opensne$ ($P$ being a nonempty open set); but in the case of $\AllBut{N}{f}$, these two properties happen to be equal.}
\item\label{item.allbut.NOI.dense.3}
$\dense(P)$ if and only if $\mycard P\geq f\plus 1$.
\end{enumerate*}
\end{prop}
\begin{proof}
We consider each part in turn:
\begin{enumerate}
\item
Suppose $N>n*f$ and suppose $O_1,\dots,O_n\in\opensne$ are nonempty open sets in $\AllBut{N}{f}$.
From Definition~\ref{defn.allbut} this means that they each contain all but $f$ elements of $\Nset{N}$.
It follows that $O_1\cap\dots \cap O_n$ contains (at least) all but $n*f$ elements of $\Nset{N}$.
Since $N\gneq n*f$, this intersection is nonempty. 

Conversely, if $N\leq n*f$ then it is routine to partition $\Nset{N}$ into $n$ sets, each of which has at most $f$ elements.
Then their complements are a set of $n$ open sets whose intersection is empty.
\item
By routine counting arguments.
\item
By routine counting arguments.
\qedhere\end{enumerate}
\end{proof}

\begin{rmrk}
\label{rmrk.ntwined.and.the.literature}
Being \xxtwined corresponds to the \emph{quorum intersection property} of~\cite[Section~2.1, Condition~B2]{DBLP:journals/tocs/Lamport98}, which requires that any two quorums intersect.
Being \xxxtwined also corresponds to a quorum intersection property with a stronger requirement that (assuming there is a quorum of honest participants) any two quorums must intersect in a participant that is honest~\cite[Definition~5.1]{DBLP:journals/dc/MalkhiR98}. 
  
These intersection properties are closely tied to necessary conditions for some consensus protocols (and other distributed algorithms) to function correctly.
Specifically for the case of the semitopology $\AllBut{N}{f}$ from Definition~\ref{defn.allbut}(\ref{item.allbut}), requiring nonempty quorum intersection is equivalent to the condition \( N > 2f \), where \( f \) denotes the number of processes that may crash during the protocol's execution.
This condition is often referred to as \( Q^2 \)~\cite[Theorem~1]{DBLP:conf/asiacrypt/DamgardDFN07}.
Similarly, being \xxxtwined corresponds to the condition \( N > 3f \), commonly referred to as \( Q^3 \)~\cite[Definition 2]{DBLP:journals/dc/AlposCTZ24}.
A more generalised version of \( Q^3 \) exists, known as \( B^3 \)~\cite[Theorem 2]{DBLP:conf/asiacrypt/DamgardDFN07}\cite[Definition 5]{DBLP:journals/dc/AlposCTZ24}, which accounts for greater heterogeneity. 
Specifically, \( Q^3 \) applies to distributed systems where failures are homogeneous, meaning they are globally known and agreed upon by all participants. 
For example, all participants understand that up to \( f \) processes may fail. 
  
In contrast, in a heterogeneous setting, participants may have differing beliefs about \emph{who} might fail, leading to a more nuanced description of trust assumptions. 
In this context, \( B^3 \) generalises \( Q^3 \) by incorporating these variations in trust and failure assumptions.
Readers interested in exploring this topic further can consult~\cite{DBLP:conf/asiacrypt/DamgardDFN07,DBLP:journals/dc/AlposCTZ24}.
\end{rmrk}

\subsection{The canonical properties}

For this subsection we fix a signature $\Sigma$ and a model $\mathcal M=(\Pnt,\opens,\Val)$.

We now prove results to express how the underlying semitopology of this model $(\Pnt,\opens)$ being \xxtwined or \xxxtwined manifests itself in the logic, in ways that will be relevant in our later proofs.
In Subsection~\ref{subsect.unpack} we concretely unpack what some relevant modalities mean, then in Subsection~\ref{subsect.the.results} we prove the promised results about them.

\subsubsection{Unpacking some modalities}
\label{subsect.unpack}

Lemma~\ref{lemm.allbut.concretely} unpacks concretely what $\someOallP\phi$ and $\allOsomeP\phi$ --- which are the subject of many results including Propositions~\ref{prop.supertwined.qcbt} and~\ref{prop.intertwined.char} --- mean in the semitopology $\AllBut{N}{f}$ from Definition~\ref{defn.allbut}(\ref{item.allbut}), which has $N$ points, and has open sets being either the empty set, or any set with at least $N\minus f$ elements: 
\begin{enumerate*}
\item
$\someOallP\phi$ means `for at least $N\minus f$ elements'; and 
\item
$\allOsomeP\phi$ means `for at least $f\plus 1$ elements'.
\end{enumerate*}
This special case is particularly relevant because (as the reader may know) $\AllBut{N}{f}$ is commonly used in the literature on distributed algorithms. 
\begin{lemm}
\label{lemm.allbut.concretely}
Suppose that
\begin{itemize*}
\item
$0\leq f<N\in\Ngeqz$, 
\item
$(\Pnt,\opens)=\AllBut{N}{f}$, 
\item
$n\in\Time$ and $\avaluation$ is a valuation, and
\item
$\phi$ is a closed predicate that is pointwise in $(n,\avaluation)$.
\end{itemize*}
Then:
\begin{enumerate*}
\item
$n\mentval\someOallP\phi$ if and only if 
$\mycard\{p\in\Pnt \mid n,p\mentval\phi\}\geq N\minus f$. 
\item
$n\mentval\allOsomeP\phi$ if and only if 
$\mycard\{p\in\Pnt \mid n,p\mentval\phi\}\geq f\plus 1$. 
\item
$n\mentval\modT\someOallP\phi$ if and only if 
$\mycard\{p\in\Pnt \mid n,p\mentval\modT\phi\}\geq N\minus f$. 
\item
$n\mentval\modT\allOsomeP\phi$ if and only if 
$\mycard\{p\in\Pnt \mid n,p\mentval\modT\phi\}\geq f\plus 1$. 
\end{enumerate*}
\end{lemm}
\begin{proof}
We consider each part in turn:
\begin{enumerate}
\item
The proof is elementary but we spell it out.
We reason as follows:
$$
\begin{array}{r@{\ }l@{\quad}l}
n\mentval\someOallP\phi
\liff&
\noi(\{p\in\Pnt \mid n,p\mentval \phi\})
&\text{Proposition~\ref{prop.pointwise.dense.char}(\ref{item.dense.char.pointwise.QuorumBox})}
\\
\liff&
\mycard\{p\in\Pnt \mid n,p\mentval \phi\}\geq N\minus f
&\text{Proposition~\ref{prop.allbut.supertwined}(\ref{item.allbut.NOI.dense})}
\end{array}
$$
\item
Much as the previous part, but using Proposition~\ref{prop.pointwise.dense.char}(\ref{item.dense.char.pointwise.CoquorumDiamond}).
\item
By Lemma~\ref{lemm.someone.commutation}, $n\mentval\modT\someOallP\phi$ if and only if $n\mentval\someOallP\modT\phi$.
We just use part~1 of this result.
\item
Much as the previous part, using part~2 of this result.
\qedhere\end{enumerate}
\end{proof}

\subsubsection{The results}
\label{subsect.the.results}

\begin{prop}
\label{prop.supertwined.qcbt}
Suppose $(\Pnt,\opens)$ is a semitopology and $n\in\Time$ and $\avaluation$ is a valuation and suppose $\phi$, $\psi$, and $\chi$ are pointwise in $(n,\avaluation)$ (Definition~\ref{defn.pointwise.meaning}). 
Then: 
\begin{enumerate*}
\item\label{item.intertwined.qcbt}
If $(\Pnt,\opens)$ is \xxtwined then we have:
\begin{enumerate*}
\item\label{item.intertwined.3}
If $n\mentval \someOallP\phi$ then $n\mentval\allOsomeP\phi$. 
\item\label{item.intertwined.qcbt.2}
If $n\mentval \someOallP\phi \tand \someOallP\psi$ then $n\mentval\someOsomeP(\phi\tand\psi)$. 
\item\label{item.intertwined.qcbt.3}
If $n\mentval \someOallP\correct{\phi} \tand \someOallP\phi$ then $n\mentval\modT\someOsomeP\phi$.
\end{enumerate*}
\item\label{item.supertwined.qcbt}
If $(\Pnt,\opens)$ is \xxxtwined then we have: 
\begin{enumerate*}
\item\label{item.supertwined.qcbt.3}
If $n\mentval \someOallP\phi \tand \someOallP\psi$ then $n\mentval \allOsomeP(\phi\tand\psi)$.
\item\label{item.supertwined.qcbt.2}
If $n\mentval \someOallP\phi \tand \someOallP\psi \tand \someOallP\chi$ then $n\mentval\someOsomeP(\phi\tand\psi\tand\chi)$. 
\item\label{item.supertwined.qcbt.4}
If $n\mentval \someOallP\correct{\phi} \tand \someOallP\phi$ then $n\mentval \modT\allOsomeP\phi$.
\end{enumerate*}
\end{enumerate*}
Recall that $\correct{\phi}$ is from Definition~\ref{defn.correct}(\ref{item.correct.phi}) and just means $\modTF\phi$.
\end{prop}
\begin{proof}
We reason as follows:
\begin{enumerate}
\item
Suppose $(\Pnt,\opens)$ is \xxtwined. 
\begin{enumerate*}
\item
We reason as follows:
$$
\hspace{-2em}
\begin{array}{r@{\ }l@{\quad}l}
n\mentval\someOallP\phi
\liff&
\noi(\{p\in\Pnt \mid n,p\mentval\phi\})
&\text{Proposition~\ref{prop.pointwise.dense.char}(\ref{item.dense.char.pointwise.QuorumBox})}
\\
\limp&
\dense(\{p\in\Pnt \mid n,p\mentval\phi\}) 
&\text{Proposition~\ref{prop.noi.dense}, $(\Pnt,\opens)$ $2$-twined} 
\\
\liff&
n\mentval\allOsomeP\phi
&\text{Proposition~\ref{prop.pointwise.dense.char}(\ref{item.dense.char.pointwise.CoquorumDiamond})}
\end{array}
$$
\item
Suppose $n\mentval \someOallP\phi\tand\someOallP\psi$.
By part~\ref{item.intertwined.3} of this result, $n\mentval\allOsomeP\phi$.
By Proposition~\ref{prop.someone.implies.someoneAll}(\ref{item.someone.implies.someoneAll.4}) $n\mentval \someOsomeP(\phi\tand\psi)$ as required.
\item
Suppose $n\mentval \someOallP\correct{\phi} \tand \someOallP\phi$.
By part~\ref{item.intertwined.qcbt.2} of this result, $n\mentval\someOsomeP(\correct{\phi}\tand\phi)$.
Using %
Proposition~\ref{prop.tv.ment.TF.model}(\ref{item.tv.ment.TF.phi.lr}) $n\mentval\someOsomeP\modT\phi$.
By Lemma~\ref{lemm.someone.commutation} $n\mentval\modT\someOsomeP\phi$ as required. 
\end{enumerate*}
\item
The reasoning for the \xxxtwined case is just as for the \xxtwined case, except we need intersection of three open sets instead of two.
\qedhere\end{enumerate} 
\end{proof}

Corollary~\ref{corr.two.quorums.intersect.correct} specialises Proposition~\ref{prop.supertwined.qcbt} to the useful special case of a correct unary predicate symbol, such that $\phi=\correct{\tf P}$ (Definition~\ref{defn.correct}(\ref{item.correct})) and $\psi=\tf P(v)$ and $\chi=\tf P(v')$.
In our later examples, $\tf P$ will take values such as $\tf{accept}$, $\tf{send}$, $\tf{write}$, and $\tf{decide}$, representing basic actions in a consensus protocol:
\begin{corr}
\label{corr.two.quorums.intersect.correct}
Suppose $(\Pnt,\opens)$ is a semitopology and suppose $\tf P$ is a unary predicate symbol and $v,v'\in\Val$ and $n\in\Time$ and $\avaluation$ is a valuation.
Then:
\begin{enumerate*}
\item
If $(\Pnt,\opens)$ is \xxxtwined then:
\begin{enumerate*}
\item\label{item.two.quorums.intersect.correct.supertwined}
If $n\mentval\someOallP\correct{\tf P}\tand\someOallP\tf P(v)\tand\someOallP\tf P(v')$ then $n\mentval\modT\someOsomeP(\tf P(v)\tand\tf P(v'))$.
\item\label{item.QuorumBox.correct.supertwined}
If $n\mentval\someOallP\correct{\tf P}\tand\someOallP\tf P(v)$ then $n\mentval\modT\allOsomeP\tf P(v)$.
\end{enumerate*}
\item\label{item.two.quorums.intersect.correct.intertwined}
If $(\Pnt,\opens)$ is \xxtwined then
$n\mentval\someOallP\correct{\tf P}\tand\someOallP\tf P(v)$ implies $n\mentval\modT\someOsomeP(\tf P(v))$.
\end{enumerate*}
\end{corr}
\begin{proof}
We consider each part in turn:
\begin{enumerate}
\item[(\ref{item.two.quorums.intersect.correct.supertwined})]
Suppose $(\Pnt,\opens)$ is \xxxtwined and $n\mentval\someOallP\correct{\tf P}\tand\someOallP\tf P(v)\tand\someOallP\tf P(v')$.
By Proposition~\ref{prop.supertwined.qcbt}(\ref{item.supertwined.qcbt.2}) (since $\tf P(v)$ and $\correct{\tf P}$ are pointwise by Lemma~\ref{lemm.atomic.O}(\ref{item.atomic.pointwise.2})) $n\mentval\someOsomeP(\correct{\tf P}\tand\tf P(v)\tand\tf P(v'))$.
By two applications of Proposition~\ref{prop.tv.ment.TF.model}(\ref{item.tv.ment.TF.P}) and Lemma~\ref{lemm.commuting.connectives} (for $\tand$)
$n\mentval\someOsomeP\modT(\tf P(v)\tand\tf P(v'))$.
By Lemma~\ref{lemm.someone.commutation} $n\mentval\modT\someOsomeP(\tf P(v)\tand\tf P(v'))$ as required.
\item[(\ref{item.QuorumBox.correct.supertwined})] 
Suppose $(\Pnt,\opens)$ is \xxxtwined and $n\mentval\someOallP\correct{\tf P}\tand\someOallP\tf P(v)$.
Using Lemma~\ref{lemm.correct.P.Pv} it follows that $n\mentval\someOallP\correct{\tf P(v)}\tand\someOallP\tf P(v)$.
We just apply Proposition~\ref{prop.supertwined.qcbt}(\ref{item.supertwined.qcbt.4}). 
\item[(\ref{item.two.quorums.intersect.correct.intertwined})]
This follows from Proposition~\ref{prop.supertwined.qcbt}(\ref{item.intertwined.qcbt.3}) just as part~\ref{item.QuorumBox.correct.supertwined} does from Proposition~\ref{prop.supertwined.qcbt}(\ref{item.supertwined.qcbt.4}).
\qedhere\end{enumerate}
\end{proof}

\subsection{Further observations on intersection properties}

For this Subsection we fix a signature $\Sigma$ and a model $\mathcal M=(\Pnt,\opens,\Val)$.
\begin{rmrk}
We collect some further results as follows:
\begin{enumerate*}
\item
Proposition~\ref{prop.intertwined.char} proves converses to the implications in Proposition~\ref{prop.supertwined.qcbt}, provided that the signature is nonempty.

This is interesting because it tells us that, provided the signature is nondegenerate, Proposition~\ref{prop.supertwined.qcbt} accurately captures and reflects the intersection properties that interest us. 
\item
Proposition~\ref{prop.reformulate} is a semitopological distillation of parts~\ref{item.intertwined.qcbt.2} and~\ref{item.supertwined.qcbt.2} of Proposition~\ref{prop.supertwined.qcbt}.
Proposition~\ref{prop.reformulate} \emph{is} Proposition~\ref{prop.supertwined.qcbt} and Corollary~\ref{corr.two.quorums.intersect.correct}, in semitopological rather than logical form.
\end{enumerate*}
\end{rmrk}

In Proposition~\ref{prop.intertwined.char} we could replace every instance of the expression `$\modTF\tf P$' with the (by Definition~\ref{defn.correct}(\ref{item.correct.phi}) and Remark~\ref{rmrk.tf.synonyms}) synonymous expression `$\correct{\tf P}$'.
However, we have not set up a corresponding synonym for $\modTB\tf P$ (it would be `$\tf{valid}[\tf P]$'; we just never wrote that definition) so it seems natural to stay with $\modTF$ here.
\begin{prop}
\label{prop.intertwined.char}
Suppose $\tf P\in\ns{PredSymb}$ is a predicate constant symbol (Notation~\ref{nttn.nullary}(\ref{item.nullary.ps}); a nullary predicate symbol).\footnote{The proof works just as well for an $n$-ary predicate symbol; we just pick \emph{any} $\arity(\tf P)$-tuple of values $\vec v\in\Val^{\arity(\tf P)}$ --- which is possible because we assumed $\Val$ is nonempty in Definition~\ref{defn.model}(\ref{item.model.val}) --- and we use $\tf P(\vec v)$ in the proof instead.}
Then:
\begin{enumerate*}
\item
If\ \  $\ment(\someOallP\modTF\tf P\tand\someOallP\modTB\tf P)\timpc \someOsomeP\tf P$ then $(\Pnt,\opens)$ is \xxtwined. 
\\
(This is a form of converse to Proposition~\ref{prop.supertwined.qcbt}(\ref{item.intertwined.qcbt.3}).)
\item
If\ \  $\ment(\someOallP\modTF\tf P\tand\someOallP\modTB\tf P)\timpc \allOsomeP\tf P$ then $(\Pnt,\opens)$ is \xxxtwined. 
\\
(This is a form of converse to Proposition~\ref{prop.supertwined.qcbt}(\ref{item.supertwined.qcbt.4}).) 
\item\label{item.intertwined.char.reverse}
As a corollary, provided that the signature $\Sigma$ has at least one predicate symbol --- i.e. provided $\PredSymb$ in
it is nonempty (Definition~\ref{defn.signature}(\ref{item.predicate.syntax.empty.signature})) --- then reverse implications to those in Proposition~\ref{prop.supertwined.qcbt} hold.\footnote{No need to also insist that $\Val$ is nonempty, because we assume this in Definition~\ref{defn.model}(\ref{item.model.val}).}
\end{enumerate*}
Recall from Definition~\ref{defn.validity.judgement}(\ref{item.ment.phi}) that $\ment\phi$ means $\Forall{\acontext\in\tf{Ctx}(\Sigma,\mathcal M)}(\acontext\ment\phi)$.
\end{prop}
\begin{proof}
We consider each part in turn:
\begin{enumerate}
\item
We prove a contrapositive of \strongmodusponens.
Suppose $\opens$ is \emph{not} \xxtwined, so there exist $O,O'\in\opensne$ with $O\cap O'=\varnothing$.
Choose a valuation $\avaluation$ such that for any $n\in\Time$, 
\begin{itemize*}
\item
$\avaluation(\tf P)_{n,p}()=\tvF$ for $p\in O$, and 
\item
$\avaluation(\tf P)_{n,p}()=\tvB$ for $p\in \Pnt\setminus O$.
\end{itemize*}
The reader can check that $O\mentval \everyone\modTF\tf P$ and $O'\mentval \everyone\modTB P$, but $\mentval\everyone\modFB\tf P$ and so $\nmentval\modT\someOsomeP\tf P$.
\item
We prove a contrapositive.
Suppose $\opens$ is \emph{not} \xxxtwined, so there exist $O,O',O''\in\opensne$ with $O\cap O'\cap O''=\varnothing$.
Choose a valuation $\avaluation$ such that for any $n\in\Time$, 
\begin{itemize*}
\item
$\avaluation(\tf P)_{n,p}()=\tvT$ for $p\in O\cap O'$, and 
\item
$\avaluation(\tf P)_{n,p}()=\tvF$ for $p\in O\setminus O'$, and 
\item
$\avaluation(\tf P)_{n,p}()=\tvB$ for $p\in \Pnt\setminus O$.
\end{itemize*}
The reader can check that $O\mentval \everyone\modTF\tf P$ and $O'\mentval \everyone\modTB P$, but $O''\mentval\everyone\modFB\tf P$ and so $\nmentval\modT\allOsomeP\tf P$.
\item
Just by the arguments above.
\qedhere\end{enumerate}
\end{proof}

\begin{nttn}
\label{nttn.cover}
Suppose $\ns X$ is a set.
Then:
\begin{enumerate*}
\item
A \deffont{cover} of $\ns X$ is a set of sets $\mathcal X\subseteq\powerset(\ns X)$ such that $\bigcup\mathcal X=\ns X$.
\item
A \deffont{partition} of $\ns X$ is a cover of $\ns X$ all of whose elements are pairwise disjoint (so $\Forall{X,X'\in\mathcal X}X\cap X'=\varnothing\lor X=X'$).
\end{enumerate*}
\end{nttn}

The reader might like to try giving a proof of Proposition~\ref{prop.reformulate} using direct calculations (without building on the results we have so far).
This is possible, but it is messy.
The use of logic endows the proofs with a structure and cleanness that the direct proof cannot match: 
\begin{prop}
\label{prop.reformulate}
Suppose $(\Pnt,\opens)$ is a semitopology and $n\in\Time$ and $\avaluation$ is a valuation, and suppose $\phi$ is a closed predicate that is pointwise in $(n,\avaluation)$.
Then:
\begin{enumerate*}
\item\label{item.reformulate.1}
If $(\Pnt,\opens)$ is \xxtwined then $n\mentval\allOsomeP\modTB\phi \tor \allOsomeP\modF\phi$. 
\item\label{item.reformulate.1b}
As a corollary, if $(\Pnt,\opens)$ is \xxtwined then any two-element partition of $\Pnt$ (Notation~\ref{nttn.cover}) contains a dense element (Definition~\ref{defn.dense.subset}(\ref{item.dense.subset})).
\item\label{item.reformulate.2}
If $(\Pnt,\opens)$ is \xxxtwined then 
$n\mentval\allOsomeP\modT\phi \tor \allOsomeP\modB\phi \tor \allOsomeP\modF\phi$. 
\item\label{item.reformulate.2b}
As a corollary, if $(\Pnt,\opens)$ is \xxxtwined then any three-element partition of $\Pnt$ contains a dense element.
\end{enumerate*} 
\end{prop}
\begin{proof}
We consider each part in turn:
\begin{enumerate}
\item
By Proposition~\ref{prop.supertwined.qcbt}(\ref{item.intertwined.qcbt.2})
$n\mentval\someOallP\modF\phi\tand\someOallP\modTB\phi$ implies $n\mentval\someOsomeP(\modF\phi\tand\modTB\phi)$.
Unpacking what $n\mentval\someOsomeP(\modF\phi\tand\modTB\phi)$ means using Definition~\ref{defn.validity.judgement} and Figures~\ref{fig.3.derived} and~\ref{fig.3.phi.f}, we see that this is impossible --- for no $p$ and $O$ can we satisfy both $n,p,O\mentval\modF\phi$ and $n,p,O\mentval\modTB\phi$ together --- so that 
$n\nmentval\someOallP\modF\phi\tand\someOallP\modTB\phi$.
By Lemma~\ref{lemm.para}(\ref{item.para.em})
$n\mentval\tneg(\someOallP\modF\phi\tand\someOallP\modTB\phi)$, and simplifying further using de Morgan equivalences as per Figure~\ref{fig.easy.equivalences}, we obtain
$n\mentval\allOsomeP\modTB\phi\tor\allOsomeP\modF\phi$ as required.
\item
Suppose $\{P_1,P_2\}$ partitions $\Pnt$.
Assume a $0$-ary predicate symbol $\tf P$, and let $\avaluation$ be the valuation that sets $\avaluation(\tf P)_{n,p}=\tvT$ if $p\in P_1$ and $\avaluation(\tf P)_{n,p}=\tvF$ if $p\in P_2$.\footnote{We can assume $\tf P$ and $\avaluation$ because this result is about semitopologies, so we are free to impose any model structure on it that we like.  It seems particularly nice here how the mere \emph{existence} of a model for our logic gives us a clean proof of a purely semitopological property which, on the face of it, has nothing to do with our three-valued modal logic.}
By part~\ref{item.reformulate.1} and Lemma~\ref{lemm.tand.tor} either $n\mentval \allOsomeP\modTB\tf P$ or $n\mentval\allOsomeP\modF\tf P$.
By Proposition~\ref{prop.pointwise.dense.char}(\ref{item.dense.char.pointwise.CoquorumDiamond}) either $P_1$ or $P_2$ is dense.
\item
We reason as for the proof of part~\ref{item.reformulate.1}, using Proposition~\ref{prop.supertwined.qcbt}(\ref{item.supertwined.qcbt.2}) and starting from $\modFB\phi \tand \modTF\phi \tand \modTB\phi$. 
\item
We reason as for the proof of part~\ref{item.reformulate.1b}, for a partition $\{P_1,P_2,P_3\}$ and setting $\avaluation(\tf P)_{n,p}=\tvT$ if $p\in P_1$ and $\avaluation(\tf P)_{n,p}=\tvB$ if $p\in P_2$ and $\avaluation(\tf P)_{n,p}=\tvF$ if $p\in P_3$.
\qedhere\end{enumerate}
\end{proof}

\section{Declarative Paxos}
\label{sect.declarative.paxos}

\subsection{Case study: a simple protocol}
\label{subsect.simple}

\begin{rmrk}
\label{rmrk.simple.protocol}
Before Paxos, we will warm up by sketching a simple protocol \theory{Simple}, and then we will give four \QLogicShort axiomatisations of \theory{Simple}, reflecting four different failure assumptions (messages can/cannot get lost; participants can/cannot crash).
\theory{Simple} is as follows:
\begin{enumerate*}
\item
Any participant $p$ may broadcast a $\tf{propose}$ message to all participants.
\item
If a participant receives a $\tf{propose}$ message, then it responds with $\tf{accept}$. 
\item
If the participant $p$, having broadcast its $\tf{propose}$ message, receives an $\tf{accept}$ message, then it declares $\tf{decide}$.
\end{enumerate*}
\end{rmrk}

\begin{rmrk}
\theory{Simple} is very simple!
Only one value can be proposed, so we do not even bother to state it; we just `propose'.
A participant accepts once it hears someone propose. 
A participant that proposes, decides once it hears that someone accepted.

For comparison, in Paxos \emph{multiple values} are available, so we need to make sure that participants all decide on the \emph{same} value if they decide on anything at all. 
This makes things much harder: so Paxos has multiple rounds; there is a designated leader participant at each round; the leader must choose a value to propose; and elaborate communication is in place to ensure that (subject to suitable precisely stated assumptions) a single choice of agreed value will eventually be decided.

But that is for later.  
Even the simple protocol \theory{Simple} has plenty of subtlety, which will allow us to exhibit (in `toy' form) many of the techniques that we will then apply to Paxos:
\end{rmrk}

\begin{rmrk}
We will axiomatise \theory{Simple} subject to two parameters, yielding four possible axiomatisations:
\begin{enumerate*}
\item
Do all messages arrive, or might there be a finite initial period during which the network may be unreliable so that during that time, messages might not arrive?
\item
Can participants crash (and if so, how many)?
\end{enumerate*}
\end{rmrk}

\subsubsection*{All messages arrive; no participants crash}

First, we assume an ideal situation such that all messages arrive, and no participants crash.
We can axiomatise this as follows:
$$
\begin{array}{l@{\qquad}l}
\rulefont{SimpAccept?}&
\tf{accept} \timpc \someOsomeP\tf{propose} 
\\
\rulefont{SimpDecide?}&
\tf{decide} \timpc (\tf{propose}\tand\someOsomeP\tf{accept}) 
\\
\rulefont{SimpAccept!}&
(\someOsomeP\tf{propose}) \timpc \tf{accept}
\\
\rulefont{SimpDecide!}&
(\tf{propose}\tand\someOsomeP\tf{accept}) \timpc \tf{decide}
\\
\rulefont{Correct}&
\correct{\tf{propose},\tf{accept},\tf{decide}}
\end{array}
$$
\noindent Recall from Definition~\ref{defn.three.imp} and Proposition~\ref{prop.mp.for.tnotor} that $\timpc$ is a \emph{strong implication} and means `if the LHS is $\tvT$ then so is the RHS', and 
recall from Proposition~\ref{prop.pointwise.dense.char}(\ref{item.dense.char.pointwise.someoneAll}) that $\someOsomeP\tf P$ means `$\tf P$ somewhere', where $\tf P\in\{\tf{propose},\tf{accept},\tf{decide}\}$.

Then we explain the axioms (briefly) as follows:
\begin{itemize*}
\item
We call \rulefont{SimpAccept?} and \rulefont{SimpDecide?} \emph{backward rules}; they express backward versions of the rules in Remark~\ref{rmrk.simple.protocol}, namely that:
\begin{itemize*}
\item
if a participant asserts $\tf{accept}$ then some participant must have asserted $\tf{propose}$; and 
\item
if a participant asserts $\tf{decide}$ then it must have asserted $\tf{propose}$, and some participant must have asserted $\tf{accept}$.
\end{itemize*}
\item
We call \rulefont{SimpAccept!} and \rulefont{SimpDecide!} \emph{forward rules}.
They express forward versions of the rules in Remark~\ref{rmrk.simple.protocol}.
Because our failure assumptions are that nothing fails --- all messages arrive; no participants crash --- the forward rules are in this case just the backward rules, reversed.
\item
Axiom \rulefont{Correct} expresses that all predicates return $\tvT$ and $\tvF$ (not $\tvB$), as per Definition~\ref{defn.correct}(\ref{item.correct}). 
This reflects our assumption that no participants crash, because we model a computation crashing %
as the predicate describing it returning $\tvB$.
\end{itemize*}
Note that $\tf{propose}$, $\tf{accept}$, and $\tf{decide}$ above are not messages; they are propositions!
There is no message-passing.
$\tf{accept}$ being true at a participant represents that the participant accepts; this is implemented in the algorithm by the participant broadcasting an accept message, and represented in the logic by $\tf{accept}$ taking truth-value $\tvT$ at that participant.

Our example protocol is simple, but we hope it is already starting to illustrate how the various parts of our logic come together to express the protocol axiomatically.
We continue:

\subsubsection*{All messages arrive; participants can crash}

We assume that all messages arrive, and participants can crash --- but we assume there is a dense set of uncrashed participants.
Recall from Definition~\ref{defn.dense.subset}(\ref{item.dense.subset}) that a \emph{dense set} is one that intersects every nonempty open set; in particular, this means that every nonempty open set of participants contains a participant that is not crashed.
We model being crashed as returning the third truth-value $\tvB$, and not being crashed as returning $\tvT$ or $\tvF$.
We can axiomatise this as follows:
$$
\begin{array}{l@{\qquad}l}
\rulefont{SimpAccept?}&
\tf{accept} \timpc \someOsomeP\tf{propose} 
\\
\rulefont{SimpDecide?}&
\tf{decide} \timpc (\tf{propose}\tand\someOsomeP\tf{accept}) 
\\
\rulefont{SimpAccept!}&
(\someOsomeP\tf{propose}) \tnotor \tf{accept}
\\
\rulefont{SimpDecide!}&
(\tf{propose}\tand\someOsomeP\tf{accept}) \tnotor \tf{decide}
\\
\rulefont{Uncrashed}&
\allOsomeP\modTF\tf P
\quad \tf P\in\{\tf{propose},\tf{accept},\tf{decide}\}
\end{array}
$$ 
\noindent 
The backward rules are unchanged: if a participant accepts then from the form of the rules, it \emph{must} be the case that somebody proposed; and if a participant decides then from the form of the rules, it \emph{must} be the case that the participant proposed and some participant accepted.
 
The forward rules \rulefont{SimpAccept!} and \rulefont{SimpDecide!} have changed, and they now use the weak implication $\tnotor$.
By Definition~\ref{defn.three.imp} and Proposition~\ref{prop.mp.for.tnotor}, $\tnotor$ means `if the LHS is $\tvT$ then the RHS is $\tvT$ or $\tvB$'.
The $\tvB$ truth-value represents `crashed', so that (for example) \rulefont{SimpAccept!} expresses that for each participant $p$, if some (possibly other) participant $q$ asserts $\tf{propose}$, then $p$ asserts $\tf{accept}$ \emph{or} $p$ has crashed. 

Axiom \rulefont{Uncrashed} reflects that any quorum (nonempty open set) of participants contains an uncrashed participant; e.g. in topological language, that the set of uncrashed participants is dense.

\subsubsection*{Messages might not arrive; no participants crash}

We assume that there is an initial period during which the network might not deliver some messages (message failure might be partial, i.e. a general broadcast might arrive at some participants and not others), but the network eventually stabilises; in the literature this is called \emph{eventual synchrony}.
No participants crash.
We can axiomatise this as follows:
$$
\begin{array}{l@{\qquad}l}
\rulefont{SimpAccept?}&
\tf{accept} \timpc \someOsomeP\tf{propose} 
\\
\rulefont{SimpDecide?}&
\tf{decide} \timpc (\tf{propose}\tand\someOsomeP\tf{accept}) 
\\
\rulefont{SimpAccept!}&
\final((\someOsomeP\tf{propose}) \timpc \tf{accept})
\\
\rulefont{SimpDecide!}&
\final((\tf{propose}\tand\someOsomeP\tf{accept}) \timpc \tf{decide})
\\
\rulefont{Correct}&
\correct{\tf{propose},\tf{accept},\tf{decide}}
\end{array}
$$ 
\noindent The backward rules are unchanged, but the forward rules \rulefont{SimpAccept!} and \rulefont{SimpDecide!} now have an $\final$ modality.
This is from Figure~\ref{fig.3.derived} and it means intuitively `eventually, this holds forever everywhere', reflecting an assumption that eventually the network stabilises.

Note that \QLogicShort does not explicitly represent message-passing, but this does not matter as much as one might suppose.
The $\final$ modality captures the essence of what we need, namely, that eventually the forward rules become valid at every open set and every point. 

\subsubsection*{Messages might not arrive; participants can crash}

We assume that messages might not arrive, and participants can crash but there is a dense set of uncrashed participants (so that every quorum of participants contains at least one uncrashed participant).
We can axiomatise this just by combining the ideas above, as follows:
$$
\begin{array}{l@{\qquad}l}
\rulefont{SimpAccept?}&
\tf{accept} \timpc \someOsomeP\tf{propose} 
\\
\rulefont{SimpDecide?}&
\tf{decide} \timpc (\tf{propose}\tand\someOsomeP\tf{accept}) 
\\
\rulefont{SimpAccept!}&
\final((\someOsomeP\tf{propose}) \tnotor \tf{accept})
\\
\rulefont{SimpDecide!}&
\final((\tf{propose}\tand\someOsomeP\tf{accept}) \tnotor \tf{decide})
\\
\rulefont{Uncrashed}&
\allOsomeP\modTF\tf P
\quad \tf P\in\{\tf{propose},\tf{accept},\tf{decide}\}
\end{array}
$$ 
\noindent The axiomatisations of Declarative Paxos and Simpler Declarative Paxos in Figures~\ref{fig.ThyPaxOne} and~\ref{fig.ThySPax} are more elaborate than the axioms above, because Paxos is a more complicated protocol, but the basic ideas are as above; just scaled up.

\subsection{A high-level view of the Paxos algorithm}

Paxos is notorious for being hard to understand.
The reader can find detailed treatments in \cite{DBLP:journals/tocs/Lamport98,lamport2001paxos,cachin2011paxos}.
A Wikipedia article~\cite{wiki:Paxos_(computer_science)} is accessible and clear, and a clear modern description is in Section~4 of~\cite{PadonLosaSagivShoham2017}. 

\begin{rmrk}[High-level view of Paxos]
\label{rmrk.high-level.paxos}
Paxos allows participants to crash, and messages to fail to arrive at their destination, but it assumes all participants are honest (i.e. that if they have not crashed then they follow the rules of the algorithm).

Here is a high-level view of the algorithm:
\begin{enumerate*}
\item\label{item.l.proposes}
\deffont{Leader proposes:}\quad

A designated leader participant (if it does not crash) sends messages to all participants, \emph{proposing} some unique value $v$.

Message-passing is unreliable, so messages may fail to arrive.
\item\label{item.p.sends}
\deffont{Participant sends:}\quad

If a participant $p$ receives a propose $v$ message from the leader --- it is assumed that every participant knows who the leader is, and other participants cannot forge the leader's signature --- then $p$ sends a \emph{send} message back to the leader, which carries the value $v_p$ that $p$ most recently accepted, along with the timestamp of when $v_p$ was accepted.\footnote{For reasons of compact presentation we do not explicitly represent the timestamps in our axiomatisation; i.e. $\tf{send}$ is a unary predicate in Definition~\ref{defn.declarative.paxos}(\ref{item.ThyPaxOne.sig}), not a binary predicate.
If this bothers the reader then we can just say that, actually, $v_p$ contains a timestamp of the stage $n\in\Time$ at which it was accepted; i.e. that $v_p$ is actually implemented as a pair $(v_p',n_p)$ of a value and a time.} 
If $p$ has never accepted any value, so no most recently accepted value exists, then $p$ sends back a special $\udfn$ value.
\item\label{item.l.writes}
\deffont{Leader writes:}\quad

If the leader receives messages from a quorum and at least one of these messages reports some previously accepted value, then the leader \emph{writes} a most recently accepted value. 
If every reply in that quorum reports $\udfn$, then the leader writes the value $v$ that it proposed in step~\ref{item.l.proposes}.
If the leader receives replies from no quorum, then it writes nothing. 
\item\label{item.p.accepts}
\deffont{Participant accepts:}\quad

If a participant receives a write message from the leader for some value, then it sends a message back \emph{accepting} that value.
\item\label{item.l.decides}
\deffont{Leader decides:}\quad

If the leader receives accept messages for some value from a quorum of participants, it decides on that value, and sends out a \emph{decide} message to all participants with that value.
\item\label{item.p.decides}
\deffont{Participant (non-leader) decides:}\quad

If a participant receives a decide message from the leader for some value, then it decides on that value and stops.
\end{enumerate*}
\end{rmrk}

\begin{rmrk}[High-level view of our axiomatisation]
\label{rmrk.elisions}
The theory $\ThyPax$ in Definition~\ref{defn.ThyPaxOne} and Figure~\ref{fig.ThyPaxOne} is a declarative abstraction of Paxos-the-algorithm.
It stands in relation to implementations of Paxos in roughly the same way that (say) the mathematical specification of the Fibonacci function stands in relation to its concrete implementations, in the sense that we elide implementational detail to preserve (and reveal) a mathematical ideal of what the implementation does.

\ThyPax elides two features of the concrete implementation:
\begin{enumerate*}
\item
There is no explicit message-passing; there are only predicates returning truth-values.

This may seem counterintuitive --- isn't Paxos \emph{supposed} to be a message-passing algorithm? --- but we shall see that a lot of the character of Paxos survives.
Indeed, arguably (and certainly for the first author) it is much easier to see what is really going on once we have done this.
\item\label{item.elide.algorithmic.time}
There is no abstract machine, and therefore no \deffont{algorithmic time}, by which we mean no sequence of state-changes of the abstract machine. 

The only notion of time is the \emph{stage} $n\in\Time$ from the context as defined in Definition~\ref{defn.contexts}(\ref{item.context.stage}).
We may call this \deffont{logical time} and (when we are modelling Paxos) it represents the sequence of what the distributed consensus algorithms literature often calls \emph{stages}, or \emph{ballot numbers}.\footnote{In a distributed algorithm, at any point in algorithmic time, different participants in the algorithm may be at different points in logical time.  A practical example: consider several participants racing across 100 metres. In logical time, they are just counting to 100.  In algorithmic time, one of them will get to 100 \emph{first}, and win.  In this paper, we only care about the logical aspects of the process; there may still be a winner participant, but this would be implemented just as a nullary predicate $\tf{winner}$ which is true at at most one participant.}
\end{enumerate*}
We will prove standard Paxos correctness properties for \ThyPax in Section~\ref{sect.paxos.correctness.properties}. 
\end{rmrk}

\subsection{Definition of Declarative Paxos}

\begin{figure}
$$
\begin{array}{l@{\quad}r@{\ }l}
\figunderline{Backward rules, for backward inference} 
\figskip
\rulefont{PaxPropose?}&
& \tall v.\ \tf{propose}(v)\timpc(\tf{leader}\tand v\tneq\udfn)
\figskip
\rulefont{PaxSend?}&
& \tall v.\ \tf{send}(v) \timpc 
\begin{array}[t]{l}
\someOsomeP\texi\tf{propose}\ \tand 
\\
\bigl(
(\mru{\tf{accept}}{v})\ \tor  
(\tneg\modT\recent\texi\tf{accept}\tand v\teq\udfn)\bigr)
\end{array}
\figskip
\rulefont{PaxWrite?}&
& \tall v.\ \tf{write}(v) \timpc 
\begin{array}[t]{l}
\tf{leader}\tand 
\\
\Quorum(\Box\texi\tf{send} \tand 
\\
\qquad\bigl((\mru{(\somewhere\tf{accept})}{v}) \tor
(\tf{propose}(v)\tand\tneg\modT\recent\texi\somewhere\tf{accept})\bigr)) 
\end{array}
\figskip
\rulefont{PaxAccept?}&
& \tall v.\ \tf{accept}(v) \timpc \someOsomeP\tf{write}(v)
\figskip
\rulefont{PaxDecideL?}&
& \tall v.\ \tf{leader} \timpc \tf{decide}(v) \timpc \someOallP\tf{accept}(v)
\figskip
\rulefont{PaxDecide\tneg L?}&
& \tall v.\ \tneg\tf{leader} \timpc \tf{decide}(v) \timpc \someOsomeP(\tf{leader}\tand\tf{decide}(v))
\figskip\figskip
\figunderline{Forward rules, for making progress} 
\figskip
\rulefont{PaxPropose!}&
& \tf{leader}\tnotor\texi\tf{propose}
\figskip
\rulefont{PaxSend!}&
& \final(\someOsomeP\texi\tf{propose} \tnotor \texi\tf{send}) 
\figskip
\rulefont{PaxWrite!}&
& \final\bigl((\tf{leader}\tand (\someOallP\texi\tf{send}) \tand \texi\tf{propose}) \tnotor \texi\tf{write}\bigr)
\figskip
\rulefont{PaxAccept!}&
& \final(\texi\someOsomeP\tf{write}\tnotor \texi\tf{accept}) 
\figskip
\rulefont{PaxDecideL!}&
& \final((\tf{leader} \tand\texi\someOallP\tf{accept}) \tnotor \texi\tf{decide}) 
\figskip
\rulefont{PaxDecide\tneg L!}&
& \final((\tneg\tf{leader} \tand \someOsomeP(\tf{leader}\tand\texi\tf{decide}))\tnotor \texi\tf{decide}) 
\figskip
\figskip
\figunderline{Other rules, for correctness \& liveness}
\figskip
\rulefont{LdrExist}&
&\modTF\tf{leader}\tand\modT\someOsomeP\tf{leader}
\figskip
\rulefont{LdrUniq}&
& \somewhere\mu\tf X.\tneg\somewhere(\tf{leader}\tand\tf X) 
\figskip
\rulefont{LdrCorrect}&
& \infinitely\someOsomeP(\tf{leader}\tand \correct{\tf{propose},\tf{write},\tf{decide}}) 
\figskip
\rulefont{PaxPCorrect}&
& 
\begin{array}[t]{l}
\someOallP\correct{\tf{propose}} \tand \someOallP\correct{\tf{send}} \tand \someOallP\correct{\tf{accept}} \tand 
\\
\qquad\someOallP\correct{\tf{write}} \tand \someOallP\correct{\tf{decide}}
\end{array}
\figskip
\rulefont{PaxPropose01}&
& \texiaffine \tf{propose}
\figskip
\rulefont{PaxWrite01}&
& \texiaffine\tf{write}
\figskip
\rulefont{Pax2Twined}&
& \tall v.\ \someOallP\tf{accept}(v) \timpc \allOsomeP\tf{accept}(v)
\figskip
\rulefont{PaxNontriv}&
& \texi v.\ v\tneq\udfn
\end{array}
$$
\caption{$\Theta_\ThyPax$: Axioms of Declarative Paxos (Definition~\ref{defn.declarative.paxos})}
\label{fig.logical.paxos}
\label{fig.ThyPaxOne}
\end{figure}

Recall from Definition~\ref{defn.theory.axiom}(\ref{item.theory}) that a \emph{theory} is a pair of a \emph{signature} and a set of \emph{axioms} over that signature.
Then:
\begin{defn}
\label{defn.declarative.paxos}
\label{defn.ThyPaxOne}
We define \deffont{Declarative Paxos} to be the logical theory 
$$
\ThyPax=(\Sigma_{\ThyPax},\Theta_{\ThyPax}),
$$ 
where:
\begin{enumerate*}
\item\label{item.ThyPaxOne.sig}
the signature $\Sigma_{\ThyPax}$ is 
$$
\Sigma_{\ThyPax}=[\tf{leader}:0;\ \tf{propose},\tf{send},\tf{write},\tf{accept},\tf{decide}:1] ,
$$ 
as per the notation in Definition~\ref{defn.signature}(\ref{item.predicate.syntax.signature.list.form}), and
\item\label{item.ThyPaxOne.axioms}
the axioms $\Theta_{\ThyPax}$ are as written in Figure~\ref{fig.logical.paxos}.
\end{enumerate*}
In axioms \rulefont{PaxPropose?}, \rulefont{PaxSend?} and \rulefont{PaxNontriv} we assume an element $\udfn\in\Val$, to model the `undefined' value from Paxos (it does not matter what $\udfn\in\Val$ is; we just have to choose some element and use it consistently).
\end{defn}

\begin{rmrk}
\label{rmrk.functional.notation.in.paxos}
For convenient reference, we give pointers to help unpack the notation used in Figure~\ref{fig.logical.paxos}.
Note that this is just an overview focussing on items of notation used in the axioms, to help with reading them; in-depth discussions of the axioms will follow: 
\begin{enumerate*}
\item
We use the functional notation from Notation~\ref{nttn.hos} to elide quantified variables where convenient.
For example, axioms \rulefont{PaxPropose!} and \rulefont{PaxAccept!} written in full are 
$$
\tf{leader}\tnotor\texi\thea.\tf{propose}(\thea)
\quad\text{and}\quad
\final((\texi\thea.\someOsomeP\tf{write}(\thea))\tnotor\texi\thea.\tf{accept}(\thea)) .
$$
\item
$\tnotor$ and $\timpc$ are the \emph{weak} and \emph{strong} implications respectively.
Truth-tables are in Figure~\ref{fig.3} and a discussion is in Subsection~\ref{subsect.3.imp}.
In particular, by Proposition~\ref{prop.mp.for.tnotor}, 
\begin{itemize*}
\item
$\phi\tnotor\phi'$ means `if $\phi$ is true then $\phi'$ is valid' (where `valid' means `$\tvT$ or $\tvB$') and 
\item
$\phi\timpc\phi'$ means `if $\phi$ is true then $\phi'$ is true'. 
\end{itemize*}
\item
The compound modalities --- $\someOsomeP$, $\allOallP$, $\someOallP$, and $\allOsomeP$ --- are discussed in Subsection~\ref{subsect.compound.modalities}.
Their meaning in full generality is in Lemma~\ref{lemm.dense.char}, but because the axioms only apply them to \emph{pointwise} predicates, their meanings simplify to `somewhere' ($\someOsomeP$), `everywhere' ($\allOallP$), `for a set with a nonempty open interior' ($\someOallP$), and `for a dense set' ($\allOsomeP$) respectively, as per Proposition~\ref{prop.pointwise.dense.char}. 
\item\label{item.discuss.final}
$\final$ means $\sometime\allOallP\forever$ and is from Figure~\ref{fig.3.derived}.

We use this in our axiomatisation to reflect a standard assumption in distributed systems of \deffont{eventual network synchrony}, or synonymously \deffont{global stabilisation (of the network)}, meaning that the communications network may be unstable and fail to pass messages at first, but then after a finite but unbounded amount of time, it becomes stable and messages arrive.
Our logic does not explicitly represent messages, but the use of $\final$ in our forward rules reflects this assumption about the underlying implementation. 
\item
$\texiaffine$ in \rulefont{PaxPropose01} and \rulefont{PaxWrite01} means `for zero or one elements', as per Subsection~\ref{subsect.unique.affine.existence}, so these two axioms intuitively state that each participant proposes or writes at most one value.

Note that our logic is three-valued, so what these axioms mean \emph{precisely} is something a little richer than is the case in two-valued logic.
In particular, if $\Forall{v}\modellabel{\tf{write}(v)}{\acontext}=\tvB$ --- corresponding to the case where a participant is crashed --- then $\acontext\ment\texiaffine\tf{write}$, \emph{even though} $\acontext\ment\tf{write}(v)$ holds for every $v$. 
See the discussion and results in Subsection~\ref{subsect.unique.affine.existence}.
\item
$\correct{\tf{propose}}$ means that $\tf{propose}$ returns $\tvT$ or $\tvF$ (not $\tvB$) on every value; see Definition~\ref{defn.correct}(\ref{item.correct}).
Similarly for the other unary predicate symbols in \rulefont{PaxPCorrect}.
\item
$\mru{\tf{accept}}{v}$ means `a most recently accepted value is $v$'.
See the discussion and definition in Subsection~\ref{subsect.mrup}.
\item
Finally, note that we intend the predicates in Figure~\ref{fig.logical.paxos} to be read as axioms.

What it means for an axiom to be valid in a model --- i.e. how we judge whether a model satisfies $\Theta_\ThyPax$, or how we judge whether a predicate is valid in all models of $\Theta_\ThyPax$ --- is unpacked in Figure~\ref{fig.validity} and Definition~\ref{defn.theory.axiom}(\ref{item.axioms.valid}).
We discuss the quantifier scoping in those definitions in Remark~\ref{rmrk.thy.quantifier.scope}.
\end{enumerate*}
\end{rmrk}

\begin{nttn}
Technically, $\ThyPax$ the theory `Declarative Paxos' from Definition~\ref{defn.ThyPaxOne} and Figure~\ref{fig.ThyPaxOne} consists of a pair of $\Sigma_{\ThyPax}$ the \emph{signature} of Declarative Paxos, and $\Theta_\ThyPax$ the \emph{axioms} of Declarative Paxos.  

For clarity --- this is a standard notational shorthand --- we may identify $\ThyPax$ with its set of axioms where convenient.
For instance we may write $\mentval\ThyPax$ as shorthand for $\mentval\Theta_\ThyPax$; see for example Lemma~\ref{lemm.validity.helper} and Theorem~\ref{thrm.validity}.
\end{nttn}

\subsection{High-level discussion of the axioms}

\begin{rmrk}
The predicates in Figure~\ref{fig.logical.paxos} are intended as a \emph{theory}, in the sense of Definition~\ref{defn.theory.axiom}, which is why we call them \emph{axioms} and not just `some predicates'.

That is: we fix a signature $\Sigma$ that has (at least) the predicate symbols required to express the axioms of the theory $\Theta_\ThyPax$ in Figure~\ref{fig.logical.paxos}, and we fix a model $\mathcal M$, and we derive properties that hold in all the valuations $\avaluation$ to $\mathcal M$ such that $\Theta_\ThyPax$ is valid in the sense of Figure~\ref{fig.validity}.
 
At least as much design effort went into crafting the axioms $\Theta_\ThyPax$ of Declarative Paxos in Figure~\ref{fig.logical.paxos}, as went into designing the three-valued modal fixpoint logic in which the axioms are expressed.\footnote{Of course there was design feedback in both directions; we put into the logic what we needed to express the axioms, and we refined the axioms to elegantly work within the logic.} 
Therefore, we will take our time to discuss the axioms in detail.
The reader who just wants to see proofs can skip the rest of this Section (for now) and look to Section~\ref{sect.paxos.correctness.properties} onwards.
\end{rmrk}

\begin{rmrk}
\label{rmrk.axioms.paxos.three.groups}
We have organised the axioms in Figure~\ref{fig.logical.paxos} into three groups, which we now discuss:
\begin{enumerate}
\item
\emph{Backward (inference) rules.}

Axioms with names ending with $?$ represent backwards reasoning and so we call them \deffont{backward rules}.
They use the strong implication $\timpc$ (Definition~\ref{defn.three.imp}(\ref{item.three.elementary.strong.implication})), which means `if the LHS is true ($\tvT$) then the RHS is true ($\tvT$)'.

The backward rules represent a strong inference that \emph{if} some $Y$ is observed to be $\tvT$ then it \emph{must} be because some other property $X$ is $\tvT$.
For instance, \rulefont{PaxAccept?} says that if $\tf{accept}(v)$ is true, then $\someOsomeP\tf{write}(v)$ must also be true.

We discuss the backward rules individually in Remark~\ref{rmrk.discuss.backward.rules} below.
\item
\emph{Forward (progress) rules.}

Axioms with names ending with $!$ represent forwards reasoning about making progress, so we call them \deffont{forward rules}.
They use $\final$ and the weak implication $\tnotor$ (Definition~\ref{defn.three.imp}(\ref{item.three.elementary.weak.implication})), which means `if the LHS is true ($\tvT$) then the RHS is valid ($\tvT$ or $\tvB$)'.

We derive the forward rules from the backward rules using a systematic recipe which we discuss in Remark~\ref{rmrk.recipe}.

The forward rules represent a weak inference that \emph{if} some precondition $X$ is $\tvT$ then --- provided the network is stable and provided that the participant has not crashed (where `crashed' is interpreted as returning $\tvB$) --- then the postcondition $Y$ is $\tvT$. 

For instance, \rulefont{PaxAccept!} expresses that, after network synchrony (Remark~\ref{rmrk.functional.notation.in.paxos}(\ref{item.discuss.final})), if $\texi\someOsomeP\tf{write}$ is true (by Notation~\ref{nttn.hos} this is shorthand for $\texi a.\someOsomeP\tf{write}(a)$) then if the participant has not crashed then $\texi\tf{accept}$ (i.e. $\texi a.\tf{accept}(a)$) is true. 

We discuss the forward rules in Remark~\ref{rmrk.discuss.forward.rules} below.
\item
\emph{Other rules, for correctness and liveness.}

Other rules reflect particular conditions of the algorithms.
For example, \rulefont{PaxNontriv} asserts that there exists a value other than $\udfn$.
Recall that $\udfn$ is a value reserved by the later stages of the protocol to indicate the absence of a previously accepted value, so we wish to exclude a pathological model in which the protocol has participants and a leader, but there are no non-$\udfn$ values for it to propose, so that it is impossible for the leader to behave correctly (even if it wants to).
This would not actually be mathematically wrong, but it is not what we intend. 
We discuss the other rules in Remark~\ref{rmrk.paxos.axioms.discussion}.
\end{enumerate}
\end{rmrk}

\begin{rmrk}[Recipe for deriving forward rules from backward rules]
\label{rmrk.recipe}
A design feature of Figure~\ref{fig.logical.paxos} is that each forward rule is obtained from the corresponding backward rule by a \deffont{recipe}.
Namely: for a backward rule of the form 
$$
\tall v.\ \chi \timpc \psi[a\ssm v] \timpc \phi[a\ssm v]
$$
where 
\begin{itemize*}
\item
$\chi$ expresses some possible \emph{side-conditions} (e.g. $\chi=\tf{leader}$ in \rulefont{PaxDecideL?}) and 
\item
$\psi$ expresses some \emph{postconditions}, and 
\item
$\phi$ expresses some \emph{preconditions}, 
\end{itemize*}
our recipe constructs a forward rule of the form
$$
(\texi a.\,\chi \tand \phi)\tnotor \texi a.\psi .
$$
In most of the forward rules this result also gets wrapped in an $\final$ modality, representing network synchrony  (Remark~\ref{rmrk.functional.notation.in.paxos}(\ref{item.discuss.final})).

This recipe holds for every pair of rules:
\begin{itemize*}
\item
We see it cleanly in the pairs \rulefont{PaxAccept?} / \rulefont{PaxAccept!}, \rulefont{PaxDecideL?} / \rulefont{PaxDecideL!}, and \rulefont{PaxDecide\tneg L?} / \rulefont{PaxDecide\tneg L!}.
These rules follow our recipe symbol-for-symbol.
\item
We see this also with \rulefont{PaxPropose?} / \rulefont{PaxPropose!}.
Our recipe generates the following version of \rulefont{PaxPropose!}: 
$$
(\texi v.\,\tf{leader}\tand v\tneq\tf{undef})\tnotor\texi v.\tf{propose}(v).
$$
The reader can check that in the presence of \rulefont{PaxNontriv}, this is logically equivalent to what is written in Figure~\ref{fig.logical.paxos}.
\item 
It is less clear how our recipe works for \rulefont{PaxSend?} and \rulefont{PaxSend!} and \rulefont{PaxWrite?} and \rulefont{PaxWrite!}; e.g. \rulefont{PaxSend!} looks different from \rulefont{PaxSend?} --- it is clearly simpler!
Remarkably, these pairs too are consistent with our recipe; see Remark~\ref{rmrk.why.evaporates} and the subsequent discussion and results.
\end{itemize*}
\end{rmrk}

\begin{rmrk}
We make some further elementary comments on our recipe from Remark~\ref{rmrk.recipe}:
\begin{enumerate*}
\item
We uncurry the forward rules %
because this is convenient for some proofs. 
\item
A forward rule of the form $\texi v.((\chi \tand \phi)\tnotor \psi)$ would be \emph{wrong}.
This is just a different predicate and is not what is intended.
\item
A forward rule of the form $\tall v.((\chi \tand \phi)\tnotor \psi)$ would not necessarily be wrong, but it is stronger than we need for our proofs.
This reflects a general design principle in logical theories that \emph{weaker axioms are better}.
If we can prove the same result from weaker axioms then generally speaking we should, because that is a stronger result.\footnote{In practice, working from weaker axioms tends to force us to think harder and so design a better-engineered proof.}
\end{enumerate*}
\end{rmrk}

\subsection{More detailed discussion of the axioms}

\begin{rmrk}
\label{rmrk.paxos.axioms.discussion}
We discuss the \emph{axioms for correctness and liveness}:
\begin{enumerate}
\item\label{item.discuss.LdrExist}
\rulefont{LdrExist}:

In implementations of Paxos, at each stage there is a unique \emph{leader}, whose identity is computed by applying (what amounts to) a random number generator to an agreed seed.
At the level of abstraction in this paper, we just assume that a leader exists --- we do not care how it is chosen.

We use a predicate constant symbol $\tf{leader}$ to identify the leader.
Because our logic is modal --- across time $n\in\Time$ and participants $p\in\Pnt$ --- we can just assume that a constant symbol $\tf{leader}$ is assigned truth-value $\tvT$ where a participant $p$ is a leader at a stage $n$, and is assigned truth-value $\tvF$ where a participant $p$ is not a leader at a stage $n$.

In more technical language: $\avaluation$ is such that $n,p\mentval \modTF\tf{leader}$ for every time $n\in\Time$ and participant $p\in\Pnt$, %
and $p$ is a leader %
precisely when $n,p\mentval\tf{leader}$.
\rulefont{LdrExist} just asserts that a participant either is a leader or it is not, and some leader exists.
\item\label{item.discuss.LdrExistUniq}
\rulefont{LdrUniq}:

This axiom expresses that there is a unique leader, as per the standard assumption of Paxos.
Technically, \rulefont{LdrUniq} expresses $\exists_{01} p\in\Pnt.\tf{leader}(p)$ --- that there exist zero or one leaders --- but combined with \rulefont{LdrExist} this means that there exists a unique leader.
The axiom does not use an explicit quantification over points because our logical syntax in Figure~\ref{fig.predicate.syntax} does not have this.
Perhaps surprisingly, it is still possible to express uniqueness of the leader. 
We defer a discussion to Subsection~\ref{subect.ldrunique}.

It may also be worth noting that \rulefont{LdrUniq}, though it expresses a standard Paxos assumption, turns out to be stronger than we need for our proofs; more on this in Subsection~\ref{subsect.ldrexistuniq'}.
\item
\rulefont{LdrCorrect}:

\rulefont{LdrCorrect} asserts that for infinitely many times $n\in\Time$, the leader $l$ (which exists by the right-hand conjunct in \rulefont{LdrExist}) is $\tf{propose}$-correct, $\tf{write}$-correct, and $\tf{decide}$-correct (terminology from Definition~\ref{defn.correct}(\ref{item.P-correct})).

Unpacking Definition~\ref{defn.correct}(\ref{item.P-correct}), this means that $n,l\mentval \modTF\tf{propose}(v)$ for every $v\in\Val$, and similarly for $\tf{write}$ and $\tf{decide}$.
Intuitively, this is a \emph{liveness assumption} that there is always a future time when the leader is not crashed in the sense that it can correctly $\tvT\tvF$-propose, $\tvT\tvF$-write, and $\tvT\tvF$-decide a value. 
\item
\rulefont{PaxPCorrect}:

Given a valuation $\avaluation$, we unpack Definition~\ref{defn.validity.judgement}(\ref{item.mentval}) to see what $\mentval \rulefont{PaxPCorrect}$ actually means, as per Definition~\ref{defn.theory.axiom}:
\begin{multline*}
\Forall{n\in\Time}\Exists{O\in\opensne}\Forall{p\in O}\Forall{v\in\Val}n,p,O\mentval\correct{\tf P(v)}
\\
\text{for each}\ \tf P\in\{\tf{propose},\tf{send},\tf{accept},\tf{write},\tf{decide}\} .
\end{multline*}
This corresponds to a standard correctness assumption in Paxos that at each stage there is a quorum of uncrashed participants.

However, the axiom \rulefont{PaxPCorrect} is subtly more general than the standard correctness assumption, because there is no restriction that the \emph{same} open set $O$ be used for each unary predicate symbol.\footnote{If we wanted the same $O$, then we would scope the $\someOallP$ modality further out and write $\someOallP(\correct{\tf{propose}} \tand \correct{\tf{send}} \tand \correct{\tf{accept}} \tand 
\correct{\tf{write}} \tand \correct{\tf{decide}})$.  This turns out to be unnecessary.}
It turns out that the proofs require only this assumption, and no more. 
\item\label{item.01.discussion}
\rulefont{PaxPropose01} and \rulefont{PaxWrite01}:

We justify these axioms with respect to the Paxos algorithm as summarised in Remark~\ref{rmrk.high-level.paxos}:
\begin{itemize*}
\item
Step~\ref{item.l.proposes} of Remark~\ref{rmrk.high-level.paxos} states that the leader is crashed or proposes one value.
The intuitive correspondence to \rulefont{PaxPropose01} is clear.
\item
Step~\ref{item.l.writes} of Remark~\ref{rmrk.high-level.paxos} states that the leader writes a value derived from \emph{the first} quorum of participants that it hears back from (or none, if it hears back from no such quorum). 
One corollary of this is that it writes at most one value.
Our logical theory has no state machine and notion of algorithmic time (as per Remark~\ref{rmrk.elisions}(\ref{item.elide.algorithmic.time})) --- so there is no sense in which information from one open set `arrives' at the leader `before' information from another --- so from the point of view of the logic, of all the open sets available, the leader picks at most one to act on.
This is axiom \rulefont{PaxWrite01}. 
\end{itemize*}
It is a fact that \rulefont{PaxPropose01}+\rulefont{PaxPropose!} and \rulefont{PaxWrite01}+\rulefont{PaxWrite!} give us
$$
\tf{leader}\tnotor\texiunique\tf{propose}
\quad\text{and}\quad
\final((\tf{leader}\tand(\someOallP\texi\tf{send})\tand\texi\tf{propose})\tnotor \texiunique\tf{write}) ,
$$ 
and
\rulefont{PaxPropose01}+\rulefont{PaxPropose?} and \rulefont{PaxWrite01}+\rulefont{PaxWrite?} give us
$$
\tneg\tf{leader}\tnotor\tneg\texi\tf{propose}
\quad\text{and}\quad
\tneg\tf{leader}\tnotor\tneg\texi\tf{write} .
$$
So it follows from our axioms that a leader participant proposes and writes one value (assuming it is not crashed, and if messages arrive), and a non-leader participant proposes and writes zero values. 
\item
\rulefont{Pax2Twined}:

This axiom collects a set of instances of Proposition~\ref{prop.supertwined.qcbt}(\ref{item.intertwined.3}), namely taking $\phi=\tf{accept}(v)$ in that Lemma for every $v\in\Val$.

There is a design choice here: instead of an axiom \rulefont{Pax2Twined}, we could insist specifically that the semitopology in our underlying model be $2$-twined.
Then \rulefont{Pax2Twined} ceases to be an axiom, but we still have it as a fact of the model by directly invoking Proposition~\ref{prop.supertwined.qcbt}(\ref{item.intertwined.3}).
Or we could be even more specific: we could insist that the underlying semitopology be $\AllBut{N}{f}$ (Definition~\ref{defn.allbut}) for $N> 2*f$; by Proposition~\ref{prop.allbut.supertwined}(\ref{item.allbut.intertwined}) this is $2$-twined, and then we can use Proposition~\ref{prop.supertwined.qcbt}(\ref{item.intertwined.3}) to obtain \rulefont{Pax2Twined}.

We choose here to express as an axiom precisely the property that our proofs require (we use \rulefont{Pax2Twined} just once, in the Agreement property Proposition~\ref{prop.integrity.helper}).
This is concise, general, and yields clean proofs.

Insisting on a $2$-twined semitopology, or specifically on the semitopology $\AllBut{N}{f}$ for $N>2*f$, would also make sense --- but the axiom is very precise, allowing us to pinpoint what logical property is required and where. 
\end{enumerate}
\end{rmrk}

\begin{rmrk}
\label{rmrk.discuss.backward.rules}
We discuss the \emph{axioms for the backward rules}:
\begin{enumerate}
\item
\rulefont{PaxPropose?}:

If we propose some value $v$ then 
\begin{enumerate}
\item
we are the leader and 
\item
$v$ is not the undefined value $\udfn$.
\end{enumerate}
\item
\rulefont{PaxSend?}:

If we send some value $v$ then 
\begin{enumerate}
\item
somebody must have proposed some value and 
\item
\begin{enumerate*}
\item
\emph{either} $v$ is the value we have most recently accepted, 
\item
\emph{or} we never accepted any value and $v$ is the undefined value.
\end{enumerate*}
\end{enumerate}
\item
\rulefont{PaxWrite?}:

If we write some value $v$ then 
\begin{enumerate}
\item
we are a leader, and
\item
there exists an open set of points (= participants) that have all sent us values, and
\item
\begin{enumerate*}
\item
\emph{either} some point in that open set has most recently accepted $v$, 
\item
\emph{or} 
$v$ is the value that we proposed, and no point in that open set has ever accepted any value. 
\end{enumerate*}
\end{enumerate}
\item
\rulefont{PaxAccept?}:

If we accept $v$ then somebody somewhere must have written $v$.
\item
\rulefont{PaxDecideL?}:

If a leader decides $v$ then an open set of points must have all accepted $v$.
\item
\rulefont{PaxDecide{\tneg}L?}:

If we are not a leader and we decide $v$, then some leader somewhere must have decided $v$. 
\end{enumerate}
\end{rmrk}

\begin{rmrk}
\label{rmrk.discuss.forward.rules}
We discuss the \emph{axioms for the forward rules}.

We noted in Remark~\ref{rmrk.recipe} our design methodology of deriving the forward rules from the backward rules; forward rules are rules for making progress, so they are existentially quantified over the variable $v$ and they use the weak implication $\tnotor$ to emulate possible crashes.

The other subtlety is the $\final$ modality wrapping all of the forward rules, except for \rulefont{PaxPropose!}.
\begin{itemize*}
\item
\emph{What $\final$ does for us in most of the forward rules.}

It models possible network failure, and eventual network synchrony.
As touched on in Remark~\ref{rmrk.functional.notation.in.paxos}(\ref{item.discuss.final}), the $\final$ modality reflects the possibility that the network might be down for a while, but eventually --- `finally' --- the network stabilises and messages arrive.

Predicates in our logic are \emph{not} messages, but we get the effect of eventual network synchrony in our logic (even though messages are not explicitly present) by wrapping the implications in our forward rules in the $\final$ modality.

More on this when we introduce the notion of \emph{global stabilisation of logical time} in Subsection~\ref{subsect.gslt}.
\item
\emph{Why $\final$ is absent from \rulefont{PaxPropose!}.}

Because this rule models a property of the implementation that involves no network communication.
The leader does not need to communicate with any other participant to pick a value to propose.

It is true that the leader's propose message might not arrive if the network is down --- but this is reflected by the use of $\final$ in the \rulefont{PaxSend!} rule.
\end{itemize*}
\end{rmrk}

\subsection{Some fiddly observations}

\subsubsection{An aside on \rulefont{PaxWrite?}}

\rulefont{PaxWrite?} from Figure~\ref{fig.ThyPaxOne} contains this predicate:
$$
\mru{(\somewhere\tf{accept})}{v} .
$$
Unpacking what this means (cf. Remark~\ref{rmrk.explain.mrup}), this looks into the past to check whether $v$ is a value that was most recently accepted by some participant.

Note that this slightly differently scoped predicate would be \emph{wrong}:
$$
\somewhere(\mru{\tf{accept}}{v}) .
$$
This finds some participant such that $v$ is a value that was most recently accepted by \emph{that} participant (but there might be another participant who more recently accepted some other value).

The reader with a background in logic will know all about how important scoping can be (e.g. $\forall x.\exists y.\phi$ vs. $\exists x.\forall y.\phi$, and $\exists x.(\phi\land \phi')$ vs. $\exists x.\phi\land\exists x.\phi'$), but perhaps the point bears repeating.

\subsubsection{An aside on the backward and forward rules for $\tf{send}$ and $\tf{write}$}

\begin{rmrk}
\label{rmrk.why.evaporates}
In Remark~\ref{rmrk.recipe} we proposed a recipe for deriving a forward rule from a backward one: a backward rule of the form
$$
\rulefont{Rule?} 
\qquad
\tall v.\ \chi \timpc \psi[a\ssm v] \timpc \phi[a\ssm v]
$$
generates a forward rule (possibly wrapped in an $\final$ modality) of the form
$$
\rulefont{Rule!}
\qquad
(\texi v.\, \chi \tand \phi)\tnotor \texi v.\psi .
$$
We noted in Remark~\ref{rmrk.recipe} that the recipe works cleanly for some rules --- but it was less clear how it can be applied to derive \rulefont{PaxSend!} from \rulefont{PaxSend?}, and \rulefont{PaxWrite!} from \rulefont{PaxWrite?}.
In these cases, the left-hand sides of \rulefont{PaxSend!} and \rulefont{PaxWrite!} seem much simpler than the right-hand sides of \rulefont{PaxSend?} and \rulefont{PaxWrite?} respectively. 

We now show why, in fact, \rulefont{PaxSend!} and \rulefont{PaxWrite!} are correct and make sense in the context of our recipe, culminating with Corollary~\ref{corr.back.ok}.
\end{rmrk}

\begin{lemm}
\label{lemm.evaporates.simplify}
We have:
\begin{enumerate*}
\item\label{item.evaporates.simplify.1}
$\ment \texi v.\bigl((\mru{\tf{accept}}{v})\ \tor (\tneg\modT\recent\texi\tf{accept}\tand v\teq\udfn)\bigr) \ \tequiv\ \tvT$.
\item\label{item.evaporates.simplify.2}
$\ment \texi v.\bigl(\ (\mru{(\somewhere\tf{accept})}{v})\ \tor\ (\tneg\modT\recent\texi\somewhere\tf{accept}\tand \tf{propose}(v))\ \bigr)\ \tequiv\ (\modT\recent\texi\somewhere\tf{accept}) \tor \texi\tf{propose}$.
\end{enumerate*}
Note that:
\begin{itemize*}
\item
The notation $\texi\tf{accept}$ and $\texi\tf{propose}$ and $\texi\somewhere\tf{accept}$ is from Notation~\ref{nttn.hos}.
\item
We do not restrict to contexts and valuations that satisfy \ThyPax; the meaning of the $\ment$ symbol above is as in Definition~\ref{defn.validity.judgement}(\ref{item.ment.phi}).
\end{itemize*}
\end{lemm}
\begin{proof}
We consider each part in turn:
\begin{enumerate}
\item
We distribute the $\texi v$ through the $\tor$, and then through the $\tand$ on the right-hand side (we can do this because $v$ is not free in the left-hand conjunct) and we see that it suffices to show 
$$
\ment \texi v.(\mru{\tf{accept}}{v})\ \tor\ (\tneg\modT\recent\texi\tf{accept}\tand \texi v.v\teq\udfn)\ \tequiv\ \tvT
$$
Using Lemma~\ref{lemm.evaporates} this simplifies to just
$$
\ment \modT\recent\texi\tf{accept}\ \tor\ (\tneg\modT\recent\texi\tf{accept}\tand \texi v.v\teq\udfn) \ \tequiv\ \tvT.
$$
This clearly holds, just taking $v\teq\udfn$.
\item
We distribute $\texi v$ into the predicate as we did for the proof of part~\ref{item.evaporates.simplify.1}.
We see that it suffices to show 
$$
\ment \bigl(\ \texi v.(\mru{(\somewhere\tf{accept})}{v})\ \tor\ (\tneg\modT\recent\texi\somewhere\tf{accept}\tand \texi\tf{propose})\ \bigr)\ \tequiv\ \modT\recent\texi\somewhere\tf{accept}\ \tor\ \texi\tf{propose}. %
$$
We obtain this by simplifying the left-hand side using Lemma~\ref{lemm.evaporates}.
\qedhere
\end{enumerate}
\end{proof}

We can now justify the left-hand sides of \rulefont{PaxSend!} and \rulefont{PaxWrite!} as being consistent with our recipe as discussed above in Remark~\ref{rmrk.why.evaporates}:
\begin{corr}
\label{corr.back.ok}
Suppose $\avaluation$ is a valuation such that $\mentval\Theta_\ThyPax$.
Then we have:
\begin{enumerate*}
\item
$\begin{array}[t]{r@{\ }l}
\mentval\texi v.&\bigl(\someOsomeP\texi\tf{propose}\ \tand \bigl((\mru{\tf{accept}}{v})\ \tor (\tneg\modT\recent\texi\tf{accept}\tand v\teq\udfn)\bigr)\bigr) 
\\
&\tequiv \someOsomeP\texi\tf{propose}
\end{array}
$
\item
$\begin{array}[t]{r@{\ }l}
\mentval\texi v.&\bigl(\tf{leader}\tand 
\Quorum(\Box\texi\tf{send} \tand 
\bigl((\mru{(\somewhere\tf{accept})}{v}) \tor
(\tneg\modT\recent\texi\somewhere\tf{accept}\tand \tf{propose}(v))\bigr)) 
\bigr)
\\
&\tequiv
\tf{leader}\tand \Quorum\bigl(\everyone\texi\tf{send} \tand (\modT\recent\texi\somewhere\tf{accept}\tor\texi\tf{propose})\bigr)
\\
&\tequiv
\tf{leader}\tand \someOallP\texi\tf{send} \tand \texi\tf{propose}
\end{array}
$
\end{enumerate*}
\end{corr}
\begin{proof}
For both parts, the first equivalence follows by elementary predicate manipulations, distributing the $\texi v$ quantifier down into the predicate and using Lemma~\ref{lemm.evaporates.simplify}.
In the second part, the second equivalence follows by routine reasoning using the other rules of \ThyPax and in particular \rulefont{PaxSend?}.
By \rulefont{PaxSend?} and \rulefont{PaxPropose?} and the other rules, if someone $\tvT$-sent a value then some (possibly distinct) value must be $\tvT$-proposed by the unique leader, so we can drop the $\modT\recent\texi\somewhere\tf{accept}$ disjunct.
The rest of the argument is routine from truth-tables.
\end{proof}

\subsubsection{Axiom \rulefont{LdrUniq}: uniqueness of the leader}
\label{subect.ldrunique}

Recall axiom \rulefont{LdrUniq} from \ThyPax in Figure~\ref{fig.ThyPaxOne}:
$$
\rulefont{LdrUniq}
\quad
\somewhere\mu\tf X.\tneg\somewhere(\tf{leader}\tand\tf X) . 
$$
Note that $\tf X$ is not positive in $\tneg\somewhere(\tf{leader}\tand\tf X)$, because $\tf X$ occurs under a negation $\tneg$.
Thus, the $\mu$ is still taking an intersection, but that intersection does not necessarily determine a fixpoint. 
This is fine mathematically.
We just are now outside the bounds of fixpoint logic, and in a larger design space.

In Remark~\ref{rmrk.paxos.axioms.discussion}(\ref{item.discuss.LdrExistUniq}) we claimed that \rulefont{LdrUniq} expresses that there are zero or one leaders.
We now show how this works, culminating with Corollary~\ref{corr.ldruniq}.

\begin{defn}
\label{defn.neg.mu.notation}
Suppose 
$n\in\Time$ and $O\in\opensne$ and $\avaluation$ is a valuation and $\varkappa\in\THREE^{\tf{Ctx}}$.
We define:
\begin{enumerate*}
\item\label{item.neg.mu.notation.1}
$\f{leaders}_{n,\avaluation}=\{p\in\Pnt \mid \avaluation(\tf{leader})_{n,p}=\tvT\}$.
Call this the set of \deffont{$n,\avaluation$-leaders}, or if $n$ and $\avaluation$ are understood, just the set of \deffont{leaders}. 
\item\label{item.neg.mu.notation.2}
$z_{n,O,\avaluation,\varkappa}= 
\modellabel{\somewhere(\tf{leader}\tand\varkappa)}{n,p,O,\avaluation}=
\bigvee_{p\in O}(\avaluation(\tf{leader})_{n,p}\tand\varkappa(n,p,O,\avaluation))$.
\item\label{item.neg.mu.notation.3}
Say that $\varkappa\in\THREE^{\tf{Ctx}}$ is \deffont{admissible} when either of the two following equivalent conditions holds:
\begin{itemize*}
\item
$\tneg z_{n',O',\avaluation',\varkappa}\leq\varkappa(n',p',O',\avaluation')$ for every $n',p',O',\avaluation'$, or equivalently
\item
$\tneg\varkappa(n',p',O',\avaluation')\leq z_{n',O',\avaluation',\varkappa}$ for every $n',p',O',\avaluation'$.
\end{itemize*}
\item
Define 
$$
\begin{array}{r@{\ }c@{\ }l}
M(n,p,O,\avaluation)
&=&
\modellabel{\mu\tf X.\tneg\somewhere(\tf{leader}\tand\tf X)}{n,p,O,\avaluation}
\\
&\stackrel{\text{Fig.~\ref{fig.3.phi.f}}}=&
\bigl(\bigwedge\{ \varkappa\in\THREE^{\tf{Ctx}}\mid \model{\tneg\somewhere(\tf{leader}\tand\varkappa)}\leq\varkappa\}\bigr)(n,p,O,\avaluation)
\\
&=&
\bigwedge\{\varkappa\in\THREE^{\tf{Ctx}} \mid \varkappa\ \text{ admissible} \}\, (n,p,O,\avaluation) 
\\
&=&
\bigwedge\{\varkappa(n,p,O,\avaluation) \mid \varkappa\in\THREE^{\tf{Ctx}}\ \text{ admissible} \} .
\end{array}
$$ 
In words: $\modellabel{\mu\tf X.\tneg\somewhere(\tf{leader}\tand\tf X)}{}$ as a context valuation is equal to the meet (in $\THREE^{\tf{Ctx}}$ the set of all context valuations) of all the admissible $\varkappa$.
\end{enumerate*}
\end{defn}

\begin{rmrk}
The proofs below are quite detailed, so we first offer an intuitive sketch of why in some sense Proposition~\ref{prop.ldruniq.helper} and Corollary~\ref{corr.ldruniq} \emph{should} work.

Note below that a key construction is to take an intersection over all $\varkappa$ such that (in simplified form) $\model{\tneg\somewhere(\tf{leader}\tand\varkappa)}\leq\varkappa$.
This condition is monotone, by which we mean that if $\varkappa$ satisfies it and $\varkappa\leq\varkappa'$ then so does $\varkappa'$.
Thus the intersection is in some sense `created' by minimal $\varkappa$.
Now note that a $\varkappa$ that identifies one leader (by being non-false at that leader and false everywhere else), satisfies the condition.
Therefore, if there is just one leader to identify, then the intersection of all suitable $\varkappa$ is just the unique least identifying $\varkappa$; but if there are two leaders then the intersection of all suitable $\varkappa$ will include an intersection of two minimal identification functions identifying two distinct leaders, whose intersection identifies no elements.
\end{rmrk}

\begin{prop}
\label{prop.ldruniq.helper}
Suppose that
$(n,p,O,\avaluation)$ is a context, so that $n\in\Time$ is a stage, $O\in\opensne$ is a nonempty open set, and $\avaluation$ is a valuation.
Suppose further that $n\mentval\modTF\tf{leader}$.

Then using the notation and terminology from Definition~\ref{defn.neg.mu.notation}, we have that %
$M(n,p,O,\avaluation)$ takes the following values on $(n,p,O,\avaluation)$:
\begin{itemize*}
\item
If $O$ contains no leaders then $M(n,q,O,\avaluation)=\tvT$ for every $q\in\Pnt$.
\item
If $O$ contains precisely one leader $p\in\Pnt$ then $M(n,p,O,\avaluation)=\tvB$ and $M(n,q,O,\avaluation)=\tvF$ for every $q\in\Pnt\setminus\{p\}$.
\item
If $O$ contains two or more leaders then $M(n,q,O,\avaluation)=\tvF$ for every $q\in\Pnt$.
\end{itemize*} 
\end{prop}
\begin{proof}
Note that since we assumed $n\mentval\modTF\tf{leader}$, we know that $\tf{leader}$ is two-valued (returns $\tvT$ or $\tvF$, but not $\tvB$); so a point either is a leader or is not a leader. 

We are interested in determining the value of $M$, and this value is a meet (greatest lower bound) of admissible context valuations $\varkappa\in\THREE^{\tf{Ctx}}$.

Fix some arbitrary $\leq$-minimal admissible context valuation $\varkappa\in\THREE^{\tf{Ctx}}$ (if any such exists).
We now consider the possible values of $\lambda p\in\Pnt.\varkappa(n,p,O,\avaluation)$, and from this we will deduce what values $M$ must take:
\begin{enumerate}
\item
\emph{Suppose $\f{leaders}_{n,\avaluation}\cap O=\varnothing$.
In words: there are no leaders in $O$.}

By Definition~\ref{defn.neg.mu.notation}(\ref{item.neg.mu.notation.2})
$z_{n,O,\avaluation,\varkappa}=\tvF$ (regardless of $\varkappa$) and $\tneg z_{n,O,\avaluation,\varkappa}=\tvT$.
We assumed that $\varkappa$ is admissible so by Definition~\ref{defn.neg.mu.notation}(\ref{item.neg.mu.notation.3}) $\varkappa(n,q,O,\avaluation)=\tvT$ for every $q\in\Pnt$, and thus 
$$
M(n,q,O,\avaluation)=\tvT
$$
for every $q\in\Pnt$.
\item
\emph{Suppose $\f{leaders}_{n,\avaluation}\cap O=\{p\}$ for some $p\in\Pnt$.
In words: there is precisely one leader $p$ in $O$.}

By Definition~\ref{defn.neg.mu.notation}(\ref{item.neg.mu.notation.2})
$z_{n,O,\avaluation,\varkappa}=\varkappa(n,p,O,\avaluation)$ and $\tneg z_{n,O,\avaluation,\varkappa}=\tneg\varkappa(n,p,O,\avaluation)$.

We assumed that $\varkappa$ is admissible, 
so by Definition~\ref{defn.neg.mu.notation}(\ref{item.neg.mu.notation.3})
$\tneg\varkappa(n,p,O,\avaluation)\leq \varkappa(n,q,O,\avaluation)$ for every $q\in\Pnt$. 
In particular, $\tneg\varkappa(n,p,O,\avaluation)\leq\varkappa(n,p,O,\avaluation)$, so that we have two subcases: that $\varkappa(n,p,O,\avaluation)=\tvB$ or $\varkappa(n,p,O,\avaluation)=\tvT$.
We consider each in turn:
\begin{itemize*}
\item
\emph{Suppose $z_{n,O,\avaluation,\varkappa}=\varkappa(n,p,O,\avaluation)=\tvB$.}

Then $\tneg\tvB=\tvB\leq\varkappa(n,q,O,\avaluation)$ for all $q\in\Pnt$.
\item
\emph{Suppose $z_{n,O,\avaluation,\varkappa}=\varkappa(n,p,O,\avaluation)=\tvT$.}

Then $\tneg\tvT=\tvF\leq\varkappa(n,q,O,\avaluation)$ for all $q\in\Pnt$.
This condition is vacuous.
\end{itemize*}
It follows that for $q\in\Pnt$, 
$$
M(n,q,O,\avaluation)=
\begin{cases}
\tvB & q=p
\\
\tvF & q\neq p
\end{cases}
$$
\item
\emph{Suppose $p\neq p'\in \f{leaders}_{n,\avaluation}\cap O$ for two distinct $p\neq p'\in\Pnt$.
In words: there is more than one leader in $O$.}

By Definition~\ref{defn.neg.mu.notation}(\ref{item.neg.mu.notation.2})
$$
z_{n,O,\avaluation,\varkappa}=\bigvee_{p\in O\cap\f{leaders}_{n,\avaluation}}\varkappa(n,p,O,\avaluation).
$$
We just require $\varkappa$ to be admissible, and unpacking this requirement, we require that 
$$
\Forall{q\in\Pnt}\ \tneg\varkappa(n,q,O,\avaluation)\leq \bigvee_{p\in O\cap\f{leaders}_{n,\avaluation}}\varkappa(n,p,O,\avaluation) .
$$
We see that possible values for $\varkappa$ include the following:
\begin{itemize*}
\item
$\varkappa(n,p,O,\avaluation)=\tvT$ and $\varkappa(n,q,O,\avaluation)=\tvF$ for every $q\in\Pnt\setminus\{p\}$, or 
\item
$\varkappa(n,p',O,\avaluation)=\tvT$ and $\varkappa(n,q,O,\avaluation)=\tvF$ for every $q\in\Pnt\setminus\{p'\}$. 
\end{itemize*}
$M(n,q,O,\avaluation)$ is the meet of all admissible $\varkappa$, and it follows that $M(n,q,O,\avaluation)=\tvF$ for every $q\in\Pnt$.
\qedhere\end{enumerate}
\end{proof}

\begin{corr}
\label{corr.ldruniq}
Suppose $(n,p,O,\avaluation)$ is a context, so that $n\in\Time$ is a stage, $O\in\opensne$ is a nonempty open set, and $\avaluation$ is a valuation such that $n\mentval\modTF\tf{leader}$.
Write 
$$
L_n=\{p\in\Pnt \mid p\text{ is an $n,\avaluation$ leader}\}.
$$
Then: 
\begin{itemize*}
\item
If $\mycard (L_n\cap O)=0$ then $\modellabel{\somewhere\mu\tf X.\tneg\somewhere(\tf{leader}\tand\tf X)}{n,p,O,\avaluation}=\tvT$.
\item
If $\mycard (L_n\cap O)=1$ then $\modellabel{\somewhere\mu\tf X.\tneg\somewhere(\tf{leader}\tand\tf X)}{n,p,O,\avaluation}=\tvB$.
\item
If $\mycard (L_n\cap O)\geq 2$ then $\modellabel{\somewhere\mu\tf X.\tneg\somewhere(\tf{leader}\tand\tf X)}{n,p,O,\avaluation}=\tvF$.
\end{itemize*}
As corollaries:
\begin{enumerate*}
\item
$n,p,O\mentval\somewhere\mu\tf X.\tneg\somewhere(\tf{leader}\tand\tf X)$ precisely when $n,q,O\mentval\tf{leader}$ for at most one $q\in O$.
\item\label{item.ldruniq}
$n\mentval\somewhere\mu\tf X.\tneg\somewhere(\tf{leader}\tand\tf X)$ precisely when there exists at most one leader at stage $n$.
\item\label{item.one.leader}
If $\mentval\ThyPax$ then there exists precisely one leader at each stage $n$.
\end{enumerate*}
\end{corr}
\begin{proof}
We consider each part in turn:
\begin{enumerate*}
\item
By elementary reasoning from Proposition~\ref{prop.ldruniq.helper}, using the definition of $n,p,O\mentval\psi$ from Definition~\ref{defn.validity.judgement}(\ref{item.ctx.ment.phi}) (for arbitrary $\psi$) and the clause for $\somewhere\phi$ in Figure~\ref{fig.3.phi.f} (for arbitrary $\phi$).
\item
By the definition of $n\mentval\psi$ from Definition~\ref{defn.validity.judgement}(\ref{item.ment.n}).
\item
Combining the above with \rulefont{LdrExist}.
\qedhere\end{enumerate*}
\end{proof}

\begin{rmrk}
We are not aware of the precise technical construction from Definition~\ref{defn.neg.mu.notation} appearing in the literature --- but arguably it echoes an idea used in Kit Fine's PhD thesis~\cite{warwick72219} (see Chapter~3) which (paraphrasing slightly) implements a `there exist zero or one points satisfying $p$' modality as $\Forall{q}(\Box(p\to q) \vee \Box(p\to\tneg q))$ --- where here the impredicative $\forall q$ plays the role of our $\mu\tf X$ as applied to a non-monotone argument.

We originally introduced $\mu\tf X$ to model iteration, which is a programming construct.
But the comparison above suggests that $\mu\tf X$ might also be useful as a logical construct, namely as an impredicative universal quantification over all $\varkappa$.
Therefore, a tempting generalisation of \QLogicShort is simply to replace $\mu\tf X$ with an explicit and full-blooded $\forall\tf X$, such that 
$$
\modellabel{\tall\tf X.\phi}{n,p,O,\avaluation}= \bigwedge_{\varkappa\in\THREE^{\tf{Ctx}}} \modellabel{\phi[\tf X\ssm\varkappa]}{n,p,O,\avaluation} .
$$
There would be no mathematical barrier to doing this if it turns out to be useful, but the example of Paxos does not require such power so we do not include it for now. 
\end{rmrk}

\section{Correctness properties}
\label{sect.paxos.correctness.properties}

\subsection{Summary of the properties}

\begin{rmrk}
\label{rmrk.standard.paxos.correctness}
We will derive the following standard correctness properties for Paxos (as laid out in~\cite[Subsection~2.2.1]{cachin2011paxos}) from the theory \ThyPax in Figure~\ref{fig.logical.paxos}:
\begin{enumerate*}
\item
\textbf{Validity:} \quad If a participant decides some value $v$, then $v$ was proposed by some participant.

A formal statement and proof is in Theorem~\ref{thrm.validity}.
\item
\textbf{Agreement:}\quad 
If two participants decide values, then those values are equal.

A formal statement and proof is in Theorem~\ref{thrm.full.agreement}.
\item
\textbf{Termination:}\quad 
Every correct participant eventually decides some value.

A formal statement and proof is in Theorem~\ref{thrm.termination}.
\item
\textbf{Integrity:}\quad 
Every correct (non-crashing) participant decides at most once.

Our logic abstracts away from a direct notion of execution, by design, but this property is represented indirectly in how Theorem~\ref{thrm.termination} is expressed.
See Remark~\ref{rmrk.paxos.integrity} below. 
\end{enumerate*}
\end{rmrk}

\begin{rmrk}[On integrity]
\label{rmrk.paxos.integrity}
In implementations of Paxos, once a participant has decided, it stops.
But in our logic, `stop' is not even a concept.
There is no abstract machine to stop; there are just predicates, with denotations.

Nevertheless, the integrity property is echoed in our maths, as the \emph{infinite repetition} in Theorem~\ref{thrm.termination}.
Once a predicate $\tf{decide}(v)$ returns $\tvT$, it will do so again and again --- and because of Agreement, it can never decide any other value.
A predicate eventually repeatedly returning the same result, and never returning any other result, is a model of an abstract machine that has computed its result --- and then just sits there forever, displaying that same result, unchanged, every time we look.
\end{rmrk}

\subsection{Validity}

\begin{rmrk}
\label{rmrk.validity.in.english}
\emph{Validity} is expressed in English as:
\begin{quote}
\emph{If a participant decides some value $v$, then $v$ was proposed by some participant.}
\end{quote}
We render this in our logic in Theorem~\ref{thrm.validity} as
$$
\ThyPax\ment\tall v.(\tf{decide}(v) \timpc \urecent\someOsomeP(\tf{leader}\tand\tf{propose}(v))).
$$

In words we can read this as: 
\begin{quote}
If someone decides $v$ now, then at some stage now or in the past a leader of that stage proposed $v$. 
\end{quote}
We will now set about proving this property from our axioms, culminating with Theorem~\ref{thrm.validity}.
\end{rmrk}

\begin{lemm}
\label{lemm.validity.helper}
Suppose $\avaluation$ is a valuation such that $\mentval\ThyPax$ ($\mentval\ThyPax$ is from Figure~\ref{fig.validity}).
Suppose $n\in\Time$ and $v\in\Val$.
Then:
\begin{enumerate*}
\item\label{item.validity.helper.3}
$n\mentval \someOsomeP\tf{write}(v) \timpc \urecent\someOsomeP\tf{propose}(v)$. 
\item\label{item.validity.helper.2}
$n\mentval \someOsomeP\tf{accept}(v) \timpc \urecent\someOsomeP\tf{propose}(v)$. 
\end{enumerate*}
\end{lemm}
\begin{proof}
We prove both parts by a simultaneous induction on $n$.
We freely use \strongmodusponens:
\begin{enumerate}
\item
Suppose $n\mentval\modT\someOsomeP\tf{write}(v)$.
Then using \rulefont{PaxWrite?} %
$$
n\mentval \modT\Quorum(\mru{\somewhere\tf{accept}}{v}) 
\quad\text{or}\quad
n\mentval \modT\someOsomeP\tf{propose}(v) .
$$
We consider the two cases:
\begin{itemize*}
\item
\emph{If $n\mentval \modT\someOsomeP\tf{propose}(v)$,}\ then 
$n\mentval\modT\urecent\someOsomeP\tf{propose}(v)$ follows immediately from the clause for %
$\urecent$ in Figure~\ref{fig.3.derived}.
\item
\emph{If $n\mentval \modT\Quorum(\mru{\somewhere\tf{accept}}{v})$,}\ 
then by Lemma~\ref{lemm.mru.to.recent}(\ref{item.mru.to.recent.1}) $n\mentval\modT\recent\someOsomeP\tf{accept}(v)$,
so that by Figure~\ref{fig.3.derived} $n'\mentval\modT\someOsomeP\tf{accept}(v)$ for some $n'<n$.
By part~\ref{item.validity.helper.2} of this result for $n'<n$, we have that $n'\mentval\modT\urecent\someOsomeP\tf{propose}(v)$, and thus using Lemma~\ref{lemm.tomorrow.forever.recent}(\ref{item.rur.r}) $n\mentval\modT\urecent\someOsomeP\tf{propose}(v)$. 
\end{itemize*}
In either case, we are done.
\item
Suppose $n\mentval\modT\someOsomeP\tf{accept}(v)$.
By \rulefont{PaxAccept?} $n\mentval\modT\someOsomeP\tf{write}(v)$.
We use part~\ref{item.validity.helper.3} of the inductive hypothesis for $n$. 
\qedhere\end{enumerate}
\end{proof}

\begin{thrm}[Validity]
\label{thrm.validity}
We have:
$$
\ThyPax\ment\tall v.(\tf{decide}(v) \timpc \urecent\someOsomeP(\tf{leader}\tand\tf{propose}(v))) .
$$
\end{thrm}
\begin{proof}
Suppose $\avaluation$ is a valuation such that $\mentval\ThyPax$. 
We reason using \strongmodusponens. 
Suppose $v\in\Val$ and $n\in\Time$ and $p\in\Pnt$, and suppose $n,p\mentval\modT\tf{decide}(v)$.
We consider two cases:
\begin{itemize}
\item
\emph{Suppose $p=\f{leader}(n)$.}

By \rulefont{PaxDecideL?} $n\mentval\modT\someOallP\tf{accept}(v)$.
It follows by Proposition~\ref{prop.someone.implies.someoneAll}(\ref{item.someone.implies.someoneAll.2}) that $n\mentval\modT\someOsomeP\tf{accept}(v)$.
By Lemma~\ref{lemm.validity.helper}(\ref{item.validity.helper.2}) $n\mentval\modT\urecent\someOsomeP\tf{propose}(v)$, and using \rulefont{PaxPropose?} we conclude that $n\mentval\modT\urecent\someOsomeP(\tf{leader}\tand\tf{propose}(v))$ as required. 
\item
\emph{Suppose $p\neq\f{leader}(n)$.}

By \rulefont{PaxDecide\tneg L?} $n,p\mentval\someOsomeP(\tf{leader}\tand\tf{decide}(v))$.
Thus there exists $p'\in\Pnt$ such that $n,p'\mentval\tf{leader}\tand\tf{decide}(v)$ and we can proceed using the previous case. 
\qedhere\end{itemize}
\end{proof}

\subsection{Agreement}

\begin{rmrk}
\emph{Agreement} is expressed in English as:
\begin{quote}
\emph{If two participants decide values, then those values are equal.}
\end{quote}
We render this in our logic in Theorem~\ref{thrm.full.agreement} as
$$
\ThyPax\ment\tall v,v'.(\urecent\someOsomeP\tf{decide}(v) \timpc \someOsomeP\tf{decide}(v') \timpc v'\teq v) .
$$
In words we can read this as: 
\begin{quote}
If someone decides $v$ now or in the past, and someone decides $v'$ now, then $v'=v$.
\end{quote}
It is convenient to prove the `now' and the `in the past' parts of this as two sub-cases: 
\begin{enumerate*}
\item
Proposition~\ref{prop.decide.n.agree} handles the case of $\someOsomeP\tf{decide}(v)\tand\someOsomeP\tf{decide}(v')$, i.e. when the two decisions are made at the same time.
\item
Theorem~\ref{thrm.full.agreement} then extends this with the case of $\recent\someOsomeP\tf{decide}(v)\tand\someOsomeP\tf{decide}(v')$, i.e. when one decision precedes the other.
\end{enumerate*}
\end{rmrk}

\subsubsection{Agreement at time $n$}

\begin{lemm}
\label{lemm.various.unique}
Suppose $\avaluation$ is a valuation such that $\mentval\ThyPax$. 
Suppose $n\in\Time$ and $v,v'\in\Val$. 
Then:
\begin{enumerate*}
\item\label{item.various.unique.write}
$n\mentval (\someOsomeP\tf{write}(v)\tand \someOsomeP\tf{write}(v')) \timpc v\teq v'$. 
\item\label{item.various.unique.accept}
$n\mentval (\someOsomeP\tf{accept}(v)\tand \someOsomeP\tf{accept}(v')) \timpc v\teq v'$. 
\item\label{item.various.unique.recent.accept}
$n\mentval \bigl(\mru{(\someOsomeP\tf{accept})}{v}\tand \mru{(\someOsomeP\tf{accept})}{v'}\bigr) \timpc v\teq v'$.
\end{enumerate*}
\end{lemm}
\begin{proof}
We consider each part in turn, using \strongmodusponens (we may also silently commute the $\modT$ modality as convenient):
\begin{enumerate}
\item
Suppose $n\mentval\modT\someOsomeP\tf{write}(v)\tand\modT\someOsomeP\tf{write}(v')$.
Using \rulefont{PaxWrite?} 
$$
n\mentval\modT\someOsomeP(\tf{write}(v)\tand\tf{leader})\tand\modT\someOsomeP(\tf{write}(v')\tand\tf{leader}). 
$$
By Corollary~\ref{corr.ldruniq}(\ref{item.one.leader}) therefore 
$$
n\mentval\modT\someOsomeP(\tf{write}(v)\tand\tf{write}(v')\tand\tf{leader}).
$$
By \rulefont{PaxWrite01} and Lemma~\ref{lemm.unique.affine.existence}(\ref{item.unique.affine.existence.01.implies}) we have $v=v'$ as required.
\item
Suppose $n\mentval \modT\someOsomeP\tf{accept}(v)\tand\modT\someOsomeP\tf{accept}(v')$.
By \rulefont{PaxAccept?} $n\mentval\modT\someOsomeP\tf{write}(v)\tand\modT\someOsomeP\tf{write}(v')$.
We use part~\ref{item.various.unique.write} of this result.
\item
Suppose $n\mentval\modT(\mru{(\someOsomeP\tf{accept})}{v})\tand\modT(\mru{(\someOsomeP\tf{accept})}{v'})$.\footnote{By Lemma~\ref{lemm.mru.unambivalent}(2) we do not need to write the $\modT$ here but we do anyway; see Remark~\ref{rmrk.explain.mru}(\ref{item.explain.mru.T}).}
Using Lemma~\ref{lemm.mru.to.recent}(\ref{item.mru.to.recent.2}) (taking $a.\phi=\someOsomeP\tf{accept}$ in that Lemma) there exists $n'<n$ such that 
$n'\mentval\modT\someOsomeP\tf{accept}(v)\tand\modT\someOsomeP\tf{accept}(v')$.
We use part~\ref{item.various.unique.accept} of this result.
\qedhere\end{enumerate}
\end{proof}

\begin{prop}
\label{prop.decide.n.agree}
Suppose $\avaluation$ is a valuation such that $\mentval\ThyPax$. 
Suppose $n\in\Time$ and $v,v'\in\Val$.
Then
$$
n\mentval (\someOsomeP\tf{decide}(v) \tand \someOsomeP\tf{decide}(v')) \timpc v\teq v' . 
$$
\end{prop}
\begin{proof}
We use \strongmodusponens.
Suppose 
$n\mentval \modT(\someOsomeP\tf{decide}(v)\tand\someOsomeP\tf{decide}(v'))$.
Using \rulefont{PaxDecide\tneg L?} if needed (as we did in the proof of Theorem~\ref{thrm.validity}), and then \rulefont{PaxDecideL?} twice, we have
$n\mentval \modT(\someOallP\tf{accept}(v)\tand\someOallP\tf{accept}(v'))$ and thus by Proposition~\ref{prop.someone.implies.someoneAll}(\ref{item.someone.implies.someoneAll.2}) and Lemma~\ref{lemm.commuting.connectives}
also $n\mentval \modT(\someOsomeP\tf{accept}(v)\tand\someOsomeP\tf{accept}(v'))$.
We can now use Lemma~\ref{lemm.various.unique}(\ref{item.various.unique.accept}).
\end{proof}

\begin{rmrk}
Proposition~\ref{prop.decide.n.agree} proves part of the agreement property; that no two correct participants decide differently \emph{at a given time}.
To prove the result across different times will require a bit more work, culminating with Theorem~\ref{thrm.full.agreement}.
\end{rmrk}

\subsubsection{Some fine structure of the rules}
\label{subsect.ldrexistuniq'}

We continue a discussion from Remark~\ref{rmrk.paxos.axioms.discussion}(\ref{item.discuss.LdrExistUniq}) and also Corollary~\ref{corr.ldruniq} from Section~\ref{subect.ldrunique}.

Note that we use \rulefont{LdrUniq} from Figure~\ref{fig.logical.paxos} just once: in the proof of Lemma~\ref{lemm.various.unique}(\ref{item.various.unique.write}), via Corollary~\ref{corr.ldruniq}(\ref{item.one.leader}).

Examining how we use it, we see that for our specific purposes it would suffice to assume a more parsimonious axiom (both in the sense that it is simpler, and in the sense that it permits more models) as follows:
$$
\hspace{-6em}\rulefont{LdrExtWrite}
\qquad\qquad
\tall v.(\tf{leader}\timpc \someOsomeP(\tf{leader}\tand\tf{write}(v))\timpc \tf{write}(v)) .
$$ 
Recall from Subsection~\ref{subect.ldrunique} that axiom \rulefont{LdrUniq} asserts that there is a unique leader at each stage (actually, at most one; \rulefont{LdrExist} does the rest).
In contrast, \rulefont{LdrExtWrite} above asserts more weakly that there is a unique leader at each stage, up to extensional equivalence \emph{for just writing values}.
In the proofs, this turns out to be all we need.

In fact, we can simplify this just to 
$$
\hspace{-6em}\rulefont{SPaxWrite01}
\qquad\qquad
\texiaffine\someOsomeP\tf{write} .
$$
In full, this is $\texiaffine v.\someOsomeP\tf{write}(v)$, and in words it expresses: 
\begin{quote}
`at each stage, at most one value is written'.
\end{quote}
In the presence of \rulefont{PaxWrite?} and \rulefont{PaxWrite01},\ \rulefont{SPaxWrite01} is implied by \rulefont{LdrExtWrite} --- but it does not imply \rulefont{LdrExtWrite}, because it permits the possibility that one leader participant writes a value and another leader participant crashes.
The reader can check that \rulefont{SPaxWrite01} is sufficient to complete the proof of Lemma~\ref{lemm.various.unique}(\ref{item.various.unique.write}).

Note also that \rulefont{SPaxWrite01} implies \rulefont{PaxWrite01}, so that --- for the proofs we care about --- \rulefont{SPaxWrite01} can replace both \rulefont{LdrUniq} and \rulefont{PaxWrite01}. 
We will do exactly this, later: the job of $\Theta_\ThyPax$ is to be `declarative Paxos', that is, its purpose is to be declarative \emph{and} to be recognisably Paxos. 
Later on in Figure~\ref{fig.ThySPax} we will define a \emph{simpler} declarative Paxos (whence the name `SPaxWrite01').
This simplified theory is derived by looking at the proofs here, and seeing what we actually require.

The journey from \rulefont{LdrExist}+\rulefont{LdrUniq} to \rulefont{SPaxWrite01} is the usual mathematical journey from something relatively concrete, to an abstraction of that thing.
Of course \rulefont{LdrExist}+\rulefont{LdrUniq} themselves are axiomatic abstractions of the programmatic fact that in a Paxos implementation, there \emph{is} precisely one leader.
How much detail we choose to retain, and how much we choose to dispense with, is a matter of taste, judgement, and the intended use-case.
Logic gives us freedom to make these choices as best suits our needs.

\subsubsection{Agreement at possibly different times}

\begin{prop}
\label{prop.integrity.helper}
Suppose $\avaluation$ is a valuation such that $\mentval\ThyPax$, and suppose $v,v'\in\Val$. 
Then:
\begin{enumerate*}
\item\label{item.integrity.helper.accept}
$\mentval\recent\someOallP\tf{accept}(v) \timpc \someOsomeP\tf{accept}(v')\timpc v'\teq v$.
\item\label{item.integrity.helper.write}
$\mentval\recent\someOallP\tf{accept}(v) \timpc \someOsomeP\tf{write}(v')\timpc v'\teq v$.
\end{enumerate*}
\end{prop}
\begin{proof}
Suppose $n\in\Time$.
We will reason using \strongmodusponens, so suppose 
\begin{equation}
\label{eq.n.1}
n\mentval\modT\someOallP\tf{accept}(v). 
\end{equation}
We will prove by induction on $n'>n$ that for every $v'\in\Val$,
$$
\begin{array}{r@{\ }l}
n'\mentval&\modT(\mru{(\someOsomeP\tf{accept})}{v'}\timpc v'\teq v)
\quad\text{and}\quad
\\
n'\mentval&\modT(\someOsomeP\tf{write}(v')\timpc v'\teq v)
\end{array}
$$
Clearly, by \strongmodusponens the $\tf{write}$ part of the inductive hypothesis is enough to prove part~\ref{item.integrity.helper.write} of this result.
For the $\tf{accept}$ part of the inductive hypothesis, we see that this implies part~\ref{item.integrity.helper.accept} of this result either by a routine argument using \rulefont{PaxAccept?} and part~\ref{item.integrity.helper.write} of this result --- or using Corollary~\ref{corr.mru.today.to.yesterday}(\ref{item.mru.today.implies.tomorrow}) (for $a.\phi=a.\someOsomeP\tf{accept}(a)$ in that result).

We consider cases:
\begin{itemize}
\item
\emph{The base case for $\tf{accept}$, where $n'=n+1$.}

Suppose $n'\mentval\modT(\mru{\someOsomeP\tf{accept}}{v'})$.

We assumed in equation~\eqref{eq.n.1} that $n\mentval\modT\someOallP\tf{accept}(v)$, so by Proposition~\ref{prop.someone.implies.someoneAll}(\ref{item.someone.implies.someoneAll.2}) also $n\mentval\modT\someOsomeP\tf{accept}(v)$, and by Corollary~\ref{corr.mru.today.to.yesterday}(\ref{item.mru.today.implies.tomorrow}) $n'\mentval\modT(\mru{(\someOsomeP\tf{accept})}{v})$. 
By Lemma~\ref{lemm.various.unique}(\ref{item.various.unique.recent.accept}) $v'=v$.

\item
\emph{The inductive step for $\tf{accept}$, where $n'>n+1$.}

Suppose $n'\mentval\modT(\mru{\someOsomeP\tf{accept}}{v'})$.
By Corollary~\ref{corr.mru.today.to.yesterday}(\ref{item.mru.tomorrow.options}) either $n'\minus 1\mentval\modT\someOsomeP\tf{accept}(v')$ or $n'\minus 1\mentval\modT(\mru{\someOsomeP\tf{accept}}{v'})$.

In the former case, using \rulefont{PaxAccept?} $n'\minus 1\mentval\someOsomeP\tf{write}(v')$, and by the $\tf{write}$ part of the inductive hypothesis, $v'=v$.
In the latter case, we just use the $\tf{accept}$ part of the inductive hypothesis. 
\item
\emph{The case of $\tf{write}$.}

Suppose $p\in\Pnt$ and $n',p\mentval\modT\tf{write}(v')$.
Using \rulefont{PaxWrite?}, there exists an $O\in\opensne$ such that 
$$
n',p,O\mentval \modT\bigl( 
(\mru{\somewhere\tf{accept}}{v'})\ \tor
(\tneg\modT\recent\texi\somewhere\tf{accept}\tand \tf{propose}(v')) \bigr),
$$
from which it follows using Lemmas~\ref{lemm.commuting.connectives} and~\ref{lemm.tand.tor} and some basic reasoning on truth-values that 
$$
\text{either}\quad
n',O\mentval \modT\mru{\somewhere\tf{accept}}{v'}
\quad\text{or}\quad
n',O\mentval \tneg\modT\recent\texi\somewhere\tf{accept} .
$$
Now $n<n'$ and we assumed in equation~\eqref{eq.n.1} that $n\mentval\modT\someOallP\tf{accept}(v)$, so using 
\rulefont{Pax2Twined} and \strongmodusponens $n\mentval\modT\allOsomeP\tf{accept}(v)$, and in particular (by routine reasoning on the clauses for $\allOsomeP$ in Figure~\ref{fig.3.derived} and for $\somewhere$ in Figure~\ref{fig.3.phi.f})
$n,O\mentval\modT\somewhere\tf{accept}(v)$.
Thus $n,O\mentval\modT\texi\somewhere\tf{accept}$ and by Lemma~\ref{lemm.tomorrow.forever.recent}(\ref{item.tomorrow.forever.recent.1}) (taking $\phi=\texi\somewhere\tf{accept}$ in that Lemma) 
$n',O\mentval\modT\recent\texi\somewhere\tf{accept}$.
It follows that we are not in the right-hand (`or') branch above, so we must be in the left-hand (`either') branch, and thus $n',O\mentval\modT\mru{\somewhere\tf{accept}}{v'}$.

Using the $\tf{accept}$ part of the inductive hypothesis, $v'=v$.
\qedhere\end{itemize}

\end{proof}

\begin{thrm}[Agreement]
\label{thrm.full.agreement}
We have:
$$
\ThyPax\ment\tall v,v'.(\urecent\someOsomeP\tf{decide}(v) \timpc \someOsomeP\tf{decide}(v') \timpc v'\teq v).
$$
\end{thrm}
\begin{proof}
Suppose $v,v'\in\Val$ and $n\in\Time$ and $\avaluation$ is a valuation such that $\mentval\ThyPax$.
We use \strongmodusponens, so suppose
$$
n\mentval\modT\urecent\someOsomeP\tf{decide}(v)
\quad\text{and}\quad
n\mentval\modT\someOsomeP\tf{decide}(v') .
$$
If $n\mentval\modT\someOsomeP\tf{decide}(v)$ then we use Proposition~\ref{prop.decide.n.agree}.

So suppose $n\mentval\modT\recent\someOsomeP\tf{decide}(v)$.
By \rulefont{PaxDecideL?} (and \rulefont{PaxDecide\tneg L?} if needed), $n\mentval\modT\recent\someOallP\tf{accept}(v)$ and $n\mentval\modT\someOallP\tf{accept}(v')$, so that by Proposition~\ref{prop.someone.implies.someoneAll}(\ref{item.someone.implies.someoneAll.2}) also %
$n\mentval\modT\someOsomeP\tf{accept}(v')$.
By Proposition~\ref{prop.integrity.helper}(\ref{item.integrity.helper.accept}) $v'=v$ as required.
\end{proof}

\subsection{Termination}
\label{subsect.pax.termination}

\begin{rmrk}
\emph{Termination} is expressed in English as:
\begin{quote}
\emph{Every correct participant eventually decides some value.}
\end{quote}
We render this in our logic in Theorem~\ref{thrm.termination} as 
$$
\ThyPax\ment\infinitely\allOallP\texi v.\tf{decide}(v).
$$
In words: 
\begin{quote}
There is always some future time when all participants decide.
\end{quote}
Note that there is no need to explicitly specify `all correct/uncrashed participants' above.
This is taken care of by our logic's in-built notion of validity, since incorrect participants return $\tvB$, which is valid as per Remark~\ref{rmrk.motivate.validity}(\ref{item.tb.valid}) (but not true or correct; cf. Remark~\ref{rmrk.motivate.validity}).
\end{rmrk}

\subsubsection{Global stabilisation of logical time (GSLT)}
\label{subsect.gslt}

Our forward rules are wrapped in an $\final$ modality, to reflect an assumption that while there may be some initial network instability, if we wait long enough the network will stabilise (see the discussion in Remark~\ref{rmrk.discuss.forward.rules}).
This leads us to the notion of \emph{global stabilisation of logical time}, or \emph{GSLT} for short: 
\begin{defn}[GSLT]
\label{defn.gslt}
Suppose $\phi$ is a closed predicate and $\avaluation$ is a valuation.
Then:
\begin{enumerate*}
\item\label{item.gslt.synchronous}
Say that \ $\final\phi$\ (i.e. $\phi$ under an $\final$ modality) is \deffont{synchronous at $n$ in $\avaluation$} when $n\mentval\allOallP\forever\phi$.\footnote{%
By Definition~\ref{defn.validity.judgement}, $n\mentval\allOallP\forever\phi$ means that $n',p',O'\mentval\phi$ for every $n'> n$ and every $O'\in\opens$ and every $p'\in O'$.
\emph{Subtlety alert:} This is not quite the same thing as just $n'\mentval\phi$ for every $n'> n$.  The reader might like to think about why.

(Hint: $p\in O$.)} 
That is:
$$
\final\phi\ \ \text{ is synchronous at $n$ in $\avaluation$}
\quad\text{when}\quad
n\mentval\allOallP\forever\phi .
$$
In words, $\final\phi$ is synchronous when the `$\final$' becomes `at every $O$, at every $p\in O$, at every future time'.
\item\label{item.gslt.gslt}
Suppose the valuation $\avaluation$ is such that $\mentval\ThyPax$ ($\mentval$ is from Figure~\ref{fig.validity} and $\ThyPax$ is from Figure~\ref{fig.ThyPaxOne}).

Then define $\gslt(\avaluation)\in\Time$ to be the least stage $n$ such that all of the rules in Figure~\ref{fig.ThyPaxOne} that start with $\final$ --- namely, the forward rules \rulefont{PaxSend!}, \rulefont{PaxWrite!}, \rulefont{PaxAccept!}, \rulefont{PaxDecideL!}, and \rulefont{PaxDecide\tneg L!} --- are synchronous at $n$ in the valuation $\avaluation$.

We call $\gslt(\avaluation)$ the \deffont{Global Stabilisation of Logical Time} in $\avaluation$. 
\end{enumerate*}
\end{defn}

\begin{rmrk}
We make some observations about Definition~\ref{defn.gslt}.
Suppose $\phi$ is a closed predicate.
Then:
\begin{enumerate}
\item
Unpacking Definition~\ref{defn.gslt}(\ref{item.gslt.synchronous}) using the clause for $\forever$ in Figure~\ref{fig.3.derived}, 
$$
\final\phi \text{ is synchronous at $n$ in $\avaluation$}
\quad\liff\quad
n\mentval\forever\allOallP\phi 
$$
\noindent
and $\gslt(\avaluation)$ is least such that $\gslt(\avaluation)\mentval\forever(\someOsomeP\texi\tf{propose} \tnotor \texi\tf{send})$, and $\gslt(\avaluation)\mentval\forever\bigl((\tf{leader}\tand (\someOallP\texi\tf{send})\tand\texi\tf{propose}) \tnotor \texi\tf{write}\bigr)$, and so on.
\item
If $\phi$ is absolutely pointwise --- meaning by Definition~\ref{defn.pointwise.meaning} that $\modellabel{\phi}{n,p,O,\avaluation}=\modellabel{\phi}{n,p,O',\avaluation}$ always, so that the truth-value of $\phi$ does not depend on the open set argument $O$ --- then we can simplify further:
$$
\begin{array}{r@{\ }l}
\final\phi \text{ is synchronous at $n$ in $\avaluation$}
\quad\liff\quad&
n\mentval\forever\allOallP\phi 
\\
\quad\liff\quad&
\Forall{n'{>}n}(n'\mentval\phi).
\end{array}
$$
The reader can easily check that the forward rules \rulefont{PaxSend!}, \rulefont{PaxWrite!}, \rulefont{PaxAccept!}, \rulefont{PaxDecideL!}, and \rulefont{PaxDecide\tneg L!} from \ThyPax in Figure~\ref{fig.ThyPaxOne} are all of the form $\final\phi$ for absolutely pointwise $\phi$.\item  
$\gslt$ is well-defined in Definition~\ref{defn.gslt}(\ref{item.gslt.gslt}) because there are finitely many $\final$-axioms in \ThyPax, mentioning finitely many $\final$ modalities, and a finite set of natural numbers has a finite least upper bound $n'$.
It is clear by construction that if $\final\phi$ is synchronous at $n$ then it is synchronous at every $n''>n$, so in particular all the $\final$-axioms are synchronous at $n''=n'$.\footnote{%
It is important that $\phi$ be closed for this argument to work.
For instance, if we had written \rulefont{PaxAccept!} as $\tall v.\final(\someOsomeP\tf{write}(v) \tnotor \tf{accept}(v))$ --- which is a distinct assertion from that used in Figure~\ref{fig.logical.paxos} --- then $\gslt$ might not be defined.

Also the fine design of the $\final$ modality, matters.
For instance, if we had defined $\final$ to be $\allOallP\sometime\,\forever$ rather than $\sometime\allOallP\forever$, then the arguments might fail (consider the case of $\Pnt$ infinite). 
}
\end{enumerate}
\end{rmrk}

\begin{lemm}
\label{lemm.someone.writes.after.gslt.accept}
Suppose $\avaluation$ is a valuation such that $\mentval\ThyPax$. 
Suppose $n\in\Time$ and $n>\gslt(\avaluation)$ and $v\in\Val$. 
Then 
$$
n\mentval\modT\someOsomeP\tf{write}(v) 
\quad\liff\quad
n\mentval\modT\someOallP\tf{accept}(v) 
\quad\liff\quad
n\mentval\modT\someOsomeP\tf{accept}(v)  .
$$
\end{lemm}
\begin{proof}
We prove a cycle of implications, freely using weak and strong modus ponens (Proposition~\ref{prop.mp.for.tnotor}(\ref{item.mp.for.tnotor}\&\ref{item.mp.for.timpc})).
\begin{itemize}
\item
\emph{Suppose $n\mentval\modT\someOsomeP\tf{write}(v)$.}

By \rulefont{PaxAccept!} (since $n>\gslt(\avaluation)$), 
$\Forall{p\in\Pnt,O\in\opensne}n,p,O\mentval\texi\tf{accept}$. 
Restricting to $p\in O$ in the quantification, we have that $n\mentval\allOallP\texi\tf{accept}$.
By Proposition~\ref{prop.correct.cobox.boxT}(\ref{item.correct.cobox.boxT.2}) and \rulefont{PaxPCorrect} (for $\tf{accept}$),
$n\mentval\modT\someOallP\texi\tf{accept}$.
Unpacking the existential $\texi\tf{accept}$ and using \rulefont{PaxAccept?} and Lemma~\ref{lemm.various.unique}(\ref{item.various.unique.write}), we conclude that $n\mentval\modT\someOallP\tf{accept}(v)$. 
\item
\emph{Suppose $n\mentval\modT\someOallP\tf{accept}(v)$.}

By Proposition~\ref{prop.someone.implies.someoneAll}(\ref{item.someone.implies.someoneAll.2}) $n\mentval\modT\someOsomeP\tf{accept}(v)$.
\item
\emph{Suppose $n\mentval\modT\someOsomeP\tf{accept}(v)$.}

By \rulefont{PaxAccept?} $n\mentval\modT\someOsomeP\tf{write}(v)$.  
\qedhere\end{itemize}
\end{proof}

\begin{lemm}
\label{lemm.leaderlive.quorumsend}
Suppose $\avaluation$ is a valuation such that $\mentval\ThyPax$. 
Suppose $n\in\Time$ and $n>\gslt(\avaluation)$. 
Then:
\begin{enumerate*}
\item\label{item.leaderlive.quorumsend.1}
$n\mentval\someOsomeP\texi\tf{propose}\tnotor\allOallP\texi\tf{send}$.
\item\label{item.leaderlive.quorumsend.2}
$n\mentval\someOsomeP\texi\tf{propose}\timpc\someOallP\texi\tf{send}$.
\end{enumerate*}
\end{lemm}
\begin{proof}
We consider each part in turn:
\begin{enumerate*}
\item
We use \weakmodusponens.
Suppose $n\mentval\modT\someOsomeP\texi\tf{propose}$.
By \rulefont{PaxSend!} (since $n>\gslt(\avaluation)$) 
$$
\Forall{p\in\Pnt}n,p\mentval\texi\tf{send}.
$$
By Lemma~\ref{lemm.atomic.O}(\ref{item.atomic.pointwise.2}) $\texi\tf{send}$ is pointwise, so by Proposition~\ref{prop.pointwise.dense.char}(\ref{item.dense.char.pointwise.everyoneAll}) $n\mentval\allOallP\texi\tf{send}$ as required.
\item
We use \strongmodusponens.
Suppose $n\mentval\modT\someOsomeP\texi\tf{propose}$.
By part~\ref{item.leaderlive.quorumsend.1} of this result and \weakmodusponens, $n\mentval\allOallP\texi\tf{send}$.
By \rulefont{PaxPCorrect} (for $\tf{send}$) $n\mentval\someOallP\correct{\tf{send}}$, so by 
Proposition~\ref{prop.correct.cobox.boxT}(\ref{item.correct.cobox.boxT.2})
$n\mentval\modT\someOallP\texi\tf{send}$ as required. 
\qedhere\end{enumerate*}
\end{proof}

\subsubsection{GSLT implies termination}

Recall from Definition~\ref{defn.correct}(\ref{item.correct}) that $\correct{\tf{decide}}=\tall v.\modTF\tf{decide}(v)$. 
\begin{lemm}
\label{lemm.someone.propose.someone.write}
Suppose $\avaluation$ is a valuation such that $\mentval\ThyPax$. 
Suppose $n\in\Time$ and $n>\gslt(\avaluation)$, and $p\in\Pnt$.
Then:
\begin{enumerate*}
\item\label{item.someone.propose.someone.write}
If $n,p\mentval\modT\texi\tf{propose}$ and $n,p\mentval\correct{\tf{write}}$ then $n,p\mentval\modT(\tf{leader}\tand\texi\tf{write})$.
\item\label{item.someone.propose.someone.write.live}
$\mentval\infinitely\someOsomeP(\tf{leader}\tand\correct{\tf{decide}}\tand\modT\texi\tf{write})$.  
\item\label{item.someone.propose.someone.accept}
$\mentval\infinitely(\someOsomeP(\tf{leader}\tand\correct{\tf{decide}})\tand\modT\texi\someOallP\tf{accept})$.
\end{enumerate*}
\end{lemm}
\begin{proof}
We consider each part in turn:
\begin{enumerate}
\item
Suppose $n,p\mentval\modT\texi\tf{propose}$.
By \rulefont{PaxPropose?} and \strongmodusponens $n,p\mentval\modT\tf{leader}$.
By Lemma~\ref{lemm.leaderlive.quorumsend}(\ref{item.leaderlive.quorumsend.2}) and strong modus ponens $n\mentval\modT\someOallP\texi\tf{send}$.
Using \weakmodusponens and \rulefont{PaxWrite!} (since $n>\gslt(\avaluation)$) $n,p\mentval\texi\tf{write}$.
By Proposition~\ref{prop.tv.ment.TF.model}(\ref{item.modTF.texi.P.valid}) (since $n,p\mentval\correct{\tf{write}}$) $n,p\mentval\modT\texi\tf{write}$. 
\item
By \rulefont{LdrExist}, and \rulefont{LdrCorrect} with Lemma~\ref{lemm.persist.char} (for $\f{tvs}=\{\tvT\}$ in that Lemma), and \strongmodusponens, 
for infinitely many $n>\gslt(\avaluation)$ there exists $l_n\in\Pnt$ such that 
$$
n,l_n\mentval\modT\tf{leader}\tand\correct{\tf{propose},\tf{write},\tf{decide}}. 
$$
(We can equivalently write $\modT\tf{leader}$ or just $\tf{leader}$ above, because \rulefont{LdrExist} insists that $\tf{leader}$ is two-valued.) 
So for each $n$ and $l_n$ we have $n,l_n\mentval\correct{\tf{decide}}$.
We will now check that $n,l_n\mentval\modT\texi\tf{write}$.
We reason as follows:
\begin{itemize*}
\item
By \rulefont{PaxPropose!} (since $n>\gslt(\avaluation)$) and \weakmodusponens we have that $n,l_n\mentval\texi\tf{propose}$.
\item
By Proposition~\ref{prop.tv.ment.TF.model}(\ref{item.modTF.texi.P.valid}) (since $n,l_n\mentval\correct{\tf{propose}}$ and $n,l_n\mentval\texi\tf{propose}$) $n,l_n\mentval\modT\texi\tf{propose}$. 
\item
By part~\ref{item.someone.propose.someone.write} of this result (since $n,l_n\mentval\correct{\tf{write}}$) $n,l_n\mentval\modT\texi\tf{write}$.
\end{itemize*}
\item
We just combine part~\ref{item.someone.propose.someone.write.live} of this result with Lemma~\ref{lemm.someone.writes.after.gslt.accept}.
\qedhere\end{enumerate}
\end{proof}

Recall that Theorem~\ref{thrm.termination} expresses termination (as the name suggests) but also, taken together with Agreement, it expresses an abstracted form of the Integrity correctness property by virtue of the use of the $\infinitely$ modality, as per Remark~\ref{rmrk.paxos.integrity}:
\begin{thrm}[Termination]
\label{thrm.termination}
We have:
$$
\ThyPax\ment\infinitely\allOallP\texi\tf{decide} .
$$ 
\end{thrm}
\begin{proof}
Suppose $\avaluation$ is a valuation and $\mentval\ThyPax$.
By Lemmas~\ref{lemm.someone.propose.someone.write}(\ref{item.someone.propose.someone.accept}) and~\ref{lemm.persist.char} (for $\f{tvs}=\{\tvT\}$ in that Lemma), for infinitely many $n\in\Time$ there exists an $l_n\in\Pnt$ such that 
\begin{enumerate*}
\item
$n,l_n\mentval\tf{leader}$,
\item
$n,l_n\mentval\correct{\tf{decide}}$ (Definition~\ref{defn.correct}(\ref{item.correct})), and
\item
$n\mentval\modT\texi\someOallP\tf{accept}$.
\end{enumerate*}
By \rulefont{PaxDecideL!} (since $n>\gslt(\avaluation)$) 
and \weakmodusponens $n,l_n\mentval\texi\tf{decide}$.
By Proposition~\ref{prop.tv.ment.TF.model}(\ref{item.modTF.texi.P.valid}) (since $n,l_n\mentval\correct{\tf{decide}}$ and $n,l_n\mentval\texi\tf{decide}$) $n,l_n\mentval\modT\texi\tf{decide}$. 

It follows using \rulefont{PaxDecide\tneg L!} and \weakmodusponens that $n,p\mentval\texi\tf{decide}$ for non-leader $p\in\Pnt$ and thus for \emph{every} $p\in\Pnt$, and so by (Lemma~\ref{lemm.atomic.O}(\ref{item.atomic.pointwise.2}) and) Proposition~\ref{prop.pointwise.dense.char}(\ref{item.dense.char.pointwise.everyoneAll}) that $n\mentval\allOallP\texi\tf{decide}$. 

The above holds for infinitely many $n\in\Time$, so by Lemma~\ref{lemm.persist.char} (for $\f{tvs}=\{\tvT,\tvB\}$ in that Lemma) we have that $\mentval\infinitely\allOallP\texi\tf{decide}$ as required. 
\end{proof}

\begin{rmrk}
Theorem~\ref{thrm.termination} proves that $\ThyPax\ment\infinitely\allOallP\texi\tf{decide}$ holds.
Note that $\tvB$ (the truth-value associated to a crashed process) is a valid truth-value, so $n\ment\allOallP\texi\tf{decide}$ can hold in the trivial case that everyone is crashed.
Thus, what Theorem~\ref{thrm.termination} asserts in English is that infinitely often, everyone who is not crashed, decides (but it might be that everyone is crashed) --- from which it obviously follows that if a participant never crashes then eventually it decides a value, as per Remark~\ref{rmrk.standard.paxos.correctness}.

But in fact, we can prove a bit more:
$$
\ThyPax\ment\infinitely\someOallP\modT\texi\tf{decide}
$$
In words: infinitely often, there exists a \emph{quorum of uncrashed participants}, who decide.

This follows from Theorem~\ref{thrm.termination} by easy reasoning using \rulefont{PaxPCorrect}, \rulefont{PaxDecide\tneg L!}, and Proposition~\ref{prop.decide.n.agree}; we leave this to the reader to check. 
\end{rmrk}

\section{Simpler Declarative Paxos}
\label{sect.simpler.paxos}

Logic is useful for simplification and abstraction in many ways.
In particular, we can axiomatise something, prove some properties --- and then look at what we actually needed from the axioms to make the proofs work.
This can feed back to suggest refinements and generalisations of the original axioms.
(For example: many properties of natural numbers are actually properties of fields, or rings, or integral domains, and so on.) 

We can do something similar here, and look for the most elementary set of axioms that are actually required for our correctness proofs as written: 

\subsection{What we used in the correctness proofs for \ThyPax}
\label{subsect.reflection.on.correctness}
 
\begin{rmrk}
\label{rmrk.spax.explanation}
We note a curious thing about our correctness proofs in Section~\ref{sect.paxos.correctness.properties}, that we never used axioms \rulefont{PaxSend?} or \rulefont{PaxPropose01}.

How can this be?
Is this an error?

No.
This is an opportunity to further simplify the Declarative Paxos from Figure~\ref{fig.ThyPaxOne} to the \emph{Simpler Declarative Paxos} in Figure~\ref{fig.ThySPax}.
To understand how and why this works, we need to highlight some details about how the Paxos algorithm works.
\end{rmrk}

\begin{rmrk}[Why we did not use \rulefont{PaxSend?}]
Recall from Remark~\ref{rmrk.high-level.paxos} our high-level view of the Paxos algorithm, and recall that \ThyPax in Figure~\ref{fig.ThyPaxOne} is a declarative version of that algorithm (obtained by abstracting away explicit message-passing and algorithmic time, as per Remark~\ref{rmrk.elisions}).

Note of the algorithm that 
\begin{itemize*}
\item
in step~\ref{item.l.proposes} of Remark~\ref{rmrk.high-level.paxos}, the leader broadcasts a `propose' message; and 
\item
in step~\ref{item.p.sends}, a participant $p$ that receives that `propose' message, replies with a `send' message carrying as its payload a value $v_p$, which is $p$'s most recently accepted value (if any); and
\item
in step~\ref{item.l.writes}, the leader gathers the $v_p$ values from these responses and, if it has a quorum (nonempty open set) of responses, it selects a most recent value to use in its write step.
\end{itemize*}
In the algorithm, the only way that the leader $l$ can know of a participant $p$'s value of $v_p$, is if $l$ receives a send message from $p$ with payload $v_p$.
If $l$ receives no such message, then it does not know the value of $v_p$.

In contrast, predicates of \QLogicShort are not assigned truth-values by participants in an algorithm; they are assigned truth-values as if by an omniscient observer who is standing outside of the model and who has full (omniscient) access to it.
This is reflected in \rulefont{PaxWrite?}, where we write $\mru{(\somewhere\tf{accept})}{v}$.
This directly examines the model to find the most recently accepted value.
No reference here is required to the payload of the quorum of $\tf{send}$ messages.
\rulefont{PaxWrite?} still insists that such a quorum exists, as per the $\everyone\texi\tf{send}$ that appears in \rulefont{PaxWrite?}, so in this case the value $\mru{(\somewhere\tf{accept})}{v}$ is computable by the leader, and we could specify a more complex predicate that checks the set of $v_p$ payloads in the send messages of that quorum to select a most recently accepted value.
In the implementation, this is necessary because the leader $l$ is not omniscient.
But in the logic we have no reason to go to this trouble, because we can just read the information directly from the model using $\mru{(\somewhere\tf{accept})}{v}$.\footnote{For example: if we want to compute $100!$ in the real world, we have to build a computer to do this for us, which involves many practical details including the construction of a semiconductor industry, compiler design and so on.  In logic, however, we can just say that it is $\prod_{i=1}^{100}i$, because that is what $100!$ \emph{is}.}

So our axioms are just doing their job, which is to be a declarative abstraction of the implementation. 
But since we do not use \rulefont{PaxSend?}, this is evidence that we may treat the $\tf{send}$ step as an implementational detail having to do with how information moves around the network, and elide it.
\end{rmrk}

\begin{rmrk}[Why we did not use \rulefont{PaxPropose01}]
This is due to another particularity of how the Paxos algorithm works.
The leader chooses one value in the propose step (or no value, if the leader has crashed), but it only uses that value in the case that it finds a quorum of participants who have never previously accepted any value (this is the right-hand disjunct in \rulefont{PaxWrite?}).
Just in this case, it writes the value that it chose in the propose step.

So what really matters is that the leader \emph{only writes one value}.
If the leader had proposed two values, this would be fine --- provided that it only writes one of them.

Leaders in the Paxos algorithm do not bother with multiple values in the propose step, because this would be a waste of time; they know they will only ever need at most one.
Thus, \rulefont{PaxPropose01} is right to reflect a true fact of the algorithm that it does not propose any more than the value that it will need, but this is an algorithmic efficiency not a logical necessity. 
Likewise, the logic is right to point this out to us by \emph{not} requiring \rulefont{PaxPropose01} in the proofs: so long as at most one value gets \emph{written} --- which is handled by \rulefont{PaxWrite01} --- none of our correctness properties are affected.
\end{rmrk}

\begin{rmrk}
By examining the proofs further, some further possible simplifications are revealed:
\begin{enumerate*}
\item
We can combine rules %
\rulefont{PaxWrite01} and \rulefont{LdrUniq}, as per the discussion in Subsection~\ref{subsect.ldrexistuniq'}, to obtain an elegantly concise rule \rulefont{SPaxWrite01}.
\item
We used \rulefont{PaxPCorrect} in our proofs to obtain $\someOallP\correct{\tf P}$ only twice: 
\begin{itemize*}
\item
in Lemma~\ref{lemm.someone.writes.after.gslt.accept} for $\tf P=\tf{accept}$, and 
\item
in Lemma~\ref{lemm.leaderlive.quorumsend} for $\tf P=\tf{send}$.
\end{itemize*}
We have elided the \emph{send} step in $\Theta_\ThySPax$, so we only need to consider $\tf{accept}$.
Accordingly we take \rulefont{SPaxPCorrect} to be just $\someOallP\correct{\tf{accept}}$.
\end{enumerate*}
\end{rmrk}

\begin{rmrk}
In conclusion, we see that the proofs in Section~\ref{sect.paxos.correctness.properties} are not only correctness proofs, but can also be read as an active commentary on the axioms.
Our logic has done what logic does, and has separated what is logically necessary from what is implementationally required or implementationally convenient.
This shows us how to simplify the axiomatisation (if we wish) further away from the implementation, and closer towards what is logically essential: 
\end{rmrk}

\subsection{The definition}

\begin{figure}
$$
\begin{array}{l@{\quad}r@{\ }l}
\figunderline{Backward rules, for backward inference} 
\figskip
\rulefont{SPaxPropose?}&
& \tall v.\ \tf{propose}(v)\timpc\tf{leader}
\figskip
\rulefont{SPaxWrite?}&
& \tall v.\ \tf{write}(v) \timpc 
\begin{array}[t]{l}
\tf{leader}\tand\Quorum( 
\\
\quad\bigl((\mru{(\somewhere\tf{accept})}{v}) \tor
(\tf{propose}(v)\tand\tneg\modT\recent\texi\somewhere\tf{accept})\bigr)) 
\end{array}
\figskip
\rulefont{SPaxAccept?}&
& \tall v.\ \tf{accept}(v) \timpc \someOsomeP\tf{write}(v)
\figskip
\rulefont{SPaxDecide?}&
& \tall v.\ \tf{decide}(v) \timpc \someOallP\tf{accept}(v)
\figskip\figskip
\figunderline{Forward rules, for making progress} 
\figskip
\rulefont{SPaxWrite!}&
& \final(\tf{leader}\tnotor\texi\tf{write})
\figskip
\rulefont{SPaxAccept!}&
& \final(\texi\someOsomeP\tf{write}\tnotor \texi\tf{accept})
\figskip
\rulefont{SPaxDecide!}&
& \final(\texi\someOallP\tf{accept} \tnotor \texi\tf{decide})
\figskip
\figskip
\figunderline{Other rules, for correctness \& liveness}
\figskip
\rulefont{SLdrExist}&
&\modTF\tf{leader}\tand \modT\someOsomeP\tf{leader}
\figskip
\rulefont{SLdrCorrect}&
& \infinitely\someOsomeP(\tf{leader}\tand \correct{\tf{propose},\tf{write},\tf{decide}}) 
\figskip
\rulefont{SPaxPCorrect}&
& \someOallP\correct{\tf{accept}}
\figskip
\rulefont{SPaxWrite01}&
& \texiaffine\someOsomeP\tf{write}
\figskip
\rulefont{SPax2Twined}&
& \tall v.\ \someOallP\tf{accept}(v) \timpc \allOsomeP\tf{accept}(v)
\end{array}
$$
\caption{Simpler Declarative Paxos (Definition~\ref{defn.thyspax})}
\label{fig.logical.paxos'}
\label{fig.ThySPax}
\end{figure}

\begin{defn}
\label{defn.thyspax}
We define \deffont{Simpler Declarative Paxos} to be the logical theory 
$$
\ThySPax=(\Sigma_{\ThySPax},\Theta_{\ThySPax}),
$$ 
where:
\begin{enumerate*}
\item
the signature $\Sigma_{\ThySPax}$ is 
$$
\Sigma_{\ThySPax}=[\tf{leader}:0;\ \tf{propose},\tf{write},\tf{accept},\tf{decide}:1] ,
$$ 
as per the notation in Definition~\ref{defn.signature}(\ref{item.predicate.syntax.signature.list.form}), and
\item
the axioms $\Theta_{\ThySPax}$ are as written in Figure~\ref{fig.ThySPax}.
\end{enumerate*}
\end{defn}

The same validity properties hold of \ThySPax from Figure~\ref{fig.ThySPax} as hold for \ThyPax from Figure~\ref{fig.ThyPaxOne}, namely: Validity (Theorem~\ref{thrm.validity}), Agreement (Theorem~\ref{thrm.full.agreement}), and Termination and Integrity (Theorem~\ref{thrm.termination}):
\begin{thrm}
\label{thrm.simpler.dp}
Suppose $\Sigma$ is a signature that includes $\Sigma_\ThySPax$, and suppose $\mathcal M$ is a model.
Then we have:
\begin{enumerate*}
\item
$\ThySPax\ment\tall v.(\tf{decide}(v) \timpc \urecent\someOsomeP(\tf{leader}\tand\tf{propose}(v)))$.
\item\label{item.simpler.dp.2}
$\ThySPax\ment\tall v,v'.(\urecent\someOsomeP\tf{decide}(v) \tand \someOsomeP\tf{decide}(v')) \timpc v'\teq v$.
\item
$\ThySPax\ment\infinitely\allOallP\texi\tf{decide}$.
\end{enumerate*}
\end{thrm}
\begin{proof}
This is by construction: from the discussion in Subsection~\ref{subsect.reflection.on.correctness}, the theory \ThySPax consists of just those parts of \ThyPax that are actually required in the correctness proofs.
\end{proof}

Thus, Theorem~\ref{thrm.simpler.dp} just notes that Simpler Declarative Paxos extracts from the proofs in Section~\ref{sect.paxos.correctness.properties} the logical essence of Figure~\ref{fig.ThyPaxOne}, as applied to the standard correctness properties of Paxos.
There is no right or wrong about whether we prefer $\ThyPax$ or $\ThySPax$; it depends on what level of abstraction we want to work with at any particular moment.
But, implicit in the formal mathematics of Theorem~\ref{thrm.simpler.dp} is an intuitive observation that \ThySPax is \emph{minimal}, in the sense that we put into it only that which we know we need to make some specific correctness proofs work. 

\section{Conclusions and future work}
\label{sect.conclusions}

\subsection{Conclusions}

We have introduced \QLogicShort, a three-valued modal fixed-point logic for the declarative specification of distributed algorithms, and demonstrated its use by considering Paxos --- a canonical and nontrivial consensus protocol which tolerates benign failures of participants.

This case study shows how our approach can capture the essential logical structure of distributed protocols in a concise and principled way. 
High points include the use of a third truth-value to handle the case of incorrect participants without requiring explicit `for a correct participant' side-conditions (Remark~\ref{rmrk.motivate.validity}), the syntax and semantics of \QLogicShort itself as a core language (Figures~\ref{fig.predicate.syntax} and~\ref{fig.3.phi.f}), the Declarative Paxos and Simpler Declarative Paxos axiomatisations (Figures~\ref{fig.ThyPaxOne} and~\ref{fig.ThySPax}),
how we capture 2-twined and 3-twined quorum intersection in clean logical form (Propositions~\ref{prop.supertwined.qcbt} and~\ref{prop.intertwined.char}; see also Remark~\ref{rmrk.motivate.n-twined}), and an illustration of the (non-evidently) rich expressive power of our framework that we can even identify unique leaders (Subsection~\ref{subect.ldrunique}).

While we have not developed full axiomatisations of more complex consensus algorithms in this paper, preliminary investigations suggest that our framework extends naturally to these settings. 
This indicates that \QLogicShort has the potential to serve as a general logical foundation for specifying and reasoning about distributed consensus, paving the way for more systematic analysis, clearer communication of protocol designs, and ultimately, stronger guarantees of correctness.  

This paper is aimed at 
\begin{itemize*}
\item
logicians, who might be interested in exploring a logic with a new and unusual combination of features inspired by strong practical motivations, and 
\item
researchers in distributed systems, who might consider (and benefit from) adopting such a framework.
\end{itemize*}
\QLogicShort offers benefits that could make it a useful addition to the toolkit of anyone working in distributed systems:
\begin{enumerate}
\item 
Much of what we have presented reflects reasoning patterns that practitioners already employ informally.  
For instance, the third truth-value ($\tvB$) in \QLogicShort corresponds to familiar side conditions such as ``for every correct participant,'' and the backward axioms capture the reverse inferences that system designers routinely rely on.  
What is new here is making these implicit practices explicit, systematic, and mathematically precise.  

By crystallising these patterns into a single logical framework, \QLogicShort provides more than just a tidy formalisation: it offers a common language in which both logicians and distributed-systems researchers can communicate, verify, and extend consensus protocols.  
This not only clarifies existing algorithms but also creates a principled foundation for analysing new ones, reducing ambiguity and opening the door to mechanical verification.  
In this way, \QLogicShort turns tacit reasoning into a tool for systematic design, bridging the gap between informal intuition and formal correctness.  
\item 
Logic is a concise, powerful, and unambiguous language. 
It distils paragraphs of English prose into a single line of precise symbols, enabling clear and rigorous reasoning. 

While explanatory English is helpful for accessibility and context, formal logic excels at working through the details of complex algorithms. 
It ensures precision and eliminates the ambiguities that often arise in informal descriptions.
\item 
\QLogicShort operates at a new (and higher) level of abstraction, which complements existing logic-based techniques. 
TLA+~\cite{DBLP:journals/toplas/Lamport94} provides a state-based view of distributed algorithms, whereas \QLogicShort abstracts lower-level implementation details while retaining the essential behaviours and properties of the algorithms.
This abstraction does not obscure the logic of the system; on the contrary we have found that it clarifies it, making the underlying mechanisms and invariants more apparent. 
\item 
One strength of logic is its ability to abstract away what is irrelevant while maintaining precision about what is important. 
In \QLogicShort, constructs like strong and weak implications ($\timpc$ and $\tnotor$) exemplify this principle, capturing subtle distinctions in meaning with precise formal symbols. 
These constructs distil complex ideas into simple yet powerful tools that can be manipulated symbolically and reasoned about systematically. 

While mastering these constructs requires some initial effort, the benefits of precision and clarity justify the investment.
\item 
Learning \QLogicShort requires some effort, 
and axiomatisations in \QLogicShort are subtle, this is the opposite of gratuitous: it reflects the inherent nature of the distributed systems they are designed to model. 
By formalising these complexities, \QLogicShort captures their structure and reveals and addresses the challenges that are \emph{already present} in the underlying algorithms. 

This transparency makes it a powerful tool for understanding and reasoning about distributed protocols.\footnote{Certainly, the first author (who has a background in logic) has found this study invaluable for gaining a better understanding of consensus algorithms.  It seems likely that other readers might benefit similarly from using \QLogicShort.}
\item 
Distributed algorithms --- and consensus protocols in particular --- are foundational to the operation of blockchain systems and other critical infrastructures.  
Yet these algorithms are notoriously difficult to design and implement correctly, and ambiguities or informal specifications can leave room for vulnerabilities and failures.  

\QLogicShort addresses this challenge by bridging the gap between declarative specification of invariants and practical verification.  
It provides researchers with a rigorous framework for clarifying, reasoning about, and verifying the essential properties of distributed protocols.  
This contribution is especially relevant to blockchain systems, where unambiguous and trustworthy specifications are not merely desirable but indispensable.  
\end{enumerate}

In summary, \QLogicShort turns implicit reasoning patterns into an explicit and principled logic, opening the door to new kinds of specification, reasoning, analysis, and ultimately we hope to more reliable distributed systems.  
\QLogicShort unites the abstraction of declarative reasoning with the rigour of logical precision.
It makes reasoning patterns explicit, and offers a common language in which correctness can be articulated, communicated, and ultimately assured.

\subsection{Related and future work}

\subsubsection{Applications of Paxos and logic}

Paxos is a foundational distributed protocol, variants and elaborations of which are very much of practical interest to this day: a Paxos-based production distributed storage system called \emph{Megastore}~\cite{Baker:2011:Megastore} was used by Google; another distributed storage system \emph{PaxosStore}~\cite{Zheng:2017:PaxosStore} was used by WeChat and was partially inspired by Megastore (so far as we know it is still in use at time of writing); Google now uses \emph{Spanner}~\cite{Taylor:2025:SpannerImpact,GoogleCloud:2026:SpannerReplication}, which is also Paxos-based.
The \emph{Chubby} distributed lock system~\cite{Burrows:2006:Chubby} is Paxos-based and is still in use today~\cite{Chandra:2007:PaxosMadeLive,Levandoski:2024:BigLake}.

A modern example of a Paxos formalisation in TLA+ is documented in~\cite{ChandLiuStoller2016}; a relevant theory file is in~\cite{LamportMerzDoligez:PaxosTLA}.
Constructing TLA+ theories and proofs has historically not been cheap,\footnote{The rise of AI-generated code, which occurred after this paper was prepared and after the literature cited was published, might change things to some extent --- but at time of writing the precise impact and how this will play out specifically for distributed protocols, has not yet become clear.}
 so attempts have been made to automate proofs; a recent example is~\cite{GoelSakallah2021}.
Another approach is to use a logic for which a practical decision procedure exists, such as \emph{Effectively Propositional Logic} (EPR), which is decidable and has a small model property (which roughly says that if a formula is satisfiable in some model, then it is satisfiable in a model that is not too large relative to the size of the formula).
This means that if a protocol and its correctness properties can be expressed in EPR, then they stand a good chance of being automatically checkable in reasonable time --- this is not actually a mathematical guarantee, but empirically this tends to be a good place to start looking.
An example of a suitable Paxos formalisation is in~\cite{PadonLosaSagivShoham2017}.
This is not TLA+, but it is still very much a transition system.

What these approaches all have in common is that they describe \emph{transition systems based on message-passing}.
\QLogicShort, with its declarative style, follows a distinct philosophy: it encourages us to axiomatise protocols such that there is no transition system, no state-machine, no explicit message-passing, and no algorithmic time (i.e. no notion of events in an explicit sequence).\footnote{Round numbers are part of the logical structure of Paxos, and as such they are retained.

It might seem extraordinary that protocols could be so represented, given that the very essence of their implementation consists precisely of message-passing state machines embedded in algorithmic time!
However, this paper shows that it can be done, and on reflection that is no more remarkable than the fact that %
declarative paradigms such as the $\lambda$-calculus or logic clauses can be used to program, even though they may be compiled down to concrete state-transitions of a CPU.}
This structural difference in understanding of what a protocol \emph{is}, is distinctive to the approach proposed in this paper.

\subsubsection{Impossibility results via logic}

A rich literature exists studying \emph{classes of protocols} using logical axiomatisations, especially to prove impossibility or consistency guarantees from relatively weak assumptions.
Examples (non-exhaustive list) include~\cite{Taubenfeld:1991:ResilientConsensus} (an axiomatic impossibility proof for asynchronous resilient consensus), \cite{Gotsman:2017:GlobalOperationSequencing} (a consistency result for global sequencing of operations), \cite{Cirisci:2023:QuorumTree}, and~\cite{Albouy:2025:AMECOS} (another impossibility result).
This literature studies classes of protocols rather than specific protocols, so that the theories and models that arise tend to be more abstract than those arising when specifying specific protocols.\footnote{This is not to be read as a claim that no other authors have ever studied specific protocols.  They have and we cite examples elsewhere (non-exhaustive list: \cite{Cirisci:2023:QuorumTree,ChandLiuStoller2016,PadonLosaSagivShoham2017}).  That is just not the part of the literature that we are discussing in this paragraph.}
For example, Taubenfeld in~\cite{Taubenfeld:1991:ResilientConsensus} writes that a protocol is a set of runs and a run is a finite sequence of events, and the analysis proceeds at this level of abstraction.

In contrast, in this paper we have presented \QLogicShort as an environment within which to specify specific protocols (and, by extension, extensions of \QLogicShort if necessary to specify specifically extended protocols).

Another equally fundamental distinction is that the four papers cited above, and others in a similar spirit, are still about \emph{events happening in order}, as the quote from Taubenfeld illustrates: a protocol is a set of runs, and a run is \emph{a finite sequence of events}.
This is a general model, dispensing with all details of specific protocols, but it does retain one thing: precisely \emph{algorithmic time}; a finite sequence of events.  
This retained abstraction is fundamentally \emph{a different thing} from the logical time in which \QLogicShort exists.

\subsubsection{Logics of knowledge and belief}

A rich field of research exists studying knowledge and collaboration in multi-agent systems, especially using modal epistemic logics.
Foundational works in this field include~\cite{halpern1990knowledge,fagin2004reasoning}. 

Unfortunately confusion can now arise via the following chain of thought: distributed systems consist of multiple agents, \emph{therefore} distributed systems are a subfield of multi-agent systems, \emph{therefore} distributed protocol design is a subfield of multi-agent systems, \emph{therefore} a modal logic for distributed protocols must be some flavour of (or admit a presentation in terms of) modal epistemic logic.

Even the reader who recognises this reasoning as somewhat reductive might still find it colouring ideas about what \QLogicShort should be, because (like epistemic logics) \QLogicShort is a modal logic for specifying and reasoning about distributed protocols.
But this is no more valid for \QLogicShort than it is for TLA+ --- this is also a modal logic, and also about specifying and reasoning about distributed protocols, but the TLA+ modalities have nothing to do with belief.

The difference is that participants in a specific protocol need not themselves care about coordination, knowledge, or belief, even as they follow the protocol.
This is touched on explicitly in Section~9 of~\cite{halpern1990knowledge}, \emph{`A paradox?'}, which essentially notes that a system can have useful notions of coordination that are weaker than a full simultaneous coordination requiring full common knowledge.
For example: do a group of ants moving a leaf know they are part of an ant colony, and does a group of neurons identifying an edge in a retina know they are part of a brain (and, perhaps, looking at a mathematical symbol)?
No. 
Their jobs are to be groups of ants or neurons.
And, does an entire ant colony or an entire brain ever even get together for a global council and synchronisation step before every action by its parts?
No, of course not.

From the point of view of Paxos, coordination is an \emph{emergent property} of participants in that system following the rules; coordination is not an explicit goal of individual participants, and parts of the system may progress at different rates.
The analysis of knowledge in distributed systems is valuable for understanding the design space of emergent properties arising from classes of protocols, %
but in this paper we use \QLogicShort as a language for positively specifying and reasoning about the behaviour of individual participants following specific protocols, which is a different goal.\footnote{There would be no barrier to \emph{adding} knowledge modalities to \QLogicShort.  That would be an extension of the logic, but it would serve no purpose here, because Paxos does not require it.}

\subsection{Coalition Logic (CL) and Alternating-time Temporal Logic (ATL)}

Coalition logic includes a modality $[C]\phi$ indicating that a coalition $C$ can act jointly to ensure that $\phi$ holds in the next state (no matter what other participants may do)~\cite{pauly2002modal}.
\QLogicShort and CL share an interest in quorums/coalitions and multi-participant computation, but as logics they are quite distinct:
\begin{enumerate*}
\item
\QLogicShort abstracts away algorithmic time and so by design lacks any explicit notion of `next state' in the sense intended by CL (see Remark~\ref{rmrk.elisions}(\ref{item.elide.algorithmic.time})).
\item
A model in \QLogicShort represents a \emph{single run} of a protocol. 
A model in CL is a game representing all possible moves, and thus containing all possible protocol runs --- just follow the `next step' accessibility relation.
In CL, the model is typically assumed as an abstract structure, not axiomatised as a specific structure representing a particular protocol (though in principle it might be). 
\item
The \QLogicShort modalities are existential and universal `for some quorum' and `for all quorums'.
It is impossible to explicitly name some particular quorum $C$ as in CL; conversely, CL lacks explicit modalities for quantifying over coalitions.
\item
\QLogicShort is optimised to be a natural logic for specifying valid runs of some specific distributed protocol. 

CL is optimised to be a modal logic for discussing possible strategies of coalitions within some Kripke-frame-style abstract space of possible runs (Pauly calls these \emph{game frames}).

We are not aware of any specification of Paxos in Coalition Logic, nor in ATL (mentioned below).
This might be possible in principle --- but even if this were done, Coalition Logic's structural commitment to algorithmic time and the lack of other \QLogicShort structure (such as explicit modalities quantifying over quorums) means that such an axiomatisation would be a different structure that does not relate to the core concern in \QLogicShort of declarative axiomatisation of distributed protocols, and especially fault-tolerant broadcast and agreement protocols and their generalisations.\footnote{This focus extends to the CL-like literature.  We are aware of no explicit axiomatisations of any fault-tolerant broadcast or agreement protocols.
Where such things are even mentioned, as in a mention in item~2 of Remark~2 of~\cite{Goranko2001}, it is only as a source of frames over which to reason; no explicit axiomatisation is considered.  This really is a distinct literature using distinct tools for distinct purposes.}
\end{enumerate*}
\QLogicShort and CL are overlapping in the sense that they are both modal logics for reasoning about coalitions in distributed systems, but within that design space the goals and technical details are very different.

ATL has broadly similar goals to CL, but has a somewhat richer language of modalities, e.g. $\langle\langle C\rangle\rangle \phi\mathsf{U}\psi$ means that coalition $C$ can ensure that $\psi$ holds eventually and $\phi$ holds up to that point~\cite{AHK2002}.
A comparison of ATL and CL is in~\cite{GorankoJamroga2004}.
For the purposes of this paper, ATL compares with \QLogicShort just as CL does.\footnote{Unfortunately, an early version of this paper called the logic \emph{in this paper} `Coalition Logic'.  
This was an error by the first author, who (with apologies to Pauly) had missed the prior art.
The reader can disregard and not mention the previous naming of \QLogicShort as irrelevant and likely to confuse --- we only mention it here to address any confusion generated.
}

\subsection{Future work}

In this paper, our focus has been to lay the foundational theory and framework of \QLogicShort and to demonstrate its application to one of the most well-known consensus protocols, Paxos~\cite{DBLP:journals/tocs/Lamport98}.
This is a natural starting point: Paxos is a cornerstone of distributed systems and remains one of the most widely studied and deployed consensus protocols.  

At the same time, Paxos is relatively straightforward by modern standards.  
In particular, Paxos addresses only crash failures, and is therefore suited to environments that assume benign faults rather than adversarial behaviour.  

A natural direction for future work involves extending our framework to handle more complex and realistic fault models, such as Byzantine faults, where participants can deviate arbitrarily from the algorithm.
Consensus protocols that address Byzantine faults, such as Practical Byzantine Fault Tolerance (PBFT)~\cite{DBLP:conf/osdi/CastroL99} and HotStuff~\cite{DBLP:conf/podc/YinMRGA19}, provide robust mechanisms for fault tolerance and serve as the foundation for many modern blockchain systems.
Tendermint (\url{https://tendermint.com})~\cite{DBLP:journals/corr/abs-1807-04938} is another BFT consensus protocol. 
Applying \QLogicShort to these protocols would show how it handles the complexity introduced by arbitrary faults, which is necessary to model and verify more sophisticated algorithms.
On precisely this topic, in a shorter sister paper~\cite{gabbay:decdaa} to this paper the first author gives a declarative study of Bracha Broadcast and Crusader Agreement.
The logic used there is deliberately smaller and simpler (as the protocols are smaller and simpler than Paxos) but in one respect that paper does substantively extend this one, because its notion of failure (again, consistent with the protocols) is full Byzantine failure.
The three-valued logic framework handles this extremely well, with absolutely no complaint or friction: it is in fact quite remarkable how flexible this framework turns out to be.
Fundamentally, this just serves to illustrate the \emph{unreasonable effectiveness of logic} (apologies to Wigner and then Hamming).
We would urge the reader who may be used to a more imperative style of specification of distributed protocols to not assume that because `imperative' is how they may have been educated, `imperative' is the one and only viable way to go.
It is not.

Recent research has explored DAG-based consensus protocols, which deviate from the traditional linear blockchain structure. 
In these protocols, transactions are organised into a directed acyclic graph (DAG), enabling greater parallelism and scalability. 
Protocols such as DAG-Rider~\cite{DBLP:conf/podc/KeidarKNS21}, Narwhal (a DAG-based dissemination protocol) and Tusk (an asynchronous consensus protocol built on top of it)~\cite{DBLP:conf/eurosys/DanezisKSS22}, and Bullshark~\cite{DBLP:conf/ccs/SpiegelmanGSK22} represent cutting-edge advancements in this domain. 
Extending \QLogicShort to analyse DAG-based protocols would be another advance, as these systems often involve more intricate interactions between components and require reasoning about concurrent and asynchronous behaviour at a much higher level of complexity.

Another promising avenue for future research involves leveraging the modularity and expressiveness of \QLogicShort to study hybrid protocols~\cite{DBLP:conf/sp/NeuTT21} that combine elements of classical consensus with newer techniques, such as Gasper~\cite{DBLP:journals/corr/abs-2003-03052} and others~\cite{DBLP:conf/esorics/DAmatoZ23,DBLP:journals/corr/abs-2411-00558}.

It would be helpful in future work to develop automated tools for \QLogicShort: the tools developed for TLA+~\cite{DBLP:books/aw/Lamport2002} could show the way.
TLA+ is a formal specification language and framework designed to model, specify and verify properties of complex systems, including distributed algorithms. 
Its TLC model checker~\cite{DBLP:conf/charme/YuML99} provides exhaustive state exploration to verify properties such as safety and liveness, and the PlusCal translator simplifies specification by offering a high-level pseudocode-like interface.
Additionally, the TLA+ Toolbox~\cite{DBLP:journals/corr/abs-1912-10633} integrates these capabilities into a development environment, supporting workflows that are valuable in both research and industrial applications.

Adapting similar tools to \QLogicShort could enable automated verification of correctness properties for consensus protocols, including invariants and liveness guarantees. 
For example: a higher-level language akin to PlusCal could simplify the specification of protocols using the axiomatic foundations of \QLogicShort.
Furthermore, leveraging techniques such as the symbolic model checking enhancements presented by Otoni et al.~\cite{DBLP:conf/tacas/OtoniKKES23} or integrating modular proof systems like TLAPS~\cite{DBLP:conf/fm/CousineauDLMRV12} could further enhance the usability and robustness of the framework.

Such tools would streamline the application of \QLogicShort to real-world systems, and also make it accessible to a broader audience of researchers and practitioners. 
This direction represents an exciting convergence of theoretical advancements and practical tool development, providing a robust framework for analysing both classical and emerging consensus protocols.
As mentioned in the Introduction, we have made a start in this direction by creating a proof-of-concept model-checker~\cite{Gabbay2026-3muQL-explorer}, which the reader can load and play with online now by going to \url{https://bellissimogiorno.github.io/3muQL-model-checker/}. 

Adopting \QLogicShort requires some investment in time and effort, but the benefits in clarity, correctness, and reliability may be worthwhile.
Consensus algorithms are known to be particularly challenging, so our relatively high-level declarative approach could help to manage (and debug) the complexity of these algorithms.

We hope this work will inspire readers %
to explore and adopt this new perspective, contributing to more robust and trustworthy distributed protocols in the future.

\newcommand{\etalchar}[1]{$^{#1}$}
\providecommand{\bysame}{\leavevmode\hbox to3em{\hrulefill}\thinspace}
\providecommand{\MR}{\relax\ifhmode\unskip\space\fi MR }
\providecommand{\MRhref}[2]{%
  \href{http://www.ams.org/mathscinet-getitem?mr=#1}{#2}
}
\providecommand{\href}[2]{#2}

\end{document}